\ifx\synctex\undefined\else\synctex=1\fi
\documentclass[a4paper]{article}

\usepackage[utf8]{inputenc}
\usepackage[T1]{fontenc}

\usepackage{microtype}

\usepackage{amsxtra} 
\usepackage{amssymb}
\usepackage{amsthm}

\usepackage{latexsym}
\usepackage{stmaryrd}
\usepackage{mathabx}
\usepackage{mathtools}

\usepackage{galois}

\usepackage[pointedenum]{paralist}

\usepackage[tworuled,vlined,english]{algorithm2e}

\usepackage{url}

\usepackage{wrapfig}
\usepackage{tikz}
\usetikzlibrary{calc}

\clearpage{}%
\newtheorem{example}{Example}[section]
\newtheorem{lemma}{Lemma}[section]
\newtheorem{proposition}{Proposition}
\newtheorem{theorem}{Theorem}
\newtheorem{corollary}{Corollary}

\newtheorem{definition}{Definition}

\newcommand{\topleft}[1] {{~}^{\scriptscriptstyle{\blacksquare}} \! M}
\newcommand{\botleft}[1] {{~}_{\scriptscriptstyle{\blacksquare}} \! M}
\newcommand{\topright}[1] {M^{\scriptscriptstyle{\blacksquare}}}
\newcommand{\botright}[1] {M_{\scriptscriptstyle{\blacksquare}}}

\renewcommand{\vec}[1]{{\bf {#1}}}
\newcommand{\true}{\mbox{\bf true}}
\newcommand{\false}{\mbox{\bf false}}

\newcommand{\card}[1]{\mbox{card}({#1})}

\newcommand{\arrow}[2]{\xrightarrow[{\scriptstyle #2}]{{\scriptstyle #1}}}
\newcommand{\nat}{{\bf \mathbb{N}}}
\newcommand{\zed}{{\bf \mathbb{Z}}}

\def\vr{\kern-\arraycolsep & \kern-\arraycolsep}
\def\VR{\kern-\arraycolsep\strut\vrule}

\def\age#1{\left[#1\right]}
\def\set#1{{\left\{ #1 \right\}}}

\def\tuple#1{{\langle #1 \rangle}}
\def\nats{{\mathbb{N}}}
\def\zed{\mathbb{Z}}
\def\card#1{{|\!|{#1}|\!|}}
\def\len#1{{|{#1}|}}
\def\prod{\Delta}
\def\pat{{\mathbf{b}}}
\def\patt{{\widetilde{\mathbf{b}}}}

\def\df#1{\mathbf{df}(#1)}
\def\fin#1{\mathop{\mathcal{F}}(#1)}
\def\nf#1{\mathop{n\!\mathcal{F}}(#1)}

\def\Vars{\ensuremath{\Xi}}

\makeatletter
\newsavebox{\@brx}
\newcommand{\llangle}[1][]{\savebox{\@brx}{\(\m@th{#1\langle}\)}%
  \mathopen{\copy\@brx\kern-0.5\wd\@brx\usebox{\@brx}}}
\newcommand{\rrangle}[1][]{\savebox{\@brx}{\(\m@th{#1\rangle}\)}%
  \mathclose{\copy\@brx\kern-0.5\wd\@brx\usebox{\@brx}}}
\makeatother

\def\calls{\langle\!\!\langle}
\def\rets{\rangle\!\!\rangle}

\def\loc{\mbox{loc}}

\DeclareFontFamily{U}{mathx}{\hyphenchar\font45}
\DeclareFontShape{U}{mathx}{m}{n}{
      <5> <6> <7> <8> <9> <10>
      <10.95> <12> <14.4> <17.28> <20.74> <24.88>
      mathx10
      }{}
\DeclareSymbolFont{mathx}{U}{mathx}{m}{n}
\DeclareFontSubstitution{U}{mathx}{m}{n}
\DeclareMathAccent{\widecheck}{0}{mathx}{"71}
\DeclareMathAccent{\wideparen}{0}{mathx}{"75}

\def\bdwords{\Upsilon}

\renewcommand{\vec}[1]{{\mathbf {#1}}}

\newcommand{\sem}[1]{\llbracket #1 \rrbracket}

\def\proj{\mathbin{\downarrow}}

\clearpage{}%

\usepackage{authblk}

\title{Underapproximation of Procedure Summaries for Integer Programs}
\author[1]{Pierre Ganty}
\author[2]{Radu Iosif}
\author[2,3]{Filip Kone\v{c}n\'{y}}
\affil[1]{\textsc{Imdea} Software Institute, Madrid, Spain}
\affil[2]{\textsc{Verimag/CNRS}, Grenoble, France}
\affil[3]{\'Ecole Polytechnique F\'ed\'erale de Lausanne (EPFL), Switzerland}
\date{}              
\begin{document}
%%%%%%%%%%%%%%%%%%%%%%%%%%%%%%%%%%%%%%%%%%%%%%%%%%%%%%%%%%%%%%%%%%%%%%%%%%%%%%%

\maketitle

%
% start input /Users/pierreganty/counter-recursive/abstract.tex
\begin{abstract}
We show how to underapproximate the procedure summaries of recursive 
programs over the integers using off-the-shelf analyzers for non-recursive programs. The novelty of
our approach is that the non-recursive program we compute may capture unboundedly many behaviors
of the original recursive program for which stack usage cannot be bounded.
Moreover, we identify a class of recursive programs on which our method
terminates and returns the precise summary relations without
underapproximation. Doing so, we generalize a similar result for non-recursive
programs to the recursive case.  Finally, we present experimental results of an
implementation of our method applied on a number of examples.
\end{abstract}
 % end input /Users/pierreganty/counter-recursive/abstract.tex
 %
% start input /Users/pierreganty/counter-recursive/introduction.tex
% vim:ts=2:sw=2
%
%        File: introduction.tex
%     Created: 
% Last Change: $Date: 2016-03-15 16:49:17 +0100 (Tue, 15 Mar 2016) $
%
% Written to be compiled by pdflatex
% usual typos to check:
% twice the same word: \(\<\w*\>\)\_s*\1\>
% too much space in math environment: \\\\\_s*\\end
% useless space at the end of a line: %s/\s*$//
% no concluded proof
% non matching parenthesis

\section{Introduction}

Formal approaches to reasoning about behaviors of programs usually
fall into one of the following two categories: {\em certification}
approaches, that provide proofs of correctness, and {\em bug-finding}
approaches, that explore increasingly larger sets of traces in order
to find possible errors. While the methods in the first category are
used typically in the development of safety-critical software whose
failures may incur dramatic losses in terms of human lives (airplanes,
space missions, or nuclear power plants), the methods in the second
category have a broad application in industry, outside of the
safety-critical market niche. Another difference between the two
categories is methodological: certification approaches are based on
{\em over-approximations} of the set of behaviors (if the
over-approximation is free of errors, the original system is correct),
while bug-finding needs systematic {\em under-approximation}
techniques (if there are errors, the method will eventually discover
all of them). Finally, over-approximation methods are guaranteed to
terminate, but the answer might be inconclusive (spurious errors are
introduced due to the abstraction), whereas under-approximation
methods provide precise results (all reported errors are real), but
with no guarantee for termination.

{\em Procedure summaries} are relations between the input and return
values of a procedure, resulting from its terminating executions.
Computing summaries is important, as they are a key enabler for the
development of modular verification techniques for inter-procedural
programs, such as checking safety, termination or equivalence
properties. Summary computation is, however, challenging in the
presence of {\em recursive procedures} with integer parameters, return
values, and local variables. While many analysis tools exist for
non-recursive programs, only a few ones address the problem of
recursion (e.g. \textsc{InterProc} \cite{interproc}).

In this paper, we propose a novel technique to generate arbitrarily
precise {\em underapproximations} of summary relations. Our technique
is based on the following idea. The control flow of procedural
programs is captured precisely by the language of a context-free
grammar. A \(k\)-index underapproximation of this language (where
\(k\geq 1\)) is obtained by filtering out those derivations of the
grammar that exceed a budget, called \emph{index}, on the number (at
most \(k\)) of occurrences of nonterminals occurring at each
derivation step. As expected, the higher the index, the more complete
the coverage of the underapproximation.  From there we define the
\(k\)-index summary relations of a program by considering the
\(k\)-index underapproximation of its control flow. Our method then
reduces the computation of \(k\)-index summary relations for a
recursive program to the computation of summary relations for a
non-recursive program, which is, in general, easier to compute because
of the absence of recursion. The reduction was inspired by a
decidability proof \cite{AG11} in the context of Petri nets.

The contributions of this paper are threefold. First, we show that,
for a given index, recursive programs can be analyzed using
off-the-shelf analyzers designed for non-recursive programs. Second,
we identify a class of recursive programs, with possibly unbounded
stack usage, on which our technique is complete, i.e.\ it terminates
and returns the precise result. Third, we present experimental
results of an implementation of our method applied on a number of
examples.

\noindent{\bf Motivating Example} To properly introduce the reader to
our result, we describe our source-to-source program transformation
through an illustrative example. Consider the recursive program
$\mathcal{P}=\set{P}$, consisting of a single recursive procedure $P$,
given in Fig. \ref{fig:program} (a), whose control flow graph is given
in Fig. \ref{fig:program} (b). The nodes of this graph represent
control locations in the program, with a designated initial location
$Q^{init}_1$ and a final location $\varepsilon$. The edges are labeled
with relations denoting the program semantics, where primed variables
$x'$ and $z'$ denote the values at the next step. For instance, the
edge $t_2: Q_2 \arrow{z'=P(x-1) \wedge x'=x}{} Q_3$ corresponds to the
recursive call on line \(3\) in the program---the edge labels of the
control flow graph explicitly mention the copies of variables not
changed by the program action corresponding to the edge, e.g.\ $x'=x$.

In this paper, we model programs using visibly pushdown grammars (VPG)
\cite{AM09}. The VPG for \(P\) is given in Fig. \ref{fig:program} (c).
The role of the grammar is to define the set of
\emph{interprocedurally valid} paths in the control-flow graph of the
program \(P\). Every edge in the control-flow graph matches one or two
symbols from the finite alphabet $\set{\tau_1, \calls\tau_2,
  \tau_2\rets, \tau_3, \tau_4}$, where $\calls\tau_2$ and
$\tau_2\rets$ denote the call and return, respectively. Each edge in
the graph translates to a production rule in the grammar, labeled
$p^b_1,p^c_2,p^a_3$ and $p^a_4$---the superscript $a$, $b$ and $c$
distinguishes rules with $0$, $1$ and $2$ nonterminals on the
right-hand side, respectively. For instance, the call edge $t_2$
becomes the rule $Q_2 \rightarrow \calls\tau_2 Q^{init}_1 \tau_2\rets
Q_3$. The language of the grammar of Fig.~\ref{fig:program} (c) (with
axiom $Q^{init}_1$) is the set $\set{\left(\tau_1\calls\tau_2\right)^n
  \tau_4 \left(\tau_2\rets\tau_3\right)^n \mid n \in \nat}$ of
interprocedurally valid paths, where each call symbol $\calls\tau_2$
is matched by a return symbol $\tau_2\rets$, and the matching relation
is well-parenthesized.

The outcome of the program transformation applied to $P$ is the
non-recursive program $\mathcal{Q} = \set{\mathit{query}^i}_{i=0}^K$,
depicted in Fig. \ref{fig:program} (d), where $K$ is a parameter of
our analysis. The main idea is that the executions of the procedure
$\mathit{query}^k$, ending with an empty stack, correspond to the
derivations of the VPG in Fig. \ref{fig:program} (c), of index at most
$k$---since there is no derivation of index $0$, the set of executions
of $\mathit{query}^0$ will be empty. The body of a procedure
$\mathit{query}^k$ consists of a main loop, starting at the control
label $begin\_loop$ in Fig. \ref{fig:program} (d). Each branch inside
the main loop corresponds to the simulation of one of the production
rules of the grammar in Fig. \ref{fig:program} (c) and starts with a
control label which is the name of that rule
($p_1^b,p_2^c,p_3^a,p_4^a$). Next, we explain the relations labeling
the control edges of $\mathit{query}^k$. For each production rule $p$
in the grammar we have a relation $\rho_p(x_I,z_I,x_O,z_O)$, where
subscript \(I\) and \(O\) denote the input and output copies of the
program variables of $P$, respectively. In addition, we consider
auxiliary copies $x_J,z_J$, $x_K,z_K$ and $x_L,z_L$, defined in a
similar way. For instance, the auxiliary variables store intermediate
results of the computation of \(p^c_2\) as follows: \( [x_I,z_I]\;
\calls\tau_2 \;[x_J,z_J]\; Q^{init}_1 \;[x_K,z_K]\; \tau_2\rets
\;[x_L,z_L]\; Q_3 \;[x_O,z_O]\). The transition $p_2^c \rightarrow
in\_order/out\_of\_order$ can be understood by noticing that
$\calls\tau_2$ gives rise to the constraint $x_J=x_I-1$, $\tau_2\rets$
to $z_L=z_K$ and $x_I=x_L$ corresponds to the frame condition $x'=x$.

The peculiarity of the resulting program is that a function call is
modeled in two possible ways: \begin{inparaenum}[(i)]
\item \emph{in-order} execution of the function body, followed by the
  continuation of the call, and
\item \emph{out-of-order} execution of the continuation, followed by
  the execution of the function body.
\end{inparaenum}
The two cases correspond to $k$-index derivations of the VPG in Fig
\ref{fig:program} (c) of the form $u Q_1^{init} v Q_3 w \Rightarrow^*
u v_1 v Q_3 w \Rightarrow^* u v_1 v v_2 w$ and $u Q_1^{init} v v_2 w
\Rightarrow^* u v_1 v v_2 w \Rightarrow^* u v_1 v v_2 w$,
respectively, where $Q_1^{init} \Rightarrow^* v_1$ and $Q_3
\Rightarrow^* v_2$ are derivations of the VPG. In the first case, the
control path simulating the derivation in $\mathit{query}^k$ follows
the left branch $in\_order/out\_of\_order \rightarrow begin\_loop$,
whereas the second case is simulated by the right branch. 

Since the only call of $query^k$ is to $query^{k-1}$, on the edges
$in\_order/out\_of\_order \rightarrow begin\_loop$, the whole program
is a non-recursive under-approximation of the semantics of the
original program $P$, amenable to analysis using intra-procedural
program analysis tools. Indeed, the computation of the pre-condition
relation of the program
$\mathcal{Q}=\{\mathit{query}^2,\mathit{query}^1$, $\mathit{query}^0\}$
with the \textsc{Flata} tool \cite{HKGIKR12} yields the formula $z_O =
2\cdot x_I$, which matches the summary $z'=2\cdot x$ of the program
$P$.

In other words, the analysis of the under-approximation of
$\mathcal{P}$ of index at most $2$ suffices to infer the complete
summary of the program (the analysis for values $K>2$ will necessarily
yield the same result, since the under-approximation method is
monotonic in $K$). This fact matches the completeness result of Section
\ref{sec:completness}, stating that the analysis needs to be carried
up to a certain bound (linear in the size of the program's VPG)
whenever the language of the VPG is included in the language of the
regular expression \(w_1^* \ldots w_n^*\), for some non-empty words
\(w_1,\ldots,w_n\). In our case, the completeness result applies due to
\(\set{\left(\tau_1\calls\tau_2\right)^n \tau_4
  \left(\tau_2\rets\tau_3\right)^n \mid n \in \nat} \subseteq
\left(\tau_1\calls\tau_2\right)^* \tau_4^*
\left(\tau_2\rets\tau_3\right)^*\).

\begin{figure}[h]
	\centering
	\includegraphics{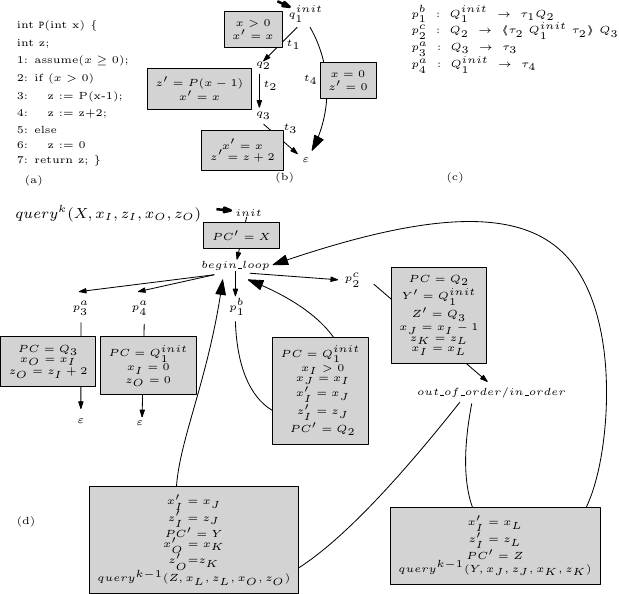}
 \caption{A recursive program returning the parameter value multiplied
  by two (a), its corresponding control flow graph (b) and visibly
  pushdown grammar (c), and the non-recursive program \(query^k(X,
  x_I, z_I, x_O, z_O)\) resulting from our index-bounded
  under-approximation (d). }
\label{fig:program}
\end{figure}

\noindent {\bf Related Work} The problem of analyzing recursive
programs handling integers (in general, unbounded data domains) has
gained significant interest with the seminal work of Sharir and Pnueli
\cite{sp81}. They proposed two orthogonal approaches for
interprocedural dataflow analysis. The first one keeps precise values
({\em call strings}) up to a limited depth of the recursion stack,
which bounds the number of executions. In contrast to the methods
based on the call strings approach, our method can also analyse
precisely certain programs for which the stack is unbounded, allowing
for unbounded number of executions to be represented at once.

The second approach of Sharir and Pnueli \cite{sp81} is based on
computing the least fixed point of a system of recursive dataflow
equations (the {\em functional approach}). This approach to
interprocedural analysis is based on computing an increasing {\em
  Kleene sequence} of abstract summaries. It is to be noticed that
abstraction is key to ensuring termination of the Kleene sequence, the
result being an over-approximation of the precise summary. Recently
\cite{EKL10:JACM}, a {\em Newton sequence} defined over the language
semiring was shown to converge at least as fast as the Kleene sequence
over the same semiring. An iterate of a Newton sequence is the set of
control paths in the program that correspond to words produced by a
grammar, with bounded number of nonterminals at each step in the
derivation. By increasing this bound, we obtain an increasing sequence
of languages that converges to the language of behavior of the
program. Our contribution can be thus seen as a technique to compute
the iterates of the Newton sequence for programs with integer
parameters, return values, and local variables, the result being, at
each step, an under-approximation of the precise summary.

The complexity of the functional approach was shown to be polynomial
in the size of the (finite) abstract domain, in the work of Reps,
Horwitz and Sagiv \cite{RHS95}. This result is achieved by computing
summary information, in order to reuse previously computed information
during the analysis. Following up on this line of work, most existing
abstract analyzers, such as \textsc{InterProc} \cite{interproc}, also
use relational domains to compute over-approximations of function
summaries -- typically widening operators are used to ensure
termination of fixed point computations. The main difference of our
method with respect to static analyses is the use of
under-approximation instead of over-approximation. If the final
purpose of the analysis is program verification, our method will not
return false positives. Moreover, the coverage can be increased by
increasing the bound on the derivation index.

Previous works have applied model checking based on abstraction
refinement to recursive programs. One such method, known as
\emph{nested interpolants} represents programs as nested word automata
\cite{AM09}, which have the same expressive power as the visibly
pushdown grammars used in our paper. Also based on interpolation is
the \textsc{Whale} algorithm \cite{AGC12}, which combines partial
exploration of the execution paths (underapproximation) with the
overapproximation provided by a predicate-based abstract post
operator, in order to compute summaries that are sufficient to prove a
given safety property. Another technique, similar to \textsc{Whale},
although not handling recursion, is the \textsc{Smash} algorithm
\cite{GNRT10} which combines may- and must-summaries for compositional
verification of safety properties. These approaches are, however,
different in spirit from ours, as their goal is proving given safety
properties of programs, as opposed to computing the summaries of
procedures independently of their calling context, which is our
case. We argue that summary computation can be applied beyond safety
checking, e.g., to prove termination
\cite{cook-podelski-rybalchenko-fmsd09}, or program equivalence.

The technique of under-approximation is typically used for bug
discovery, rather than certification of correctness. For instance, bug
detection based on under-approximation has been developed for
non-recursive C programs with arrays \cite{KroeningLW13}. Our approach
in orthogonal, as we consider more complex control structures
(possibly recursive procedure calls) but simpler data domains (scalar
values such as integers).

\noindent {\bf Paper organization.} After introducing the basic definition in
Section~\ref{sec:prelim}, we present, in Section~\ref{sec:int:programs}, our
model for programs, a semantics based on nested words and another one,
equivalent, based on derivations of the underlying grammar. Then, in
Section~\ref{sec:bounded-query}, we present our main contribution which is a
program transformation underapproximating the semantics of the input program.
In Section~\ref{sec:completness}, we define a class of programs for which the
underapproximation is complete.  Finally, after reporting on experiments in
Section~\ref{sec:experiments} we conclude in Section~\ref{sec:conclusions}.
 % end input /Users/pierreganty/counter-recursive/introduction.tex
 %
% start input /Users/pierreganty/counter-recursive/preliminaries.tex
% vim:ts=2:sw=2
%
%        File: preliminaries.tex
%     Created: Fri July 27
% Last Change: $Date: 2015-07-11 00:10:32 +0200 (Sat, 11 Jul 2015) $
%
% Written to be compiled by pdflatex
% usual typos to check:
% twice the same word: \(\<\w*\>\)\_s*\1\>
% too much space in math environment: \\\\\_s*\\end
% useless space at the end of a line: %s/\s*$//
% no concluded proof
% non matching parenthesis

\section{Preliminaries}\label{sec:prelim}

\subsection{Grammars} 
Let $\Sigma$ be an \emph{alphabet}, that is a finite non-empty set of
symbols. We denote by $\Sigma^*$ the set of finite words over $\Sigma$
including $\varepsilon$, the empty word.  Given a word
\(w\in\Sigma^*\), let \(|w|\) denote its length and let \( (w)_i\),
with \(1\leq i\leq |w|\), be the \(i\)-th symbol of \(w\).  By \(
(w)_{i \ldots j} \), with \(1\leq i \leq j \leq |w|\), we denote the
subword \( (w)_i \ldots (w)_j \) of \(w\). For a word $w \in \Sigma^*$ and
\(\Sigma' \subseteq \Sigma\), we denote by $w \proj_{\Sigma'}$ the result of
erasing all symbols of \(w\) not in \(\Sigma'\).

A \emph{context-free grammar} (or simply grammar) is a tuple
$G=\tuple{\Vars,\Sigma,\prod}$, where $\Vars$ is a finite nonempty set
of \emph{nonterminals}, $\Sigma$ is an alphabet, such that $\Vars \cap
\Sigma = \emptyset$, and $\prod \subseteq \Vars\times (\Sigma \cup
\Vars)^*$ is a finite set of \emph{productions}.  
%
%% For a production $(X,w) \in \prod$, often conveniently noted \(X
%% \rightarrow w\), we define its {\em size} as $\len{(X,w)} =
%% \len{w}+1$, and $\len{G} = \sum_{p\in\prod} \len{p}$ is the size of
%% the grammar $G$. 
%
A production $(X,w) \in \prod$ is often conveniently noted \(X \rightarrow w\).
Also define \(\mathit{head(X\rightarrow w)=X}\) and \(\mathit{tail}(X\rightarrow w)=w\).
Given two strings $u,v \in (\Sigma \cup \Vars)^*$, a production \(
(X,w)\in\prod\) and \(1\leq j\leq |u|\), we define a \emph{step} $u
\xRightarrow{(X,w)/j}_G v$ if, and only if, \( (u)_j = X\) and \( v =
(u)_1 \cdots (u)_{j-1} \cdot w\cdot (u)_{j+1}\cdots (u)_{|u|}\).  We
omit \( (X,w)\) or \(j\) above the arrow when it
is not important. In this notation and others, when \(G\) is clear from
the context, we omit it.  \emph{Step sequences} (including the empty
sequence) are defined using the reflexive transitive closure of the
step relation $\xRightarrow{}_G$, denoted $\xRightarrow{}^*_G$. For
instance, \(X \xRightarrow{}^*_G w\) means there exists a sequence of
steps that produces the word \(w \in (\Sigma \cup \Vars)^*\), starting
from \(X\). We call any \emph{step sequence} \( v \xRightarrow{}^*_G
w\) \emph{a derivation} whenever \(v\in\Vars\) and \(w\in\Sigma^*\).
The language produced by \(G\), starting with a
nonterminal \(X\) is the set \(L_X(G) = \set{w\in\Sigma^* \mid X
\xRightarrow{}^*_G w}\).

By defining a \emph{control word} to be a sequence of productions $\gamma \in \prod^*$, we can
annotate step sequences as expected: \(\varepsilon\in\prod^*\) is the
control word for empty step sequences, and given a control word
\(\gamma\) of length \(n\) we write \( u \xRightarrow{\gamma}_G v\)
whenever there exists \(w_0,\ldots, w_{n} \in (\Vars\cup\Sigma)^*\)
such that
\[u = w_{0} \xRightarrow{(\gamma)_1}_G w_1
\xRightarrow{(\gamma)_2}_G \ldots w_{n-1}\xRightarrow{(\gamma)_n}_G
w_n = v\enspace .\]  

Given a nonterminal $X \in \Vars$ and a set $\Gamma \subseteq \prod^*$ of
control words (a.k.a \emph{control set}), we denote by $\hat{L}_{X}(\Gamma, G)
= \{ w\in \Sigma^* \mid \exists \gamma \in \Gamma \colon X \xRightarrow{\gamma}
w\}$ the language generated by $G$ using only control words in $\Gamma$. 

\subsection{Visibly Pushdown Grammars} 
To model the control flow of procedural programs we use languages
generated by visibly pushdown grammars, a subset of context-free
grammars. In this setting, words are defined over a {\em tagged
  alphabet} $\widehat{\Sigma} = \Sigma \cup \calls\Sigma \cup
\Sigma\rets$, where $\calls\Sigma = \{\calls a \mid a \in \Sigma\}$
represents procedure {\em call} sites and $\Sigma\rets = \{a \rets
\mid a \in \Sigma\}$ represents procedure {\em return}
sites. Formally, a \emph{visibly pushdown grammar} $G= \tuple{\Vars,
  \widehat{\Sigma}, \prod}$ is a grammar that has only productions of
the following forms, for some $a,b\in\Sigma$:
\begin{align*}
	X &\rightarrow a & X &\rightarrow a \, Y & X &\rightarrow \calls a \, Y\,  b \rets \, Z \enspace .
\end{align*}
 It is worth pointing that,
for our purposes, we do not need a visibly pushdown grammar to
generate the empty string \(\varepsilon\). Each tagged word generated
by visibly pushdown grammars is associated a \emph{nested word}
\cite{AM09} the definition of which we briefly recall.  Given a finite
alphabet $\Sigma$, a {\em nested word} over $\Sigma$ is a pair $(w,
\leadsto)$, where $\mathord{\leadsto} \subseteq \{1, \ldots, \len{w}\}
\times \{1, \ldots, \len{w}\}$ is a set of {\em nesting edges} (or
simply edges) where:
\begin{enumerate}
\item $i \leadsto j$ only if $i < j$; edges only go forward;
\item $\card{\{j \mid i \leadsto j\}} \leq 1$ and $\card{\{i \mid i
  \leadsto j\}} \leq 1$; no two edges share a call/return position;
\item if $i \leadsto j$ and $k \leadsto \ell$ then it is not the case
  that $i < k \leq j < \ell$; edges do not cross.
\end{enumerate}
Intuitively, we associate a nested word to a tagged word as follows:
there is an edge between tagged symbols \(\calls a\) and \(b\rets\) if
and only if both symbols are produced by the same derivation step.
Finally, let \(w\_nw\) denote the mapping which given a tagged word in
the language of a visibly pushdown grammar returns the nested word
thereof.

\begin{example}
For the tagged word
\(w=\tau_1\calls\tau_2\tau_1\calls\tau_2\tau_4\tau_2\rets\tau_3\tau_2\rets\tau_3\),
\(w\_nw(w)=(\tau_1\tau_2\tau_1\tau_2\tau_4\tau_2\tau_3\tau_2\tau_3,\set{2
  \leadsto 8, 4 \leadsto 6})\) is the associated nested word.
\(\blacksquare\)
\end{example}

% WE DON'T NEED BOTH DIRECTION + THE ONE WE NEED IS NOT WELL DEFINED.
%Finally, let \(w\_nw\) denote the mapping of a tagged word into its
%corresponding nested word, and \(nw\_w\) be its inverse, mapping a
%nested word into its corresponding tagged word.

\subsection{Integer Relations} 
Given a set $S$, let $\card{S}$ denote its cardinality. We denote by
$\zed$ the set of integers. Let $\vec{x} = \tuple{x_1,\ldots,x_d}$ be
a tuple of variables, for some $d > 0$. We define by
\(\vec{x}^\prime\) the \emph{primed} variables of \(\vec{x}\) to be
the tuple $\tuple{x'_1,x'_2,\ldots,x'_d}$. We consider implicitly that
all variables range over $\zed$. We denote by $\len{\vec{x}} = d$ the
length of the tuple $\vec{x}$, and for a tuple $\vec{y} =
\tuple{y_1,\ldots,y_e}$, we denote by $\vec{x} \cdot \vec{y} =
\tuple{x_1, \ldots, x_d, y_1, \ldots, y_e}$ their concatenation. For
two tuples of variables $\vec{t}$ and $\vec{s}$ such that
$\len{t}=\len{s}=k$, we denote by $\vec{t}=\vec{s}$ the conjunction
$\bigwedge_{i=1}^k t_i = s_i$.

A \emph{linear term} $t$ is a linear combination of the form $a_0 +
\sum_{i=1}^d a_ix_i$, where $a_0, \ldots, a_d \in \zed$. An
\emph{atomic proposition} is a predicate of the form $t \leq 0$, where
$t$ is a linear term.  We consider formulae in the first-order logic
over atomic propositions $t \leq 0$, also known as {\em Presburger
  arithmetic}.
A \emph{valuation} of $\vec{x}$ is a function $\smash{\nu : \vec{x}
  \arrow{}{} \zed}$. The set of all valuations of $\vec{x}$ is denoted
by $\zed^{\vec{x}}$. If $\vec{x} = \tuple{x_1, \ldots, x_d}$ and $\nu
\in \zed^{\vec{x}}$, then $\nu(\vec{x})$ denotes the tuple
$\tuple{\nu(x_1), \ldots, \nu(x_d)}$. An arithmetic formula
$\mathcal{R}(\vec{x}, \vec{y}')$ defining a relation \(R\subseteq
\zed^{\vec{x}}\times \zed^{\vec{y}}\) is evaluated with respect to two
valuations $\nu_1 \in \zed^{\vec{x}}$ and $\nu_2 \in \zed^{\vec{y}}$,
by replacing each $x \in \vec{x}$ by $\nu_1(x)$ and each $y' \in
\vec{y}^{\prime}$ by $\nu_2(y)$ in $\mathcal{R}$. The composition of
two relations $R_1 \subseteq \zed^{\vec{x}} \times \zed^{\vec{y}}$ and
$R_2 \subseteq \zed^{\vec{y}} \times \zed^{\vec{z}}$ is denoted by
$R_1 \comp R_2 = \{\tuple{\vec{u}, \vec{v}} \in \zed^{\vec{x}} \times
\zed^{\vec{z}} \mid \exists \vec{t}\in\zed^{\vec{y}} \ldotp
\tuple{\vec{u}, \vec{t}} \in R_1 ~\mbox{and}~ \tuple{\vec{t}, \vec{v}}
\in R_2\}$. We denote $\vec{y} \subseteq \vec{x}$ if $\vec{y} =
\tuple{x_{i_1}, \ldots, x_{i_\ell}}$, for a sequence of indices $1
\leq i_1 < \ldots < i_\ell \leq d$ of $\vec{x}$. For a valuation $\nu
\in \zed^{\vec{x}}$ and a tuple $\vec{y}\subseteq\vec{x}$, we denote
by $\nu\proj_\vec{y} \in \zed^{\vec{y}}$ the projection of $\nu$ onto
variables $\vec{y}$, i.e.\ $\nu\proj_\tuple{y_1, \ldots, y_k} =
\tuple{\nu(y_1), \ldots, \nu(y_k)}$. Finally, given two valuations $I,
O \in \zed^{\vec{x}}$, we denote by $I \cdot O$ the valuation
$I(\vec{x}) \cdot O(\vec{x})$, and we define
$\zed^{\vec{x}\times\vec{x}} = \{ I\cdot O ~|~ I, O \in \zed^\vec{x}
\}$.

\subsection{Parikh Images}
Let $\Theta = \{\theta_1, \ldots, \theta_k\}$ be a linearly ordered
subset of the alphabet \(\Sigma\). For a symbol \(a \in \Sigma\) its
{\em Parikh image} is defined as $Pk_{\Theta}(a) = \vec{e}_i$ if \(a =
\theta_i\), where \(\vec{e}_i\) is the $k$-dimensional vector having
$1$ on the $i$-th position and $0$ everywhere else. Otherwise, if \(a
\in \Sigma \setminus \Theta\), let \(Pk_{\Theta}(a) = \vec{0}\) where
\(\vec{0}\) is the \(k\)-dimensional vector with \(0\) everywhere. For
a word $w\in \Sigma^*$ of length \(n\), we define $Pk_{\Theta}(w) =
\sum_{i=1}^n Pk_{\Theta}((w)_i)$.%
\footnote{\label{foot:emptysum}We adopt the convention that the empty
  sum evaluates to \(\vec{0}\).} Furthermore, let $Pk_{\Theta}(L) = \{
Pk_{\Theta}(w) \mid w \in L \}$ for any language $L \subseteq
\Sigma^*$.

\subsection{Labelled Graphs}
In this paper we use of the notion of {\em labelled graph} $\mathcal{G} =
\tuple{Q, \mathcal{L}, \delta}$, where $Q$ is a finite set of vertices,
\(\mathcal{L}\) is a set of labels whose elements label edges as defined by the
edge relation $\delta \subseteq Q \times S \times Q$. We denote by $q
\arrow{\ell}{} q'$ the fact that $(q,\ell,q') \in \delta$. A \emph{path}
\(\pi\) in $\mathcal{G}$ is an alternating sequence of vertices and edges whose
endpoints are vertices. Sometimes, \(\pi\) is conveniently written as $q_0
\arrow{\ell_1}{} q_1 \arrow{\ell_2}{} \ldots q_{n-1} \arrow{\ell_{n}}{} q_n$ and
further abbreviated \(q_0 \arrow{w}{} q_n\) where \(w=\ell_1\ldots\ell_n\). 
% We denote by $\len{\pi} = n$ the length of the path $\pi$.

%% In this work, we will consider the monoids \(
%% (\Sigma^*, \cdot, \varepsilon)\) -- of finite words over \(\Sigma\)
%% with concatenation and $\varepsilon$ as identity element, and \(
%% (\nats^p, +, \vec{0}) \) -- of \(p\)-dimensional vectors with
%% pointwise addition and identity element $\vec{0}$. Because it is
%% convenient, we also write \(q \xrightarrow{v} q'\) to denote a path
%% from \(q\) to \(q'\) whose value is \(v\in S\).
 % end input /Users/pierreganty/counter-recursive/preliminaries.tex
 %
% start input /Users/pierreganty/counter-recursive/integer-programs.tex
% vim:ts=2:sw=2
%
%        File: integer-programs.tex
%     Created: Fri July 27
% Last Change: $Date: 2015-07-13 13:55:15 +0200 (Mon, 13 Jul 2015) $
%
% Written to be compiled by pdflatex
% usual typos to check:
% twice the same word: \(\<\w*\>\)\_s*\1\>
% too much space in math environment: \\\\\_s*\\end
% useless space at the end of a line: %s/\s*$//
% no concluded proof
% non matching parenthesis

\section{Integer Recursive Programs}\label{sec:int:programs}

We consider in the following that programs are collections of
procedures calling each other, possibly according to recursive
schemes. Formally, an \emph{integer program} is an indexed tuple
$\mathcal{P} = \langle P_1,\ldots,P_n \rangle$, where $P_1,\ldots,P_n$
are \emph{procedures}. Each procedure is a tuple $P_i = \langle
\vec{x}_i, \vec{x}^{in}_i, \vec{x}^{out}_i, S_i, q^{\mathit{init}}_i,
F_i, \Delta_i \rangle$, where $\vec{x}_i$ are the \emph{local}
variables\footnote{Observe that there are no global variables in the
  definition of integer program. Those can be encoded as input and
  output variables to each procedure.} of $P_i$ ($\vec{x}_i \cap
\vec{x}_j = \emptyset$ for all $i \neq j$), $\vec{x}^{in}_i,
\vec{x}^{out}_i \subseteq \vec{x}_i$ are the tuples of input and
output variables, $S_i$ are the \emph{control states} of $P_i$ ($S_i
\cap S_j = \emptyset$, for all $i \neq j$), $q^{\mathit{init}}_i \in
S_i \setminus F_i$ is the {\em initial}, and $F_i \subseteq S_i$ ($F_i
\neq \emptyset$) are the {\em final} states of $P_i$, and $\Delta_i$
is a set of \emph{transitions} of one of the following forms:

\begin{itemize}
\item $q \arrow{\mathcal{R}(\vec{x}_i,\vec{x}'_i)}{} q'$ is an
  \emph{internal transition}, where $q,q' \in S_i$, and
  $\mathcal{R}(\vec{x}_i,\vec{x}'_i)$ is a Presburger arithmetic
  relation involving only the local variables of $P_i$;

\item $q \arrow{\vec{z}' = P_j(\vec{u})}{} q'$ is a \emph{call}, where
  $q,q' \in S_i$, $P_j$ is the callee, $\vec{u}$ are linear terms over
  $\vec{x}_i$, $\vec{z}\subseteq\vec{x}_i$ are variables, such that
  $\len{\vec{u}}=\len{\vec{x}^{in}_j}$ and
  $\len{\vec{z}}=\len{\vec{x}^{out}_j}$. The call is said to be {\em
  terminal} if $q' \in F_i$. It is well-known that terminal calls can
  be replaced by internal transitions.
\end{itemize}
%
%% We define the {\em size} of the program as $\size\mathcal{P} =
%% \sum_{i=1}^n \card{\Delta_i}$ to be the total number of transition
%% rules, and $\loc(\mathcal{P}) = \sum_{i=1}^n \card{S_i}$ be the number
%% of control locations in $\mathcal{P}$. 
%
The {\em call graph} of a program $\mathcal{P} = \langle
P_1,\ldots,P_n \rangle$ is a directed graph with vertices
$P_1,\ldots,P_n$ and an edge $(P_i, P_j)$, for each $P_i$ and $P_j$,
such that $P_i$ has a call to $P_j$. A program is {\em recursive} if
its call graph has at least one cycle, and {\em non-recursive} if its
call graph is a dag.

In the rest of this paper, we denote by
\(\fin{\mathcal{P}}=\bigcup_{i=1}^n F_i\) the set of final states of
the program $\mathcal{P}$, by \(\nf{P_i}\) the set \(S_i\setminus
F_i\) of non-final states of \(P_i\), and by
\(\nf{\mathcal{P}}=\bigcup_{i=1}^n \nf{\mathcal{P}} \) be the set of
non-final states of $\mathcal{P}$.

\subsection{Simplified syntax} 

To ease the description of programs defined in this paper, we use a
simplified, human readable, imperative language such that each
procedure of the program conforms to the following grammar:%
\footnote{Our simplified syntax does not seek to capture the
  generality of integer programs. Instead, our goal is to give a
  convenient notation for the programs given in
  this paper and only those.}
\begin{align*}
P &::= \mathbf{proc}\ P_i ( \mathit{id}^* ) ~\mathbf{begin}~\mathbf{var}~\mathit{id}^*~S_0;\, S~\mathbf{end} \\
\begin{split}
S_0 &::= \mathbf{assume}\ f \mid \mathbf{goto}\ \ell^+ \mid \mathbf{havoc}\ \mathit{id}^+ \mid id \leftarrow t\\
S   &::= S_0 \mid S;\, S \mid id \leftarrow P_i(t^*);\, S_0 \mid P_i(t^*);\, S_0 \mid \mathbf{return}\ \mathit{id}
\end{split}
\end{align*}
The local variables occurring in \(P\) are denoted by \(id\), linear
terms by \(t\), Presburger formulae by \(f\), and control labels by
\(\ell\). Each procedure consists in local declarations followed by a
sequence of statements. Statements may carry a label. Program
statements can be either {\bf assume} statements\footnote{{\bf assume}
  $\phi$ is executable if and only if the current values of the
  variables satisfy the Presburger formula \(\phi\).}, assignments,
procedure calls (possibly with a return value), return to the caller
(possibly with a value), non-deterministic jumps {\bf goto} \(\ell_1
~\mbox{\bf or}~ \ldots ~\mbox{\bf or}~ \ell_n\), and {\bf havoc}
\(x_1,x_2,\ldots,x_n\) statements\footnote{{\bf havoc} assigns non
  deterministically chosen integers to \(x_1,x_2,\ldots,x_n\).}. In
order to simplify the upcoming technical developments, we forbid empty
procedures, procedures starting with a call or a return, i.e.\ each
procedure must start with a statement generated by the $S_0$
nonterminal. We consider the usual syntactic requirements (used
variables must be declared, jumps are well defined, no jumps outside
procedures, etc.). We do not define them, it suffices to know that all
simplified programs in this paper comply with the requirements.  A
program using the simplified syntax can be easily translated into the
formal syntax (Fig.~\ref{fig:program}).

\begin{example}\label{ex:prg}
Figure~\ref{fig:program} shows a program in our simplified imperative
language and its corresponding integer program \(\mathcal{P}\).
Formally, $\mathcal{P} = \langle P \rangle$, where \(P\) is the only
procedure in the program, defined as:
\begin{multline*}
P = \langle \{ x,z \}, \{ x \}, \{ z \},
  \{q_1^{\mathit{init}},q_2,q_3,\varepsilon\},
  q_1^{\mathit{init}}, \{\varepsilon\}, \{t_1, t_2, t_3, t_4\} \rangle
\end{multline*}
Since $P$ calls itself once (within the call transition
$t_2$), this program is recursive.\hfill \(\blacksquare\)%
\end{example}

\subsection{Semantics} 

We are interested in computing the {\em summary relation} between the
values of the input and output variables of a procedure. To this end,
we give the semantics of a program $\mathcal{P} = \langle
P_1,\ldots,P_n \rangle$ as a tuple of relations, denoted $\sem{q}$ in
the following, describing, for each non-final control state $q \in
\nf{P_i}$ of a procedure $P_i$, the effect of the program when started
in $q$ upon reaching a state in $F_i$. The summary of a procedure
$P_i$ is the relation corresponding to its unique initial state,
i.e.\ $\sem{q_i^{\mathit{init}}}$.

An {\em interprocedurally valid path} is represented by a tagged word
over an alphabet $\widehat{\Theta}$, which maps each internal
transition $t$ to a symbol $\tau$, and each call transition \(t\) to a
pair of symbols $\calls \tau, \tau \rets \in \widehat{\Theta}$. In the
sequel, we denote by $Q$ the nonterminal corresponding to the control
state $q$, and by \(\tau\in \Theta\) the alphabet symbol corresponding
to the transition \(t\) of \(\mathcal{P}\). Formally, we associate
$\mathcal{P}$ a visibly pushdown grammar, denoted in the rest of the
paper by $G_{\mathcal{P}} = \tuple{\Vars, \widehat{\Theta}, \prod}$,
such that $Q \in \Vars$ if and only if $q \in \nf{\mathcal{P}}$
and:
\begin{compactitem}
\item[(a)] $Q \rightarrow \tau \in \prod$ if and only if $t \colon q
  \arrow{\mathcal{R}}{} q'$ and $q' \in \fin{\mathcal{P}}$

\item[(b)] $Q \rightarrow \tau ~ Q' \in \prod$ if and only if $t
  \colon q \arrow{\mathcal{R}}{} q'$ and $q' \in \nf{\mathcal{P}}$

\item[(c)] $Q \rightarrow \calls\tau ~ Q_j^{\mathit{init}} ~ \tau\rets
  ~ Q' \in \prod$ if and only if $t\colon q \arrow{\vec{z}' =
  P_j(\vec{u})}{} q'$.
\end{compactitem}
It is easily seen that interprocedurally valid paths in \(\mathcal{P}\) and
tagged words in \(G_{\mathcal{P}}\) are in one-to-one correspondence.  In fact,
each interprocedurally valid path of \(\mathcal{P}\) between state $q \in
\nf{P_i}$ and a state of $F_i$, where $1\leq i\leq n$, corresponds exactly to
one tagged word of $L_{Q}(G_{\mathcal{P}})$. 

\begin{example}\label{ex:vpl} (contd. from Ex.~\ref{ex:prg}) 
  The visibly pushdown grammar \(G_{\mathcal{P}}\) corresponding to
  \(\mathcal{P}\) is given in Fig. \ref{fig:program} (c). In the
  following, we use superscripts $a,b,c$ to distinguish productions of
  the form ($a$) \(Q\rightarrow \tau\), ($b$) \(Q\rightarrow \tau\,
  Q'\) or ($c$) \(Q\rightarrow \calls \tau \,Q_j^{\mathit{init}}\,
  \tau \rets\, Q'\), respectively.  The language
  \(L_{Q_1^{\mathit{init}}}(G_{\mathcal{P}})\) generated by
  $G_{\mathcal{P}}$ starting with $Q_1^{\mathit{init}}$ contains the
  word
  \(w=\tau_1\calls\tau_2\tau_1\calls\tau_2\tau_4\tau_2\rets\tau_3\tau_2\rets\tau_3\),
  of which
  \(w\_nw(w)=(\tau_1\tau_2\tau_1\tau_2\tau_4\tau_2\tau_3\tau_2\tau_3,\set{2
    \leadsto 8, 4 \leadsto 6})\) is the corresponding nested word.
  The word \(w\) corresponds to an interprocedurally valid path where
  \(P\) calls itself twice. The control words
  \(\gamma_1=p_1^bp_2^cp_1^bp_2^cp_4^ap_3^ap_3^a\) and
  \(\gamma_2=p_1^bp_2^cp_3^ap_1^bp_2^cp_4^ap_3^a\) both produce $w$ in
  this case, i.e.\ \(Q_1^{\mathit{init}}
  \stackrel{\gamma_1}\Longrightarrow w\) and \(Q_1^{\mathit{init}}
  \stackrel{\gamma_2}\Longrightarrow w\).  \hfill \(\blacksquare\)%
\end{example}

The semantics of a program is the union of the semantics of the nested
words corresponding to its executions, each of which being a relation
over input and output variables. To define the semantics of a nested
word, we first associate to each $\tau \in \widehat{\Theta}$ an
integer relation $\rho_\tau$, defined as follows:
\begin{compactitem}
\item for an internal transition $t \colon q \arrow{\mathcal{R}}{} q'
  \in \Delta_i$, we define $\rho_{\tau} \equiv
  \mathcal{R}(\vec{x}_i,\vec{x}'_i) \subseteq \zed^{\vec{x}_i} \times
  \zed^{\vec{x}_i}$;

\item for a call transition $t : q \arrow{\vec{z}' = P_j(\vec{u})}{}
  q' \in \Delta_i$, we define a \emph{call relation} $\rho_{\calls
    \tau} \equiv ({\vec{x}^{in}_j}^\prime = \vec{u}) \subseteq
  \zed^{\vec{x}_i} \times \zed^{\vec{x}_j}$, a \emph{return relation}
  $\rho_{\tau \rets} \equiv ({\vec{z}}^\prime = \vec{x}^{out}_j)
  \subseteq \zed^{\vec{x}_j} \times \zed^{\vec{x}_i}$ and a
  \emph{frame relation} $\phi_{\tau}\equiv
  \textstyle{\bigwedge_{x\in\vec{x}_i\setminus\vec{z}}} x'=x \subseteq
  \zed^{\vec{x}_i} \times \zed^{\vec{x}_i}$. Intuitively, the frame
  relation copies the values of all local variables, that are not
  involved in the call as return value receivers ($\vec{z}$), across
  the call.
\end{compactitem}
We define the semantics of the program $\mathcal{P} = \langle
P_1,\ldots,P_n \rangle$ in a top-down manner. Assuming a fixed
ordering of the non-final states in the program,
i.e. $\nf{\mathcal{P}} = \langle q_1,\ldots,q_m \rangle$, the
semantics of the program $\mathcal{P}$, denoted \(\sem{\mathcal{P}}\),
is the tuple of relations $\langle \sem{q_1}, \ldots, \sem{q_m}
\rangle$.  For each non-final control state $q \in \nf{P_i}$ where
\(1\leq i\leq n\), we denote by $\sem{q} \subseteq \zed^{\vec{x}_{i}}
  \times \zed^{\vec{x}_i}$ the relation (over the local variables of
procedure $P_i$) defined as \(\sem{q} = \bigcup_{\alpha \in
  L_{Q}(G_{\mathcal{P}})} \sem{\alpha}\).

It remains to define \(\sem{\alpha}\), the semantics of the tagged
word (or equivalently interprocedural valid path) \(\alpha\). Out of convenience, we define the semantics of its
corresponding nested word \(w\_nw(\alpha)=(\theta,\mathord{\leadsto})\)
over alphabet \(\Theta\), and define $\sem{\alpha} =
\sem{w\_nw(\alpha)}$. For a nesting relation \(\mathord{\leadsto}
\subseteq \{1, \ldots, \len{\theta}\} \times \{1, \ldots,
\len{\theta}\}\), we define \(\mathord{\leadsto}_{i,j} =
\{(s-(i{-}1),t-(i{-}1)) \mid (s,t) \in \mathord{\leadsto} \cap
\set{i,\ldots,j}\times \set{i,\ldots,j}\}\), for some $i,j \in \{1,
\ldots, \ell\}$, $i < j$. Finally, we define \(\sem{(\theta,
  \mathord{\leadsto})} \subseteq \zed^{\vec{x}_i} \times
\zed^{\vec{x}_i}\) as follows:
	%\sem{(\theta_, \mathord{\leadsto})} &=
\[\begin{cases}
  \rho_{(\theta)_1} &\text{if } \len{\theta} = 1 
  \\
  \rho_{(\theta)_1} \comp \sem{((\theta)_{2\ldots\len{\theta}},\mathord{\leadsto}_{2,\len{\theta}})} & 
	\text{if } \len{\theta} > 1 ,\, 1 \leadsto j \text{ for no } j 
  \\
	\mathit{CaRet}^j_\theta\comp \sem{((\theta)_{{j+1} \ldots \len{\theta}},\mathord{\leadsto}_{j+1,\len{\theta}})} & 
	\text{if } \len{\theta} > 1 ,\, 1 \leadsto j \text{ for a }j
\end{cases}\]
where, in the last case, which corresponds to call transition \(t \in
\Delta_i\), we have \((\theta)_1=(\theta)_j=\tau\) and define
\(\mathit{CaRet}^j_\theta= \bigl(\rho_{\calls\tau} \comp \sem{(\theta)_{2\ldots{j-1}}, \mathord{\leadsto}_{2,j-1})} \comp \rho_{\tau\rets}\bigr) \cap \phi_{\tau}\).

\begin{example} (contd. from Ex.~\ref{ex:vpl})
The semantics of a given the nested word \(
\theta=(\tau_1\tau_2\tau_1\tau_2\tau_4\tau_2\tau_3\tau_2\tau_3,\set{2
  \leadsto 8, 4 \leadsto 6})\) is a relation between valuations of
$\{x,z\}$, given by:
\[\begin{array}{rcl}
\sem{\theta} & = & \rho_{\tau_1} \comp \bigl((\rho_{\calls\tau_2} \comp \rho_{\tau_1} 
\comp \bigl((\rho_{\calls\tau_2} \comp \rho_{\tau_4} \comp \rho_{\tau_2\rets}) \cap \phi_{\tau_2}\bigr) \\
&&  \comp \rho_{\tau_3} \comp \rho_{\tau_2\rets}) \cap \phi_{\tau_2}\bigr) \comp \rho_{\tau_3}
\end{array}\] 
One can verify that $\sem{\theta} \equiv x=2 \wedge z'=4$, i.e. the
result of calling $P$ with input valuation $x=2$ is an output
valuation $z=4$.\hfill \(\blacksquare\)%
\end{example}

Finally, we introduce a few useful notations.  An interprocedural
valid path \(\alpha\) is said to be \emph{feasible} whenever
\(\sem{\alpha}\neq\emptyset\). We denote by $\sem{\mathcal{P}}_{q}$
the component of $\sem{\mathcal{P}}$ corresponding to \(q\in
\nf{\mathcal{P}}\). Notice that $\sem{\mathcal{P}}_q \in
\zed^{\vec{x}_i} \times \zed^{\vec{x}_i}$, i.e.\ is a relation over
the valuations of the local variables of the procedure $P_i$ if $q$ is
a state of $P_i$, i.e.\ $q \in S_i$.  Slightly abusing notations, we
define \(L_{P_i}(G_{\mathcal{P}})\) as
\(L_{Q_i^{\mathit{init}}}(G_{\mathcal{P}})\) and
\(\sem{\mathcal{P}}_{P_i}\) as
\(\sem{\mathcal{P}}_{q^{\mathit{init}}_i}\).  Clearly we have that
$\sem{\mathcal{P}}_{P_i} \subseteq \zed^{\vec{x}_i} \times
\zed^{\vec{x}_i}$.

\subsection{A Semantics of Depth-First Derivations}

We present an alternative, but equivalent, program semantics, using
derivations of visibly pushdown program grammars, instead of the
generated (nested) words. This semantics brings us closer to the
notion of under-approximation defined in the next section. 

We start by defining {\em depth-first derivations}, that have the
following informal property: if \(X\) and \(Y\) are two nonterminals
produced by the application of one rule, then the steps corresponding
to a full derivation of the form \(X \xRightarrow{}^* u\) will be
applied {\em without interleaving} with the steps corresponding to a
derivation of the form \(Y \xRightarrow{}^* v\). In other words, once
the derivation of \(X\) has started, it will be finished before the
derivation of \(Y\) begins.

For an integer tuple $\vec{\alpha} = \tuple{\alpha_1, \ldots,
  \alpha_n}$, we denote by $\| \vec{\alpha} \|_{\max} = \max_{i=1}^n
\alpha_i$. For a set of symbols $S \subseteq \Vars \cup \Sigma$, and a
set of positive integers $I \subseteq \nat$, we define $S^I =
\{x^{\tuple{i}} \mid x\in S,\, i \in I \}$.  Given a word \(w\in
(\Vars\cup\Sigma)^*\) of length \(n\geq 0\), and a \(n\)-dimensional
vector $\vec{\alpha} = \tuple{ \alpha_1, \ldots, \alpha_n } \in
\nats^{n}$, we define \(w^{\vec{\alpha}}\) as the
\emph{birthdate-annotated word} (bd-word) \(
     {(w)_1}^{\tuple{\alpha_1}}\ldots {(w)_n}^{\tuple{\alpha_n}}\)
     over the alphabet \((\Vars \cup \Sigma)^{\nats}\). We denote
		 \(w^{ {\langle}{\langle} {c} {\rangle}{\rangle}} = w^{\vec{c}}\), where \(c \in \nat\) and
     \(\vec{c} = \tuple{c,\ldots,c}\in\nat^{|w|}\). For instance, \(
     abc^{\tuple{1,2,3}} = a^{\tuple{1}}\, b^{\tuple{2}}\,
		 c^{\tuple{3}} \) and \( abc^{{\langle}{\langle} {2} {\rangle}{\rangle}} = a^{\tuple{2}}
     b^{\tuple{2}} c^{\tuple{2}}\).

Let $G = \tuple{\Vars,\Sigma,\prod}$ be a grammar and \(u
\xRightarrow{(Z,w)/j} v\) be a step, for some production $(Z,w) \in
\prod$ and \( 1\leq j\leq |u|\). If \(\alpha\in\nat^{|u|}\) is a
vector of birthdates, the corresponding \emph{birthdate-annotated
  step} (bd-step) is defined as follows: \(u^{\vec{\alpha}}
\xRightarrow{(Z,w)/j} v^{\vec{\beta}}\) if and only if \(
(u^{\vec{\alpha}})_j = Z^{\tuple{i}} \) and \(v^{\vec{\beta}} =
(u^{\vec{\alpha}})_1 \cdots (u^{\vec{\alpha}})_{j-1} \cdot w^{{\langle}{\langle}
{\| \vec{\alpha} \|_{\max}+1} {\rangle}{\rangle} } \cdot
(u^{\vec{\alpha}})_{j+1}\cdots (u^{\vec{\alpha}})_{|u|} \).

\begin{example}\label{ex:ba-derivations}
  Consider the grammar $G = \tuple{\{X,Y,Z\},\, \{a,b\}, \prod}$ with
  rules $\prod = \{X \rightarrow YZ,\ Y \rightarrow aY \mid \varepsilon,\ Z
  \rightarrow Zb \mid \varepsilon\}$. Then $X^{\tuple{0}}
  \xRightarrow{(X,YZ)} Y^{\tuple{1}}Z^{\tuple{1}} \xRightarrow{(Y,aY)}
  a^{\tuple{2}}Y^{\tuple{2}}Z^{\tuple{1}} \xRightarrow{(Z,Zb)}
  a^{\tuple{2}}Y^{\tuple{2}}Z^{\tuple{3}}b^{\tuple{3}} \xRightarrow{(Y,\varepsilon)}
  a^{\tuple{2}}Z^{\tuple{3}}b^{\tuple{3}} \xRightarrow{(Z,\varepsilon)}
  a^{\tuple{2}}b^{\tuple{3}}$ and \linebreak$X^{\tuple{0}} \xRightarrow{(X,YZ)}
  Y^{\tuple{1}}Z^{\tuple{1}} \xRightarrow{(Y,aY)}
  a^{\tuple{2}}Y^{\tuple{2}}Z^{\tuple{1}} \xRightarrow{(Y,\varepsilon)}
  a^{\tuple{2}}Z^{\tuple{1}} \xRightarrow{(Z,Zb)}
  a^{\tuple{2}}Z^{\tuple{3}}b^{\tuple{3}} \xRightarrow{(Z,\varepsilon)}
  a^{\tuple{2}}b^{\tuple{3}}$ are birthdate-annotated step sequences. \(\blacksquare\)
\end{example}

A birthdate annotated step is further said to be \emph{depth-first}
whenever, in the above definition of a bd-step, we have, moreover,
that \(i\) is the most recent birthdate among the nonterminals of
\(u\) , i.e.\ \(i = \max \set{j \mid
  Pk_{\Vars^{\{j\}}}(u^{\vec{\alpha}})\neq\vec{0}}\).  We write this
fact as follows \( u^{\vec{\alpha}} \xRightarrow[\textbf{df}]{}
v^{\vec{\beta}}\). A birthdate annotated step sequence is said to be
depth-first if all of its steps are depth-first.
Finally, a step sequence \(w_0 \xRightarrow{(\gamma)_1/j_1} w_1 \ldots w_{n-1} \xRightarrow{(\gamma)_{n}/j_{n}} w_n\) for some control word \(\gamma\) is said to be depth-first,
written \(w_0 \xRightarrow[\mathbf{df}]{\gamma} w_n\), if there exist vectors
\(\vec{\alpha}_1 \in \nats^{\card{Pk_\Vars(w_1)}},\ldots,\vec{\alpha}_n \in \nats^{\card{Pk_\Vars(w_n)}}\) 
such that  \(w_0^{ {\langle}{\langle} 0 {\rangle}{\rangle} } \xRightarrow[\mathbf{df}]{(\gamma)_1/j_1} w_1^{\vec{\alpha}_1} \ldots w_{n-1}^{\vec{\alpha}_{n-1}} \xRightarrow[\mathbf{df}]{(\gamma)_{n}/j_{n}} w_{n}^{\vec{\alpha}_{n}}\) holds.

\begin{example}\label{ex:depth-first}(contd. from Ex. \ref{ex:ba-derivations})
  Consider the grammar $G$ from Example \ref{ex:ba-derivations}. Then
  $X \xRightarrow{(X,YZ)} YZ \xRightarrow{(Y,aY)} aYZ
  \xRightarrow{(Z,Zb)} aYZb \xRightarrow{(Y,\epsilon)} aZb
  \xRightarrow{(Z,\epsilon)} ab$ is not a depth-first derivation,
  whereas $X \xRightarrow{(X,YZ)} YZ \xRightarrow{(Y,aY)} aYZ
  \xRightarrow{(Y,\epsilon)} aZ \xRightarrow{(Z,Zb)} aZb
  \xRightarrow{(Z,\epsilon)} ab$ is a depth-first derivation. \(\blacksquare\)
\end{example}

Since we are dealing with visibly pushdown grammars $G_{\mathcal{P}} =
\tuple{\Vars, \widehat{\Theta}, \prod}$ corresponding to programs
$\mathcal{P}$, for every production $Q \arrow{}{} \calls \tau
Q_j^{init} \tau\rets Q' \in \prod$ we have $Q_j^{init} \neq Q'$. 
Hence, we can assume wlog that for all productions \(p\in\prod\), all nonterminals occurring in
\(\mathit{tail}(p)\) are distinct (e.g.\ \(X\rightarrow Z\, Z\) is not allowed).
As we show next, under that assumption, a control word uniquely identifies a
depth-first derivation:
\begin{lemma}\label{lemma:unique-df-control-word}
  Let $G_{\mathcal{P}} = \tuple{\Vars, \widehat{\Theta}, \prod}$ be a
  visibly pushdown grammar corresponding to a program $\mathcal{P}$,
  $Q \in \Vars$ be a nonterminal, $Q \xRightarrow[\textbf{df}]{\gamma}
  u$ and $Q \xRightarrow[\textbf{df}]{\gamma} v$ be two
  depth-first derivations of $G_{\mathcal{P}}$. Then they differ in no step, hence $u = v$.
\end{lemma}
\begin{proof}
By contradiction, suppose that there exists a step that differs in the two derivations
from \(Q\) with control word \(\gamma\in\prod^*\). Thus, there exists an integer \(i\), \(1\leq i < \len{\gamma}\), such that \(Q =
w_0 \stackrel{(\gamma)_1}{\Longrightarrow} w_1 \cdots w_{i-1}
\stackrel{(\gamma)_i}{\Longrightarrow} w_i \) and 
\(w_i\) contains two
occurrences of the nonterminal \(\mathit{head}( (\gamma)_{i+1} )\), that is,
there exists \(p_1 \neq p_2\) \( (w_i)_{p_1} = (w_i)_{p_2} = \mathit{head}( (\gamma)_{i+1} )\). Two cases arise:
\begin{compactenum}
\item \((w_i)_{p_1}\) and \( (w_i)_{p_2}\) result from the occurrence of some \( (\gamma)_j \)
  with \(j\leq i\) which contradicts that all nonterminals occurring in
  \(\mathit{tail}( (\gamma)_j)\) are distinct.
\item \((w_i)_{p_1}\) and \( (w_i)_{p_2}\) result from the occurrence of \( (\gamma)_k\) and
  \( (\gamma)_l\) with \(k\neq l\) respectively. Hence in the bd-step sequence thereof, their birthdate necessarily differ.
	Therefore there is only one occurence of \(\mathit{head}( (\gamma)_{i+1} )\) with the most recent birthdate which contradicts the
  existence of two distinct depth-first derivations.
\end{compactenum}%
\qed\end{proof}

Consequently, in a visibly pushdown grammar corresponding to a
program, a control word uniquely determines a step sequence, and,
moreover, if this step sequence is a derivation, the control word
determines the word produced by it. This remark leads to the
definition of an alternative semantics of programs, based on control
words, instead of produced words. To this end, for each non-final
control location $q \in \nf{P_i}$, of a program $\mathcal{P} =
\tuple{P_1, \ldots, P_n}$, where $1 \leq i \leq n$, we define the
semantics of a control word $\gamma$ that induces a depth-first
derivation $Q \xRightarrow[\textbf{df}]{\gamma} w$ of the grammar
$G_{\mathcal{P}} = \tuple{\Vars, \widehat{\Theta}, \prod}$, as a set
$\sem{\gamma} \subseteq \zed^{\vec{x}} \times \zed^{\vec{x}}$, where
$\vec{x} = \vec{x}_1 \cdot \ldots \cdot \vec{x}_n$ is the set of
variables in $\mathcal{P}$. The definition of $\sem{\gamma}$ is by
induction on the structure of $\gamma$:
\begin{compactenum}[(a)]
\item if $\gamma = Q \arrow{}{} \tau$ then $\sem{\gamma} = \set{I
  \cdot O \mid \tuple {I\proj_{\vec{x}_i}, O\proj_{\vec{x}_i}} \in
  \rho_\tau}$, where $Q\in\Vars$ corresponds to $q \in \nf{P_i}$;
\item if $\gamma = (Q \arrow{}{} \tau Q') \cdot \gamma'$ then
  \[\sem{\gamma} = \set{I \cdot O \mid \exists J ~.~
  \tuple{I\proj_{\vec{x}_i}, J\proj_{\vec{x}_i}} \in \rho_\tau
  ~\mbox{and}~ J \cdot O \in \sem{\gamma'}}\] where $Q,Q' \in \Vars$
  correspond to $q,q' \in \nf{P_i}$;
\item if $\gamma = (Q \arrow{}{} \calls\tau Q_j^{init} \tau\rets Q')
	\cdot \gamma'$ then \(\sem{\gamma}\) is given by  
\begin{multline*}
\{I \cdot O  \mid  \exists J,K,L\in\zed^{\vec{x}} \ldotp
\tuple{I\proj_{\vec{x}_i}, J\proj_{\vec{x}_j}} \in \rho_{\calls\tau},\, J \cdot K \in \sem{\gamma_1}, \\
\tuple{K\proj_{\vec{x}_j},L\proj_{\vec{x}_i}} \in \rho_{\tau\rets},\, \tuple{I \proj_{\vec{x}_i}, 
L \proj_{\vec{x}_i}} \in \phi_\tau,\, L \cdot O \in \sem{\gamma_2}\}\enspace ,
\end{multline*}
	where $Q_j^{init}, Q' \in \Vars$ correspond to $q_j^{init}$
        (the initial control location of $P_j$), $q' \in \nf{P_i}$,
        and $Q_j^{init} \xRightarrow[\textbf{df}]{\gamma_1} w_1$, $Q'
        \xRightarrow[\textbf{df}]{\gamma_2} w_2$, $\gamma'=
        \gamma_1\gamma_2$, respectively; since $\gamma$ is the control
        word of a depth-first derivation, the derivations of
        $Q_j^{init}$ and $Q'$ are unique, and will not interleave with
        each other.
\end{compactenum}
The following lemma proves the equivalence of the semantics of a (tagged) word
generated by a visibly pushdown grammar and that of a control word that
produces it.

\begin{lemma}\label{lemma:sem-equiv}
  Let $G_{\mathcal{P}} = \tuple{\Vars, \widehat{\Theta}, \prod}$ be a
  visibly pushdown grammar for a program $\mathcal{P} = \tuple{P_1,
    \ldots, P_n}$, $\vec{x} = \vec{x}_1 \cdot \ldots \cdot \vec{x}_n$
  be the concatenation of all tuples of local variables in
  $\mathcal{P}$, $Q \in \Vars$ be a nonterminal corresponding to a
  non-final control location $q \in \nf{P_i}$, and $Q
  \xRightarrow[\textbf{df}]{\gamma} \alpha$ be a depth-first
  derivation of $G_{\mathcal{P}}$, where $\alpha \in
  \widehat{\Theta}^*$ and $\gamma \in \prod^*$. Then, we have:
  \[\sem{\gamma} = \set{I \cdot O \in \zed^{\vec{x} \times \vec{x}}
    \mid \tuple{I\proj_{\vec{x}_i}, O\proj_{\vec{x}_i}} \in
    \sem{\alpha}} \enspace.\]
\end{lemma}
\begin{proof}
By induction on $\len{\gamma} > 0$. If $\len{\gamma} = 1$,
i.e.\ $\gamma = Q \arrow{}{} \tau$, we have $\alpha = \tau$, hence
$\sem{\alpha} = \sem{w\_nw(\alpha)} = \rho_\tau$ and the equality
follows trivially. If $\len{\gamma} > 1$, let $\gamma = p \cdot
\gamma'$, for some $p \in \prod$ and some $\gamma' \in \prod^*$. We
distinguish two cases, based on the type of $p$:
\begin{compactitem}
\item $p = Q \arrow{}{} \tau\, Q'$: in this case $\alpha = \tau \cdot \beta$
  and $Q' \xRightarrow[\textbf{df}]{\gamma'} \beta$ is a depth-first
  derivation of $G_{\mathcal{P}}$. By the induction hypothesis, since
  $\len{\gamma'} < \len{\gamma}$, we have $\sem{\gamma'} = \set{J
    \cdot O \mid \tuple{J\proj_{\vec{x}_i}, O\proj_{\vec{x}_i}} \in
    \sem{\beta}}$.
  \[\begin{array}{rcl} 
  \sem{\gamma} & = & \{I \cdot O \mid \exists J ~.~ 
    \tuple{I\proj_{\vec{x}_i}, J\proj_{\vec{x}_i}} \in \rho_\tau ~\mbox{and}~ 
    \tuple{J\proj_{\vec{x}_i}, O\proj_{\vec{x}_i}} \in \sem{\beta}\} \\ 
  & = & \{I \cdot O \mid \tuple{I\proj_{\vec{x}_i},
      O\proj_{\vec{x}_i}} \in \sem{w\_nw(\alpha)}\} \\
  & = & \{I \cdot O \mid \tuple{I\proj_{\vec{x}_i},
      O\proj_{\vec{x}_i}} \in \sem{\alpha}\} 
  \end{array}\]

\item $p = Q \arrow{}{} \calls\tau\, Q_j^{init}\, \tau\rets\, Q'$: in this
  case $\alpha = \calls\tau\, \beta_1\, \tau\rets\, \beta_2$ and
  $G_{\mathcal{P}}$ has depth-first derivations $Q_j^{init}
  \xRightarrow[\textbf{df}]{\gamma_1} \beta_1$ and $Q'
  \xRightarrow[\textbf{df}]{\gamma_2} \beta_2$. We have two
  symmetrical cases: either $\gamma' = \gamma_1\gamma_2$ or
  $\gamma' =\gamma_2\gamma_1$. We consider the first case in the
  following:
  \[\begin{array}{rcll}
  \sem{\gamma} & = & \{I \cdot O \mid & \exists
	J,K,L\in\zed^{\vec{x}} ~.~
        \tuple{I\proj_{\vec{x}_i}, J\proj_{\vec{x}_j}} \in \rho_{\calls\tau}, \\ 
    &&& J \cdot K \in \sem{\gamma_1},~ \tuple{K\proj_{\vec{x}_j},~ 
      L\proj_{\vec{x}_i}} \in \rho_{\tau\rets}, \\
    &&& \tuple{I \proj_{\vec{x}_i}, L \proj_{\vec{x}_i}} \in \phi_\tau,~ 
        L \cdot O \in \sem{\gamma_2}\}
  \end{array}\]
\end{compactitem}
We apply the induction hypothesis to $\gamma_1$ and $\gamma_2$, since
$\len{\gamma_1} < \len{\gamma}$ and $\len{\gamma_2} < \len{\gamma}$,
and obtain: 
\[\begin{array}{rcll}
\sem{\gamma} & = & \{I \cdot O \mid & \exists
	J,K,L\in\zed^{\vec{x}} ~.~
        \tuple{I\proj_{\vec{x}_i}, J\proj_{\vec{x}_j}} \in \rho_{\calls\tau}, \\ 
    &&& \tuple{J\proj_{\vec{x}_j}, K\proj_{\vec{x}_j}} \in \sem{\beta_1},~ \tuple{K\proj_{\vec{x}_j},~ 
      L\proj_{\vec{x}_i}} \in \rho_{\tau\rets}, \\
    &&& \tuple{I \proj_{\vec{x}_i}, L \proj_{\vec{x}_i}} \in \phi_\tau,~ 
        \tuple{L\proj_{\vec{x}_i}, O\proj_{\vec{x}_i}} \in \sem{\beta_2}\} \\
& = & \{I \cdot O \mid & \tuple{I\proj_{\vec{x}_i},O\proj_{\vec{x}_i}} \in \sem{w\_nw(\alpha)}\} \\
& = & \{I \cdot O \mid & \tuple{I\proj_{\vec{x}_i},O\proj_{\vec{x}_i}} \in \sem{\alpha}\}
\end{array}\]
\qed
\end{proof}

Consequently, the semantics of a program $\mathcal{P} = \tuple{P_1,
  \ldots, P_n}$ can be equivalently defined considering the sets
  \[ \sem{\mathcal{P}}_q = \{\tuple{I \proj_{\vec{x}_i}, O \proj_{\vec{x}_i}}
	\mid I \cdot O \in \textstyle{\bigcup_{Q \xRightarrow[\textbf{df}]{\gamma} w}}
  \sem{\gamma} \}\enspace ,\]
	for each  non-final state $q \in \nf{P_i}$ of the procedure $P_i$ of $\mathcal{P}$.

 % end input /Users/pierreganty/counter-recursive/integer-programs.tex
 %
% start input /Users/pierreganty/counter-recursive/bounded-query.tex
% vim:ts=2:sw=2
%
%        File: bounded-query.tex
%     Created: Fri July 27
% Last Change: $Date: 2016-03-15 18:28:06 +0100 (Tue, 15 Mar 2016) $
%
% Written to be compiled by pdflatex
% usual typos to check:
% twice the same word: \(\<\w*\>\)\_s*\1\>
% too much space in math environment: \\\\\_s*\\end
% useless space at the end of a line: %s/\s*$//
% no concluded proof
% non matching parenthesis

\section{Underapproximating the Program Semantics}\label{sec:bounded-query}

In what follows we define context-free language underapproximations by
filtering out derivations. In particular, in this section, we define a
family of underapproximations of \(\sem{\mathcal{P}}\), called {\em
  bounded-index underapproximations}. Then we show that each
\(k\)-index underapproximation of the semantics of a (possibly
recursive) program \(\mathcal{P}\) coincides with the semantics of a
non-recursive program computable from \(\mathcal{P}\) and \(k\). 

\subsection{Index-bounded derivations}
The central notion of this section are {\em index-bounded
  derivations}, i.e.\ derivations in which each step has a {\em
  limited budget} of nonterminals. This notion is the key to our
underapproximation method.

 For a given integer constant \(k>0\), a word $u \in
(\Sigma \cup \Vars)^*$ is said to be of index $k$, if $u$ contains at
most $k$ occurrences of nonterminals (formally,
\(\len{{u}\proj_{\Vars}} \leq k\) ). A step \(u \Rightarrow v\) is
said to be \(k\)-indexed, denoted \(u \xRightarrow[(k)]{} v \), if and
only if both $u$ and $v$ are of index $k$. As expected, a step
sequence is \(k\)-indexed if all its steps are \(k\)-indexed. For
instance, both derivations from Ex. \ref{ex:depth-first} are of index
2.

\begin{lemma}\label{lemma:leibniz}
For every grammar \(G = \tuple{\Vars, \Sigma, \prod}\) the following
properties hold:
\begin{enumerate}[\upshape(\itshape 1\upshape)]
\item\label{item:derivdecomp1} \(\xRightarrow[(k)]{}^* ~\subseteq~
  \xRightarrow[(k+1)]{}^* \) for all \(k \geq 1\)

\item\label{item:derivdecomp1b} \(\xRightarrow{} ~=~
  \bigcup_{k=1}^\infty \xRightarrow[(k)]{}^*\)

\item\label{item:derivdecomp3} for all \(X,Y \in \Vars\), \(XY
  \xRightarrow[(k)]{}^* w \in \Sigma^*\) if and only if there exist
  \(w_1, w_2 \in \Sigma^*\), such that \(w = w_1 w_2\) and either:
  \begin{inparaenum}[\upshape(\itshape i\upshape)]
  \item \(X \xRightarrow[(k-1)]{}^* w_1 \) and \(Y \xRightarrow[(k)]{}^* w_2\), or 
  \item \(Y \xRightarrow[(k-1)]{}^* w_2\) and \(X \xRightarrow[(k)]{}^* w_1
    \).
  \end{inparaenum}
  \end{enumerate}
\label{lem:derivdecomp}
\end{lemma}
\begin{proof}
The proof of points (\ref{item:derivdecomp1}) and
(\ref{item:derivdecomp1b}) follow immediately from the definition of
$\xRightarrow[(k)]{}^*$. Let us now turn to the proof of point
(\ref{item:derivdecomp3}) (only if). First we define \(w_1\) and
\(w_2\). Consider the step sequence \(XY \xRightarrow[(k)]{}^* w \) and
look at the last step. It must be of the form \( u Z v
\xRightarrow[(k)]{}^* u y v = w\), where $u,v,y \in \Sigma^*$, and one
of the following must hold: \(Z\) has been generated from either \(X\)
or \(Y\). Suppose that \(Z\) stems from \(Y\) (the other case is
treated similarly). In this case, transitively remove from the
step-sequence all the steps transforming the rightmost occurrence of
\(Y\). Hence we obtain a step sequence \(XY \xRightarrow[(k)]{}^* w_1
Y\). Then \(w_2\) is the unique word satisfying \(w=w_1 w_2\). Since
\(XY \xRightarrow[(k)]{}^* w_1 Y\), by removing the occurrence of \(Y\)
in rightmost position at every step, we find that \(X
\xRightarrow[(k-1)]{}^* w_1\), and we are done. Having \(Z\) stemming
from \(X\) yields \(Y \xRightarrow[(k-1)]{}^* w_2\). For the proof of
the other direction (if) assuming \upshape(\itshape i\upshape) (the
other case is similar), it is easily seen that \(XY
\xRightarrow[(k)]{}^* w_1 Y \xRightarrow[(k)]{}^* w_1 w_2\). \qed
\end{proof}

The previous definitions extend naturally to bd-steps and bd-step
sequences, and we define \(\bdwords^{(k)} = \{w^{\vec{\beta}} \in
\big((\Vars \cup \Sigma)^{\nats}\big)^{*} \mid
\len{{w^{\vec{\beta}}}\proj_{\Vars^{\nats}}} \leq k\}\) the set of
bd-words with at most $k$ occurrences of nonterminals. We write the
fact that a bd-step sequence \(u^{\vec{\alpha}} \Rightarrow^*
v^{\vec{\beta}}\) is both \(k\)-indexed and depth-first as \(
u^{\vec{\alpha}} \xRightarrow[\df{k}]{}^* v^{\vec{\beta}}\). For any
symbol $X \in \Vars$ and constant \(k>0\), we define the languages:
\begin{align*}
L^{(k)}_{X}(G) &= \{w \in \Sigma^* \mid X \xRightarrow[(k)]{}^* w\} \\
\Gamma^{\df{k}}(G) &= \{\gamma\in\prod^* \mid \exists u^{\vec{\alpha}},v^{\vec{\beta}} \in \bdwords^{(k)} \colon 
u^{\vec{\alpha}} \xRightarrow[\df{k}]{\gamma} v^{\vec{\beta}}\} \enspace.
\end{align*}

\begin{example}\label{ex:betterthanboundedstack}(contd. from Ex.~\ref{ex:vpl})
Inspecting the grammar \(G_{\mathcal{P}}\) from Ex.\ref{ex:vpl}
reveals that \[L_{Q_1^{\mathit{init}}}(G_{\mathcal{P}}) = \{\left(\tau_1\calls\tau_2\right)^n \tau_4
  \left(\tau_2\rets\tau_3\right)^n \mid n \in \nat \}\enspace.\] For each value of \(n\) we
give a \(2\)-index derivation capturing the word: repeat \(n\) times
the steps 
\begin{align*}
& Q_1^{\mathit{init}} \xRightarrow{p_1^b p_2^c}
\tau_1 \calls\tau_2 Q_1^{\mathit{init}}\tau_2\rets Q_3
\xRightarrow{p_3^a} \tau_1 \calls\tau_2 Q_1^{\mathit{init}}\tau_2\rets \tau_3
	\shortintertext{followed by the step}	
& Q_1^{\mathit{init}}\stackrel{p_4^a}{\Longrightarrow} \tau_4\enspace .
\end{align*}
Therefore the \(2\)-index approximation of \(G_{\mathcal{P}}\) shows that
\(L_{Q_1^{\mathit{init}}}(G_{\mathcal{P}}) = L^{(2)}_{Q_1^{\mathit{init}}}(G_{\mathcal{P}})\).  
$\blacksquare$
\end{example}

\begin{example}(contd. from Ex. \ref{ex:depth-first})
For the grammar $G$ from Ex. \ref{ex:depth-first}, we obtain the
following control sets:
\[\begin{array}{rcl}
\Gamma^{\df{1}} & = & (Y,aY)^*(Y,\varepsilon) \cup (Z,Zb)^*(Z,\varepsilon)
\\ \Gamma^{\df{2}} & = &
(X,YZ)(Y,aY)^*(Y,\varepsilon)(Z,Zb)^*(Z,\varepsilon) \cup \\
&& (X,YZ)(Z,Zb)^*(Z,\varepsilon)(Y,aY)^*(Y,\varepsilon) \cup \Gamma^{\df{1}} \enspace.~ 
\blacksquare
\end{array}\] 
\end{example}

We recall a known result.
\begin{proposition}[\cite{Latteux78}]\label{df-are-complete}
For all $k \geq 1$, \(G=(\Vars,\Sigma,\prod)\) and \(X\in\Vars\), 
we have \(L_X^{(k)}(G)=\hat{L}_X(\Gamma^{\df{k}}, G)\). 
\end{proposition}

Finally, given \(k\geq 1\), we define the \emph{$k$-index semantics}
of $\mathcal{P}$ as $\sem{\mathcal{P}}^{(k)} = \langle
\sem{q_1}^{(k)}, \ldots, \sem{q_m}^{(k)} \rangle$, where
$\nf{\mathcal{P}} = \set{q_1, \ldots, q_m}$ and the $k$-index
semantics of a non-final control state $q \in \nf{P_i}$ of a procedure
$P_i$ of the program $\mathcal{P}$ is the relation $\sem{q} =
\sem{\mathcal{P}}_{q}^{(k)} \subseteq \zed^{\vec{x}_i} \times
\zed^{\vec{x}_i}$, defined as: \[\sem{\mathcal{P}}_q^{(k)} = \{\tuple{I
  \proj_{\vec{x}_i}, O \proj_{\vec{x}_i}} \mid I \cdot O \in
\textstyle{\bigcup_{Q \xRightarrow[\df{k}]{\gamma} w}} \sem{\gamma}
\} \enspace.\]

\subsection{Depth-first index-bounded control sets}
For a bd-word $w^{\vec{\alpha}}$, let
\[\age{w^{\vec{\alpha}}}=Pk_{\Vars^{\{\| \vec{\alpha} \|_{\max}
    \}}}(w^{\vec{\alpha}})\cdot Pk_{\Vars^{\{\| \vec{\alpha}
    \|_{\max}-1 \}}}(w^{\vec{\alpha}})\cdots
Pk_{\Vars^{\{0\}}}(w^{\vec{\alpha}})\enspace .\] Each symbol in
\(\age{w^{\vec{\alpha}}}\) is a \(\card{\Vars}\)-dimensional vector,
that is \(\age{w^{\vec{\alpha}}}\in
(\nats^{\card{\Vars}})^*\). Therefore with a slight abuse, we can view
each of these tuples as a multiset on \(\Vars\). Moreover, each tuple
$Pk_{\Vars^{\{i\}}}(w^{\vec{\alpha}})$ in $\age{w^{\vec{\alpha}}}$ is
the multiset of nonterminals that occur in $w^{\vec{\alpha}}$ with the
same birthdate $0 \leq i \leq \| \vec{\alpha} \|_{\max}$, and the
elements of $\age{w^{\vec{\alpha}}}$ are ordered in the reversed order
of their birthdates. For instance, the first tuple $Pk_{\Vars^{\{\|
    \vec{\alpha} \|_{\max} \}}}(w^{\vec{\alpha}})$ is the multiset of
the most recently added nonterminals. Notice that for each bd-word
\(u\) we have $\age{u} = \vec{0}$ if $u \in
(\Sigma^{\nat})^*$. Finally, let \(\vec{0}\) be the identity element
for concatenation, i.e.\ $\age{w^{\vec{\alpha}}} \cdot \vec{0} =
\vec{0} \cdot \age{w^{\vec{\alpha}}} = \age{w^{\vec{\alpha}}}$.

\begin{example}\label{ex:age-tuple}(contd. from Ex.~\ref{ex:depth-first})
  For the bd-step sequence $X^{\tuple{0}} \xRightarrow{} Y^{\tuple{1}}
  Z^{\tuple{1}} \xRightarrow{} a^{\tuple{2}}Y^{\tuple{2}}Z^{\tuple{1}}
  \xRightarrow{} a^{\tuple{2}}Y^{\tuple{2}}Z^{\tuple{3}}b^{\tuple{3}}$
  (Ex.~\ref{ex:depth-first}) we have $\age{X^{\tuple{0}}} =
  \{X\}$, $\age{Y^{\tuple{1}} Z^{\tuple{1}}} = \{Y,Z\}$,
  $\age{a^{\tuple{2}}Y^{\tuple{2}}Z^{\tuple{1}}} = \{Y\} \cdot \{Z\}$
  and $\age{a^{\tuple{2}}Y^{\tuple{2}}Z^{\tuple{3}}b^{\tuple{3}}} =
  \{Z\} \cdot \{Y\} \enspace.$ ~$\blacksquare$
\end{example}
The $\age{.}$ operator is lifted from bd-words to sets of bd-words,
i.e.\ subsets of $\bigl( (\Sigma \cup \Vars)^{\nats} \bigr)^*$. The set $\age{\bdwords^{(k)}}$
is of particular interest in the following developments. 
Next we define the graph $A^{\df{k}}(G) = \tuple{ \age{\bdwords^{(k)}},
  (\prod^*,\cdot) , \rightarrow}$, where \(\age{\bdwords^{(k)}}\) is the set of
vertices, \(\prod\) is the set of edge labels and \(\rightarrow\) is
the edge relation, defined as: $\widetilde{v} \xrightarrow{(Z,w)}
\widetilde{w}$ if and only if: 
\begin{itemize}
\item $\widetilde{v} = (\widetilde{v})_1 \cdot
\widetilde{v}_{t}$, where \((\widetilde{v})_1\in\nats^{\card{\Vars}}\), 
and $Pk_{\Vars}(Z) \leq (\widetilde{v})_1$, i.e.\ $Z$ occurs with
maximal birthdate $\widetilde{v}$, that is, it occurs in $(\widetilde{v})_1$, and
\item $\widetilde{w} = Pk_{\Vars}(w) \cdot \left((\widetilde{v})_1 - Pk_{\Vars}(Z)\right) 
\cdot \widetilde{v}_t$, i.e. $Z$ is removed from its multiset $(\widetilde{v})_1$, 
and the nonterminals of $w$ are added, with maximal birthdate to
obtain $\widetilde{w}$.
\end{itemize}
\begin{figure}
\centering
%
% start input /Users/pierreganty/counter-recursive/adfk.pdf_t
\begin{picture}(0,0)%
\includegraphics{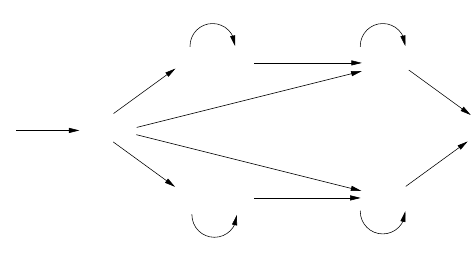}%
\end{picture}%
\setlength{\unitlength}{2368sp}%
\begingroup\makeatletter\ifx\SetFigFont\undefined%
\gdef\SetFigFont#1#2#3#4#5{%
  \reset@font\fontsize{#1}{#2pt}%
  \fontfamily{#3}\fontseries{#4}\fontshape{#5}%
  \selectfont}%
\fi\endgroup%
\begin{picture}(6330,3501)(2386,-3574)
\put(8701,-1861){\makebox(0,0)[b]{\smash{{\SetFigFont{8}{9.6}{\rmdefault}{\mddefault}{\updefault}{\color[rgb]{0,0,0}$\vec{0}$}%
}}}}
\put(3901,-1861){\makebox(0,0)[b]{\smash{{\SetFigFont{8}{9.6}{\rmdefault}{\mddefault}{\updefault}{\color[rgb]{0,0,0}$\{Y,Z\}$}%
}}}}
\put(2401,-1861){\makebox(0,0)[b]{\smash{{\SetFigFont{8}{9.6}{\rmdefault}{\mddefault}{\updefault}{\color[rgb]{0,0,0}$\{X\}$}%
}}}}
\put(7501,-961){\makebox(0,0)[b]{\smash{{\SetFigFont{8}{9.6}{\rmdefault}{\mddefault}{\updefault}{\color[rgb]{0,0,0}$\{Z\}$}%
}}}}
\put(7501,-2761){\makebox(0,0)[b]{\smash{{\SetFigFont{8}{9.6}{\rmdefault}{\mddefault}{\updefault}{\color[rgb]{0,0,0}$\{Y\}$}%
}}}}
\put(5251,-961){\makebox(0,0)[b]{\smash{{\SetFigFont{8}{9.6}{\rmdefault}{\mddefault}{\updefault}{\color[rgb]{0,0,0}$\{Y\}\{Z\}$}%
}}}}
\put(5251,-2761){\makebox(0,0)[b]{\smash{{\SetFigFont{8}{9.6}{\rmdefault}{\mddefault}{\updefault}{\color[rgb]{0,0,0}$\{Z\}\{Y\}$}%
}}}}
\put(2996,-1606){\makebox(0,0)[b]{\smash{{\SetFigFont{7}{8.4}{\rmdefault}{\mddefault}{\updefault}{\color[rgb]{0,0,0}$(X,YZ)$}%
}}}}
\put(4001,-1201){\makebox(0,0)[b]{\smash{{\SetFigFont{7}{8.4}{\rmdefault}{\mddefault}{\updefault}{\color[rgb]{0,0,0}$(Y,aY)$}%
}}}}
\put(3956,-2416){\makebox(0,0)[b]{\smash{{\SetFigFont{7}{8.4}{\rmdefault}{\mddefault}{\updefault}{\color[rgb]{0,0,0}$(Z,Zb)$}%
}}}}
\put(5201,-256){\makebox(0,0)[b]{\smash{{\SetFigFont{7}{8.4}{\rmdefault}{\mddefault}{\updefault}{\color[rgb]{0,0,0}$(Y,aY)$}%
}}}}
\put(6431,-751){\makebox(0,0)[b]{\smash{{\SetFigFont{7}{8.4}{\rmdefault}{\mddefault}{\updefault}{\color[rgb]{0,0,0}$(Y,\varepsilon)$}%
}}}}
\put(7556,-256){\makebox(0,0)[b]{\smash{{\SetFigFont{7}{8.4}{\rmdefault}{\mddefault}{\updefault}{\color[rgb]{0,0,0}$(Z,Zb)$}%
}}}}
\put(7541,-3496){\makebox(0,0)[b]{\smash{{\SetFigFont{7}{8.4}{\rmdefault}{\mddefault}{\updefault}{\color[rgb]{0,0,0}$(Y,aY)$}%
}}}}
\put(8606,-2416){\makebox(0,0)[b]{\smash{{\SetFigFont{7}{8.4}{\rmdefault}{\mddefault}{\updefault}{\color[rgb]{0,0,0}$(Y,\varepsilon)$}%
}}}}
\put(8636,-1216){\makebox(0,0)[b]{\smash{{\SetFigFont{7}{8.4}{\rmdefault}{\mddefault}{\updefault}{\color[rgb]{0,0,0}$(Z,\varepsilon)$}%
}}}}
\put(6431,-2971){\makebox(0,0)[b]{\smash{{\SetFigFont{7}{8.4}{\rmdefault}{\mddefault}{\updefault}{\color[rgb]{0,0,0}$(Z,\varepsilon)$}%
}}}}
\put(6476,-1456){\makebox(0,0)[b]{\smash{{\SetFigFont{7}{8.4}{\rmdefault}{\mddefault}{\updefault}{\color[rgb]{0,0,0}$(Y,\varepsilon)$}%
}}}}
\put(6446,-2251){\makebox(0,0)[b]{\smash{{\SetFigFont{7}{8.4}{\rmdefault}{\mddefault}{\updefault}{\color[rgb]{0,0,0}$(Z,\varepsilon)$}%
}}}}
\put(5251,-3436){\makebox(0,0)[b]{\smash{{\SetFigFont{7}{8.4}{\rmdefault}{\mddefault}{\updefault}{\color[rgb]{0,0,0}$(Z,Zb)$}%
}}}}
\end{picture}%
 % end input /Users/pierreganty/counter-recursive/adfk.pdf_t
 \caption{The graph \(A^{\df{k}}(G)\) for \(k\geq 2\) and for the grammar 
\(G\) of Ex.~\ref{ex:depth-first}}\label{fig:adfk}
\end{figure}
Next, define \(L(A^{\df{k}}(G)) = \{ \gamma\in\prod^* \mid
\widetilde{v}\xrightarrow{\gamma}\widetilde{w} \text{ in }
A^{\df{k}}(G) \} \).  For example, Fig.~\ref{fig:adfk} shows the
$A^{\df{k}}$ graph for the grammar $G$ from Ex.~\ref{ex:depth-first}.
The next lemma proves that the paths of $A^{\df{k}}(G)$ represent the
control words of the depth-first derivations of $G$ of index $k$. 
%An immediate consequence is that the set of $k$-index depth-first derivations $\Gamma^{\df{k}}(G)$ is a regular language (previously proved by Luker \cite{Luker80}). 
In the following, we omit the
argument $G$ from $\Gamma^{\df{k}}(G)$, or
$A^{\df{k}}(G)$, when it is clear from the context.

\begin{lemma}\label{fsa-dfk}
Given a grammar $G = \tuple{ \Vars, \Sigma, \prod }$, and $k > 0$, for
each $X \in \Vars$ and $\gamma \in \prod^*$, there exists a derivation
$X \xRightarrow[\df{k}]{\gamma} w$, for some $w \in \Sigma^*$, if and
only if $\age{X} \arrow{\gamma}{} \vec{0}$ in $A^{\df{k}}(G)$.
\end{lemma}
\begin{proof}
  \noindent % 
  ``$\Rightarrow$'' We shall prove the following more general
  statement. Let \(u^{\vec{\alpha}} \xRightarrow[\df{k}]{\gamma}
  w^{\vec{\beta}}\) be a \(k\)-indexed depth-first bd-step sequence.
  By induction on \(\len{\gamma}\geq 0\), we show the existence of a
  path \(\age{u^{\vec{\alpha}}}\xrightarrow{\gamma}
  \age{w^{\vec{\beta}}}\) in \(A^{\df{k}}\).

  For the base case \(\len{\gamma}=0\), we have \(u^{\vec{\alpha}} =
  w^{\vec{\beta}}\) which yields
  \(\age{u^{\vec{\alpha}}}=\age{w^{\vec{\beta}}}\) and since \(
  u^{\vec{\alpha}}\in\bdwords^{(k)}\) by definition of
  \(\Gamma^{\df{k}}\) we have that
  \(\age{u^{\vec{\alpha}}}\in\age{\bdwords^{(k)}}\) and we are done.

  For the induction step \(\len{\gamma}>0\), let \( v^{\vec{\eta}}
  \xRightarrow[\df{k}]{(Z,x)} w^{\vec{\beta}} \) be the last step of
  the sequence, for some \( (Z,x) \in \prod\), i.e.\ \(\gamma=\sigma
  \cdot (Z,x)\) with \(\sigma\in\prod^*\). By the induction
  hypothesis, \(A^{\df{k}}\) has a path \(\age{u^{\vec{\alpha}}}
  \xrightarrow{\sigma} \age{v^{\vec{\eta}}}\). Let
  \(\age{v^{\vec{\eta}}}= \vec{v}' \cdot \widetilde{v}_t \), where
  \(\vec{v}'= (\age{v^{\vec{\eta}}})_1 \in\nats^{|\Vars|}\), and
  \(\widetilde{v}_t\in (\nats^{|\Vars|})^* \) is a sequence of
  multisets of nonterminals. It remains to show that
	\(\age{w^{\vec{\beta}}}\in\bdwords^{(k)}\), \(Pk_{\Vars}(Z) \leq
  \vec{v}'\) and \(\age{w^{\vec{\beta}}} = Pk_{\Vars}(x) \cdot
  (\vec{v}' - Pk_{\Vars}(Z)) \cdot \widetilde{v}_t\) to conclude that
  \(A^{\df{k}}\) has an edge \(\age{v^{\vec{\eta}}}
  \xrightarrow{(Z,x)} \age{w^{\vec{\beta}}} \), hence a path \(
  \age{u^{\vec{\alpha}}}\xrightarrow{\gamma} \age{w^{\vec{\beta}}}\).

  Since \(v^{\vec{\eta}} \xRightarrow[\df{k}]{(Z,x)/j}
	w^{\vec{\beta}}\) for some \(1\leq j\leq \len{v^{\vec{\eta}}}\) we have that \( (v^{\vec{\eta}})_j=Z^{\tuple{i}}\)
  where \(i = \max\set{j\mid Pk_{\Vars^{\{j\}}}(v^{\vec{\eta}})\neq\vec{0}}\) and
	\(w^{\vec{\beta}} = (v^{\vec{\eta}})_1\ldots (v^{\vec{\eta}})_{j-1}\cdot
	x^{ {\langle}{\langle} \| \vec{\eta} \|_{\max}+1 {\rangle}{\rangle} }\cdot (v^{\vec{\eta}})_{j+1}\ldots
	(v^{\vec{\eta}})_{\len{v^{\vec{\eta}}}}\).
	It is easily seen that \(\|\vec{\beta}\|_{\max} = \| \vec{\eta} \|_{\max}+1\).
	Moreover, since \(i\) is the maximal birthdate among the non-terminals of $v^{\vec{\eta}}$, we have
  \(\age{v^{\vec{\eta}}} = Pk_{\Vars^{\{i\}}}(v^{\vec{\eta}}) \ldots
	Pk_{\Vars^{\{0\}}}(v^{\vec{\eta}})\), hence \(\vec{v}'=Pk_{\Vars^{\{i\}}}(v^{\vec{\eta}})\) and
  \(\widetilde{v}_t=Pk_{\Vars^{\{i-1\}}}(v^{\vec{\eta}})\ldots Pk_{\Vars^{\{0\}}}(v^{\vec{\eta}})\).
	Also we have \(Pk_{\Vars^{\{j\}}}(w^{\vec{\beta}}) = \vec{0}\) for all \(j, i<j\leq \| \vec{\eta} \|_{\max}\),
	\(Pk_{\Vars^{\{i\}}}(w^{\vec{\beta}}) = Pk_{\Vars^{\{i\}}}(v^{\vec{\eta}}) - Pk_{\Vars^{\{i\}}}(Z^{\tuple{i}})\) and 
	\(Pk_{\Vars^{\{\ell\}}}(w^{\vec{\beta}}) = Pk_{\Vars^{\{\ell\}}}(v^{\vec{\eta}})\) for all \(\ell, 0\leq \ell < i\). 
   Using the foregoing properties of \(w^{\vec{\beta}}\) the following equalities are easy to check:
  \[\begin{array}{rcll}
    &   & \age{w^{\vec{\beta}}} \\
    & = & Pk_{\Vars^{\{\| \vec{\eta} \|_{\max}+1\} }}(w^{\vec{\beta}}) \cdot Pk_{\Vars^{\{\| \vec{\eta} \|_{\max} \}}}(w^{\vec{\beta}})\ldots Pk_{\Vars^{\{0\}}}(w^{\vec{\beta}})\\
    & = & Pk_{\Vars^{\{\| \vec{\eta} \|_{\max}+1 }\}}(w^{\vec{\beta}}) \cdot Pk_{\Vars^{\{i\}}}(w^{\vec{\beta}})\cdot Pk_{\Vars^{\{i-1\}}}(w^{\vec{\beta}}) \ldots Pk_{\Vars^{\{0\}}}(w^{\vec{\beta}})\\
		& = & Pk_{\Vars}(x) \cdot Pk_{\Vars^{\{i\}}}(w^{\vec{\beta}})\cdot Pk_{\Vars^{\{i-1\}}}(w^{\vec{\beta}})\ldots Pk_{\Vars^{\{0\}}}(w^{\vec{\beta}})\\ %& Pk_{\Vars}(x)=Pk_{\Vars^{\{\| \vec{\eta} \|_{\max}+1 \}}}(x^{ {\langle}{\langle}\{\| \vec{\eta} \|_{\max}+1 \}{\rangle}{\rangle} })  \\
    & = & Pk_{\Vars}(x) \cdot (Pk_{\Vars^{\{i\}}}(v^{\vec{\eta}}) - Pk_{\Vars^{\{i\}}}(Z^{\tuple{i}})) \ldots Pk_{\Vars^{\{0\}}}(w^{\vec{\beta}}) \\
    & = & Pk_{\Vars}(x) \cdot (\vec{v}' - Pk_{\Vars}(Z))\cdot Pk_{\Vars^{\{i-1\}}}(w^{\vec{\beta}}) \ldots Pk_{\Vars^{\{0\}}}(w^{\vec{\beta}})\\ %& \vec{v}' = Pk_{\Vars^{\{i\}}}(v^{\vec{\eta}}) \\
    & = & Pk_{\Vars}(x) \cdot (\vec{v}' - Pk_{\Vars}(Z)) \cdot Pk_{\Vars^{\{i-1\}}}(v^{\vec{\eta}}) \ldots Pk_{\Vars^{\{0\}}}(v^{\vec{\eta}}) \\
    & = & Pk_{\Vars}(x) \cdot (\vec{v}' - Pk_{\Vars}(Z)) \cdot \widetilde{v}_t
  \end{array}\] 
  This concludes that \(\age{w^{\vec{\beta}}} = Pk_{\Vars}(z) \cdot
  (\vec{v}'-Pk_{\Vars}(Z))\cdot \widetilde{v}_t\), and since
  \(w^{\vec{\beta}}\in\bdwords^{(k)}\), we obtain that
  \(\age{v^{\vec{\eta}}} \xrightarrow{(Z,x)} \age{w^{\vec{\beta}}}\)
  is an edge in \(A^{\df{k}}\), and finally that
  \(\age{u^{\vec{\alpha}}}\xrightarrow{\gamma}\age{w^{\vec{\beta}}}\)
  is a path in $A^{\df{k}}$.

  \noindent ``$\Leftarrow$'' We prove a more general statement. Let
  $\widetilde{u} \xrightarrow{\gamma} \widetilde{w}$ be a path of
  $A^{\df{k}}(G)$. We show by induction on $\len{\gamma}$ that there
  exist bd-words $u^{\vec{\alpha}}, w^{\vec{\beta}} \in \bdwords^{(k)}$,
  such that \(\age{u^{\vec{\alpha}}}=\widetilde{u}\),
  \(\age{w^{\vec{\beta}}} = \widetilde{w}\), and \(u^{\vec{\alpha}}
  \xRightarrow[\df{k}]{\gamma} w^{\vec{\alpha}}\).
  
  % THIS IS SUBSUMED BY Uniqueness
  %
  % First, observe that, since every
  % production in $G$ has distinct nonterminals on the right-hand side,
  % for any two bd-annotated nonterminals $Z^{\tuple{ i }}$ and $Z^{\tuple{ j }}$ in
  % $w^{\vec{\alpha}}$, we have $i \neq j$, i.e.\ the nonterminals
  % belong to sets on different positions of
  % $\age{w^{\vec{\alpha}}}$.
  %\annot[pg]{Think about externalizing this fact?}
  % hence there exist a word $v \in (\Vars
  %\cup \Sigma)^*$ with at most $k$ nonterminals, such that
  %$\age{v^{\vec{\alpha}'}} = \age{w^{\vec{\alpha}}}$. Then the set of
  %ages are equal in $\alpha$ and $\alpha'$, and the set of
  %nonterminals with the same age are equal in $w^{\vec{\alpha}}$ and
  %$v^{\vec{\alpha}'}$. Since $w^{\vec{\alpha}}$ may not have two equal
  %nonterminals with the same age, it follows that $w^{\vec{\alpha}}$
  %has at most $k$ nonterminals.

  The base case $\len{\gamma}=0$ is trivial, because
  \(\widetilde{u}=\widetilde{w}\) and since
  \(\widetilde{u}\in\age{\bdwords^{(k)}}\) then there exists
  \(u^{\vec{\alpha}}\in \bdwords^{(k)}\) such that
  \(\age{u^{\vec{\alpha}}}=\widetilde{u}\), and we are done.

  For the induction step $\len{\gamma} > 0$, let $\gamma = \sigma
  \cdot (Z,x)$, for some production $(Z,x) \in \prod$ and $\sigma \in
  \prod^*$. By the induction hypothesis, there exist bd-words
  $u^{\vec{\alpha}}, v^{\vec{\eta}} \in \bdwords^{(k)}$ such that
  \(\widetilde{u}=\age{u^{\vec{\alpha}}} \xrightarrow{\sigma}
  \age{v^{\vec{\eta}}} \xrightarrow{(Z,x)} \widetilde{w}\) is a path
  in $A^{\df{k}}$, and \(u^{\vec{\alpha}} \xRightarrow[\df{k}]{\sigma}
  v^{\vec{\eta}}\) is a $k$-index bd-step sequence. By the definition
  of the edge relation in $A^{\df{k}}$, it follows that
  \(\age{v^{\vec{\eta}}}= Pk_{\Vars^{\{i\}}}(v^{\vec{\eta}}) \cdot
	\widetilde{v}_t\) where \(i=\max \set{j \mid Pk_{\Vars^{\{i\}}}(v^{\vec{\eta}}) \neq
	\vec{0}}\). Moreover, there exists \(j\), \(1\leq j \leq \len{v^{\vec{\eta}}}\) such that \( (v^{\vec{\eta}})_j = Z^{\tuple{i}}\) since 
	\(Pk_{\Vars}(Z)\leq Pk_{\Vars^{\{i\}}}(v^{\vec{\eta}})\). 
	Now define \(w^{\vec{\beta}}=(v^{\vec{\eta}})_{1}\ldots (v^{\vec{\eta}})_{j-1}\cdot x^{ {\langle}{\langle}{\| \vec{\eta} \|_{\max}+1 }{\rangle}{\rangle}}\cdot (v^{\vec{\eta}})_{j+1}\ldots (v^{\vec{\eta}})_{\len{v^{\vec{\eta}}}}\). It is routine to check \(v^{\vec{\eta}} \xRightarrow[\mathbf{df}]{(Z,x)/j} w^{\vec{\beta}}\) holds.
	Next we show, \(\widetilde{w}=\age{w^{\vec{\beta}}}\) which concludes the proof.
	\begin{align*}
    & \widetilde{w}\\ 
    =& Pk_{\Vars}(x) {\cdot} (Pk_{\Vars^{\{i\}}}(v^{\vec{\eta}}){-}Pk_{\Vars}(Z)) {\cdot} \widetilde{v}_t \\
    %& =  Pk_{\Vars}(x){\cdot} (Pk_{\Vars^{\{i\}}}(v^{\vec{\eta}}){-}Pk_{\Vars^{\{i\}}}(Z^{\tuple{i}})){\cdot}\widetilde{v}_t \\
		=&  Pk_{\Vars^{\{\| \vec{\eta} \|_{\max}+1 \}}}(x^{ {\langle}{\langle}{\| \vec{\eta} \|_{\max}+1 }{\rangle}{\rangle}}) {\cdot} 
    (Pk_{\Vars^{\{i\}}}(v^{\vec{\eta}}){-}Pk_{\Vars^{\{i\}}}(Z^{\tuple{i}})) {\cdot} \widetilde{v}_t \\
		=&  Pk_{\Vars^{\{\| \vec{\eta} \|_{\max}+1 \}}}(w^{\vec{\beta}}) {\cdot} 
    (Pk_{\Vars^{\{i\}}}(v^{\vec{\eta}}){-}Pk_{\Vars^{\{i\}}}(Z^{\tuple{i}})) {\cdot} \widetilde{v}_t \\
		\begin{split}	
			=& Pk_{\Vars^{\{\| \vec{\eta} \|_{\max}+1 \}}}(w^{\vec{\beta}}) {\cdot}(Pk_{\Vars^{\{i\}}}(v^{\vec{\eta}}) {-}Pk_{\Vars^{\{i\}}}(Z^{\tuple{i}})) {\cdot} \\
		&\quad\quad\quad\quad\quad\quad\quad\quad Pk_{\Vars^{\{i{-}1\}}}(v^{\vec{\eta}})\ldots Pk_{\Vars^{\{0\}}}(v^{\vec{\eta}})
	\end{split}\\
		%\intertext{Now define \(w^{\vec{\beta}}=(v^{\vec{\eta}})_{1}\ldots (v^{\vec{\eta}})_{j{-}1}{\cdot} x^{ {\langle}{\langle}{\| \vec{\eta} \|_{\max}+1 }{\rangle}{\rangle}}{\cdot} (v^{\vec{\eta}})_{j+1}\ldots (v^{\vec{\eta}})_{\len{v^{\vec{\eta}}}}\). On the one hand, we have \(v^{\vec{\eta}} \xRightarrow{(Z,x)/j} w^{\vec{\beta}}\) and on the other \(\age{w^{\vec{\beta}}}=\widetilde{w}\) since \(Pk_{\Vars^{\set{\ell}}}(w^{\vec{\beta}}) = Pk_{\Vars^{\set{\ell}}}(v^{\vec{\eta}})\) for \(\ell=0,\ldots,i{-}1\) and \(Pk_{\Vars^{\set{i}}}(w^{\vec{\beta}}) = (Pk_{\Vars^{\{i\}}}(v^{\vec{\eta}}){-}Pk_{\Vars^{\{i\}}}(Z^{\tuple{i}}))\).}
		\intertext{Since \(i=\max \set{j \mid Pk_{\Vars^{\{i\}}}(v^{\vec{\eta}}) \neq \vec{0}}\); \(Pk_{\Vars^{\set{\ell}}}(w^{\vec{\beta}}) = Pk_{\Vars^{\set{\ell}}}(v^{\vec{\eta}})\) for \(0\leq \ell < i\) and \(Pk_{\Vars^{\set{i}}}(w^{\vec{\beta}}) = Pk_{\Vars^{\{i\}}}(v^{\vec{\eta}}){-}Pk_{\Vars^{\{i\}}}(Z^{\tuple{i}})\) show that}
		=& Pk_{\Vars^{\{\| \vec{\eta} \|_{\max}+1 \}}}(w^{\vec{\beta}}) {\cdot} Pk_{\Vars^{\{i\}}}(w^{\vec{\beta}}) Pk_{\Vars^{\{i{-}1\}}}(w^{\vec{\beta}})\ldots Pk_{\Vars^{\{0\}}}(w^{\vec{\beta}})\\
		=& Pk_{\Vars^{\{\| \vec{\eta} \|_{\max}+1 \}}}(w^{\vec{\beta}}) {\cdot} Pk_{\Vars^{\{i\}}}(w^{\vec{\beta}}) Pk_{\Vars^{\{i{-}1\}}}(w^{\vec{\beta}})\ldots Pk_{\Vars^{\{0\}}}(w^{\vec{\beta}})\\
		=& \age{w^{\vec{\beta}}}
	\end{align*}
	\qed
\end{proof}
Consequently, we have the following (also proved in \cite{Luker80}):
\begin{corollary}\label{szilard-regular}
	For all $k \geq 1$, \(G=(\Vars,\Sigma,\prod)\) and \(X\in\Vars\), we have \(\Gamma^{\df{k}}\) is regular.
\end{corollary}

\subsection{Bounded-index Underapproximations of Control Structures}

We start describing our program transformation, from a recursive
program to a non-recursive program in which all computation traces
correspond to words generated by an index-bounded grammar. In the
beginning we choose to ignore the data manipulations, and give the
non-recursive program only in terms of transitions between control
locations and (non-recursive) calls. Then we show that the execution
traces of this new program match the depth-first index-bounded
derivations of the visibly pushdown grammar of the original program.

Let $\mathcal{P} = \langle P_1, \ldots, P_n \rangle$ be a recursive
program. For the moment, let us assume that $\mathcal{P}$ has no
(local) variables, and thus, all the labels of the internal
transitions, as well as all the call, return and frame relations are
trivially $\true$. As we did previously, we assume a fixed ordering
$q_1,\ldots,q_m$ on the set $\nf{\mathcal{P}}$ of non-final states of
$\mathcal{P}$. Let $G_{\mathcal{P}} = \tuple{\Vars, \widehat{\Theta},
  \prod}$ be the visibly pushdown grammar associated with
$\mathcal{P}$, where each non-final state $q$ of $\mathcal{P}$ is
associated a nonterminal $Q \in \Vars$. Then, for a given constant
$K>0$, we define a {\em non-recursive} program $\mathcal{H}^K$ that
captures only the traces of $\mathcal{P}$ corresponding to $K$-index
depth-first derivations of $G_{\mathcal{P}}$
(Algorithm~\ref{alg:ctl-query-nodots}). Formally, we define
$\mathcal{H}^K = \tuple{\mathit{query}^0, \mathit{query}^1, \ldots,
  \mathit{query}^K}$, i.e.\ the program is structured in $K+1$
procedures, such that:
\begin{compactitem}
\item $\mathit{query}^0$ consists of a single statement $\textbf{assume}\ \false$, i.e.\ no execution going through a call of
  $\mathit{query}^0$ is possible,
\item all executions of $\mathit{query}^k$, for each $1 \leq k \leq K$
  correspond to $k$-index depth-first derivations of
  $G_{\mathcal{P}}$.
\end{compactitem}
We distinguish between grammar productions of type (a) $Q \arrow{}{}
\tau$, (b) $Q \arrow{}{} \tau Q'$ and (c) $Q \arrow{}{} \calls\tau,
Q_j^{init} \tau\rets\; Q'$ (see Ex. \ref{ex:vpl}) of the visibly
pushdown grammar $G = \tuple{\Vars, \widehat{\Theta}, \prod}$. Since
$\Vars$ and $\widehat{\Theta}$ are finite sets, we associate each
nonterminal $Q \in \Vars$ an integer $1 \leq \mathcal{I}_Q \leq
\card{\Vars}$, each alphabet symbol $\tau \in \widehat{\Theta}$ an
integer $1 \leq \mathcal{I}_\tau \leq \card{\widehat{\Theta}}$, and
define the productions by the following formulae:
\begin{align*}
\pi_a(x,y) & \equiv  \bigvee_{(Q \arrow{}{} \tau) \in \prod} 
x = \mathcal{I}_Q \wedge y = \mathcal{I}_\tau \\
\pi_b(x,y,z) & \equiv  \bigvee_{(Q \arrow{}{} \tau Q') \in \prod} 
x = \mathcal{I}_Q \wedge y = \mathcal{I}_\tau \wedge z = \mathcal{I}_{Q'} \\
\begin{split}
\pi_c(x,y,z,t,s) & \equiv \bigvee_{(Q \arrow{}{} \calls\tau Q_j^{\mathit{init}} \tau\rets Q') \in \prod} 
\bigl( x = \mathcal{I}_Q \wedge y = \mathcal{I}_{\calls\tau} \wedge \\
&\quad\quad\quad z = \mathcal{I}_{Q_j^{\mathit{init}}} \wedge t = \mathcal{I}_{\tau\rets} \wedge s = \mathcal{I}_{Q'} \bigr)
\end{split}
\end{align*}
It is easy to see that the sizes of the $\pi_a$, $\pi_b$ and $\pi_c$
formulae are linear in the size of $\mathcal{P}$ (there is one
disjunctive clause per production of $G_{\mathcal{P}}$, and each such
production corresponds to a transition of $\mathcal{P}$). The
translation of $\mathcal{P}$ into $\mathcal{H}$ can hence be
implemented as a linear time source-to-source program transformation.

%
% start input /Users/pierreganty/counter-recursive/control-query.tex
% vim:ts=2:sw=2

\begin{algorithm}
\SetAlgoLined
\SetKw{goto}{goto}
\SetKw{assume}{assume}
\SetKw{havoc}{havoc}
\SetKw{Kor}{or}
\SetKwFunction{swap}{swap}
\Begin{
\textbf{var} \(\textsc{PC}\), \(\textsc{y}\), \(\textsc{z}\) \;
\nlset{$\textbf{asgn}^k_0$:} $\textsc{PC} \leftarrow \textsc{X}$ \;
\nlset{\(\textbf{start}^k\):} \goto $\textbf{prod}^k_a$ \Kor $\textbf{prod}^k_b$ \Kor $\textbf{prod}^k_c$ \label{ternary-non-det-choice} \;
\nlset{$\textbf{prod}^k_a$:} \assume \( \exists \tau \ldotp \pi_a(\textsc{PC},\tau) \)
\tcc*[r]{\(Q \rightarrow \tau\)}
\nlset{$\textbf{asgn}^k_a$:} \assume $\true$ \;
\Return \;
\nlset{$\textbf{prod}^k_b$:} \havoc(\textsc{y}) \;
\assume \(\exists \tau \ldotp \pi_b(\textsc{PC},\tau,\textsc{y}) \)
\tcc*[r]{\(Q \rightarrow \tau\, Q'\)}
\nlset{$\textbf{asgn}^k_b$:} \(\textsc{PC} \leftarrow \textsc{y}\) \; 
\goto \(\textbf{start}^k\) \;
\nlset{$\textbf{prod}^k_c$:} \havoc(\textsc{y},\textsc{z}) \;
\assume \(\exists \tau, \tau' \ldotp \pi_c(\textsc{PC}, \tau, \textsc{y}, \tau', \textsc{z})\) 
\tcc*[r]{\(Q {\rightarrow} \tau\, Q_j^{\mathit{init}} \tau' Q'\)}
\nlset{$\textbf{ndet}^k$:} \goto \(\textbf{swap}^k\) \Kor $\textbf{asgn}^k_c$ \label{rod-ord-choice}\;
\nlset{$\textbf{swap}^k$:} \swap(\textsc{y}, \textsc{z}) \;
\nlset{$\textbf{asgn}^k_c$:} \(\textsc{PC} \leftarrow \textsc{z}\) \;
\(\mathit{query}^{k-1}(\textsc{y})\) \;
\goto \(\textbf{start}^k\)\;
}
\caption{\(\textbf{proc}\ \mathit{query}^k(\textsc{X})\) 
for \(1 \leq k \leq K\)\label{alg:ctl-query-nodots}}
\end{algorithm}%
 % end input /Users/pierreganty/counter-recursive/control-query.tex
 
Next, we show a mapping from the paths of \(A^{\df{k}}\) onto the feasible interprocedural valid paths of \(query^k\).
To relate these paths, we need to introduce the notion of gsm mappings.
\begin{definition}[\cite{ginsburg}]\label{def:gsm}
	A \emph{generalized sequential machine}, abbreviated gsm, is a \(6\)-tuple 
	\(S= \tuple{K,\Sigma,\Delta,\delta,\lambda,q_1}\) where
	\begin{inparaenum}[\upshape(\itshape 1\upshape)]
	\item \(K\) is a finite non-empty set of \emph{states};
	\item \(\Sigma\) and \(\Delta\) respectively are \emph{input} and \emph{output} alphabet;
	\item \(\delta\) and \(\lambda\) are mappings from \(K\times \Sigma\) into \(K\) and \(\Delta^*\), respectively;
	\item \(q_1\in K\) is the \emph{start} state.
	\end{inparaenum}
The functions \(\delta\) and \(\lambda\) are extended by induction to \(K\times \Sigma^*\) by defining
for every state \(q\), \(x\in\Sigma^*\), and \(y\in \Sigma\):
\begin{compactitem}
\item \(\delta(q,\varepsilon)=q\) and \(\lambda(q,\varepsilon)=\varepsilon\).
\item \(\delta(q,xy)=\delta(\delta(q,x),y)\) and \(\lambda(q,xy)=\lambda(q,x)\lambda(\delta(q,x),y)\).
\end{compactitem}
%It is readily seen that the second item holds for all words \(x\) and \(y\) in \(\Sigma^*\).
The operation defined by \(S(x)=\lambda(q_1,x)\) for each \(x\in\Sigma^*\) is called a \emph{gsm mapping}.
\end{definition}
We define the gsm \(\mathit{SC}_Q^k = \tuple{\age{\Upsilon^{(k)}} \cup \{ \mathit{sink} \} , \prod, \mathcal{L}, \delta, \lambda, \age{Q}}\) upon \(A^{\df{k}}\), where \(\mathcal{L}\) denotes the statement labels found in \(\mathit{query}^{0},\ldots,\mathit{query}^{k}\); and the mappings \(\delta\) and \(\lambda\) are given by the rules of Fig.~\ref{fig:lambda}.
\begin{figure*}
\raggedright 
Given \(s \in \age{\Upsilon^{(k)}} \cup \{ \mathit{sink} \}\) and \(p\in\prod\) define \(\delta( s, p ) = s' \) if \( s \xrightarrow{p} s'\) holds in \(A^{\df{k}}\) for some \(s'\), otherwise 
(\( s \xrightarrow{p} s'\) holds for no \(s'\)) then \(\delta( s, p ) = \mathit{sink} \).
The output mapping \(\lambda\) is defined as follows:
\begin{compactenum}
\item \(\lambda ( \set{X}\cdot \widetilde{v}, (X,\tau)	) = \mathbf{start}^{k-\len{\widetilde{v}}} \mathbf{prod}^{k-\len{\widetilde{v}}}_a \mathbf{asgn}^{k-\len{\widetilde{v}}}_a \mathbf{start}^{k-\len{\widetilde{v}}+1}\), if \(\widetilde{v}\neq\varepsilon\);  
\item \(\lambda ( \set{X}, (X,\tau)) = \mathbf{start}^{k}\; \mathbf{prod}^{k}_a\; \mathbf{asgn}^{k}_a\)
\item \(\lambda ( \set{X}\cdot \widetilde{v}, (X,\tau\, X')	)  = \mathbf{start}^{k-\len{\widetilde{v}}} \mathbf{prod}^{k-\len{\widetilde{v}}}_b \mathbf{asgn}^{k-\len{\widetilde{v}}}_b \mathbf{start}^{k-\len{\widetilde{v}}}\)
\item \(\lambda ( \set{X}\cdot \widetilde{v}, (X,\tau\, X_1\, \tau'\, X_2)	)  = \mathbf{start}^{k-\len{\widetilde{v}}} \mathbf{prod}^{k-\len{\widetilde{v}}}_c \mathbf{ndet}^{k-\len{\widetilde{v}}} \)
%\item\label{point1} \(\lambda ( \set{Q^{init},Q'}\cdot \widetilde{v}, (Q^{init},\tau)	)  = \mathbf{asgn}^{k-\len{\widetilde{v}}}_c \mathbf{asgn}^{k-\len{\widetilde{v}}-1}_0 \mathbf{start}^{k-\len{\widetilde{v}}-1} \mathbf{prod}^{k-\len{\widetilde{v}}-1}_a \mathbf{asgn}^{k-\len{\widetilde{v}}-1}_a \mathbf{start}^{k-\len{\widetilde{v}}}\)	
\item\label{point2} \(\lambda ( \set{Q^{init},Q'}\cdot \widetilde{v}, (Q^{init},\tau\, Q'')	)  = \mathbf{asgn}^{k-\len{\widetilde{v}}}_c \mathbf{asgn}^{k-\len{\widetilde{v}}-1}_0 \mathbf{start}^{k-\len{\widetilde{v}}-1}\mathbf{prod}^{k-\len{\widetilde{v}}-1}_b \mathbf{asgn}^{k-\len{\widetilde{v}}-1}_b \mathbf{start}^{k-\len{\widetilde{v}}-1}\)
% Should not be legal
%\item \(\lambda ( \set{Q^{init},Q'}\cdot \widetilde{v}, (X,\tau\, X_1\, \tau'\, X_2)	)  = \)
\item \(\lambda ( \set{Q^{init},Q'}\cdot \widetilde{v}, (Q',\tau)) = \mathbf{swap}^{k-\len{\widetilde{v}}} \mathbf{asgn}^{k-\len{\widetilde{v}}}_c \mathbf{asgn}^{k-\len{\widetilde{v}}-1}_0 \mathbf{start}^{k-\len{\widetilde{v}}-1} \mathbf{prod}^{k-\len{\widetilde{v}}-1}_a \mathbf{asgn}^{k-\len{\widetilde{v}}-1}_a \mathbf{start}^{k-\len{\widetilde{v}}}\)
%\lambda ( \set{Q^{init},Q'}\cdot \widetilde{v}, (Q^{init},\tau)) \)
% \) is the same as \ref{point1} but prefixed with \(\mathbf{swap}^{k-\len{\widetilde{v}}} \)
\item \(\lambda ( \set{Q^{init},Q'}\cdot \widetilde{v}, (Q',\tau\, Q'')) = \mathbf{swap}^{k-\len{\widetilde{v}}} \mathbf{asgn}^{k-\len{\widetilde{v}}}_c \mathbf{asgn}^{k-\len{\widetilde{v}}-1}_0 \mathbf{start}^{k-\len{\widetilde{v}}-1}\mathbf{prod}^{k-\len{\widetilde{v}}-1}_b \mathbf{asgn}^{k-\len{\widetilde{v}}-1}_b \mathbf{start}^{k-\len{\widetilde{v}}-1}\)% \lambda ( \set{Q^{init},Q'}\cdot \widetilde{v}, (Q^{init},\tau\, Q''))\) 
%\lambda ( \set{Q^{init},Q'}\cdot \widetilde{v}, (Q^{init},\tau\, Q'')	)\) 
% is the same as \ref{point2} but prefixed with \(\mathbf{swap}^{k-\len{\widetilde{v}}} \)
\item \(\lambda ( s, p) = \bot \),  for all \(s\) and \(p\), such that \(\delta(s,p)=\mathit{sink}\) holds.
\end{compactenum}
\caption{Definition of the mappings \(\delta\) and \(\lambda\) for \(\mathit{SC}_Q^k\).\label{fig:lambda}}
\end{figure*}

\begin{lemma}\label{query-fsa}
For a visibly pushdown grammar $G = \tuple{ \Vars, \widehat{\Theta},
  \prod }$, and $k > 0$, for each $Q \in \Vars$ 
	the set of feasible interprocedural valid paths of \(query^{k}(Q)\) coincides with the set \(\{ \mathit{SC}_Q^k(\gamma) \mid \age{Q} \xrightarrow{\gamma} \vec{0} \text{ in } A^{\df{k}} \} \).
% and $\gamma \in \prod^*$, Algorithm~\ref{alg:ctl-query-nodots} has a trace $Q \arrow{\gamma}{} \varepsilon$ if and only if $\age{Q} \arrow{\gamma}{} \vec{0}$ in $A^{\df{k}}(G)$.
\end{lemma}
\begin{proof}
The feasible interprocedural valid paths of \(query^{k}(Q)\) at Algorithm~\ref{alg:ctl-query-nodots} matches sequences of
the form $\sigma_0 \arrow{\delta_0}{} \sigma_1 \arrow{\delta_1}{} \ldots
\arrow{\delta_{n-1}}{} \sigma_n$, where each $\sigma_i \in \Vars^*$ is a {\em
stack}, i.e.\ a possibly empty sequence of frames each containing a snapshot of the value of the local variable $\textsc{PC}$,
$\delta_i \in \prod$ are productions of $G$. The
sequence of stacks $\sigma_0, \sigma_1, \ldots, \sigma_n$ are snapshots of
values of the local variable $\textsc{PC}$ between two consecutive visit to a \textbf{start} label or
between the last visit to a \textbf{start} label and the last \textbf{return}.
Instances of such consecutive visits are given by \(\textbf{start}^k\), $\textbf{prod}^k_a$, $\textbf{asgn}^k_a$; or
$\textbf{start}^k$, $\textbf{prod}^k_a$, $\textbf{asgn}^k_a$, $\textbf{return}$, $\textbf{start}^{k+1}$ (when returning from a previous call);
or $\textbf{start}^k$, $\textbf{prod}^k_c$, $\textbf{ndet}^k$, $\textbf{swap}^k$, $\textbf{asgn}^k_c$, \(\textbf{start}^{k-1}\) (immediately after entering the call \(query^{k-1}\)).

When Algorithm~\ref{alg:ctl-query-nodots} is started with a call to
$query^k(Q)$, the first stack in the trace is $Q$. The
set of stack sequences are generated by a labelled graph defined by the following rules, where 
the stack on both sides of each rule are words \(w\in\Vars^*\) such that \(\len{w} \leq k\). 
\begin{compactenum}[(a)]
\item $Q \cdot \sigma \arrow{(Q,\tau)}{} \sigma$%, i.e.\ an execution of $\mathit{query}^i$ through $\mathbf{start}, \mathbf{prod}_a, \mathbf{asgn}_a$ and \(\mathbf{start}\) if \(\sigma\neq\varepsilon\);  $\mathbf{start}, \mathbf{prod}_a, \mathbf{asgn}_a$ otherwise.
\item $Q \cdot \sigma \arrow{(Q,\tau Q')}{} Q' \cdot \sigma$%, i.e.\ an execution of $\mathit{query}^i$ through $\mathbf{start}, \mathbf{prod}_b, \mathbf{asgn}_b, \mathbf{start}$.
\item we have either \begin{inparaenum}[\upshape(\itshape i\upshape)] 
\item $Q \cdot \sigma \arrow{(Q,\calls\tau Q' \tau\rets Q'')}{} Q' \cdot Q'' \cdot \sigma$, or %, i.e.\ an execution of $\mathit{query}^i$ through $\mathbf{start}, \mathbf{prod}_c, \mathbf{ndet}, \mathbf{asgn}_c, \mathbf{start}$ (skipping $\textbf{swap}$)
\item $Q \cdot \sigma \arrow{(Q,\calls\tau Q' \tau\rets Q'')}{} Q'' \cdot Q' \cdot \sigma$%, i.e.\ an execution of $\mathit{query}^i$ through $\mathbf{start}, \mathbf{prod}_c, \mathbf{ndet}, \mathbf{swap}, \mathbf{asgn}_c, \mathbf{start}$.
\end{inparaenum}
\end{compactenum}
Following the previous definition, we find that the set of sequences of control labels \(\{ \mathit{SC}_Q^k( \gamma ) \mid Q \xrightarrow{\gamma} \varepsilon \} \) coincides with the feasible interprocedural valid path of \(query^k(Q)\).

Next we show that $Q \arrow{\gamma}{} \varepsilon$ is a valid stack sequence of $\mathit{query}^k(Q)$ if and only if $\age{Q}
\arrow{\gamma}{} \vec{0}$ in $A^{\df{k}}(G)$.
For this, consider the following relation between the stacks $\sigma \in
\Vars^*$ such that \(\len{\sigma}\leq k\) and words $\widetilde{w} \in
\age{\bdwords^{(k)}}$: we write $\sigma \llcurly \widetilde{w}$ if and
only if exactly one of the following holds:
\begin{compactenum}[(1)]
\item\label{query-inv:it1} $\len{\sigma} = \len{\widetilde{w}}$ and,
	for all $1 \leq i \leq \len{\widetilde{w}}\colon \{(\sigma)_i\} = (\widetilde{w})_i$, or
\item\label{query-inv:it2} $\len{\sigma} = \len{\widetilde{w}} + 1$,
	$(\widetilde{w})_{1} = \{(\sigma)_1,(\sigma)_2\}$, and for all $1
	< i \leq \len{\widetilde{w}}\colon \{(\sigma)_{i+1}\} = (\widetilde{w})_i$.
\end{compactenum}
The proof goes by induction and shows the following stronger statement relating the reachable stacks and the states of \(A^{\df{k}}\) reachable from \(\age{Q}\):
for any stack sequence $Q \arrow{\gamma}{} \sigma$, there exists a path $\age{Q} \arrow{\gamma}{}
\widetilde{w}$ in $A^{\df{k}}$, such that $\sigma \llcurly
\widetilde{w}$, and vice versa. 

By putting together the previous result about the feasible interprocedural valid paths of \(query^k(Q)\) we find that they coincide with
the set \(  \{ \mathit{SC}_Q^k( \gamma ) \mid \age{Q} \xrightarrow{\gamma} \vec{0} \text{ in } A^{\df{k}} \}\).\qed
\end{proof}

\subsection{Bounded-index Underapproximations of Programs}

Algorithm \ref{alg:ctl-query-nodots} implements the transformation of
the control structure of a recursive program $\mathcal{P}$ into a
non-recursive program $\mathcal{H}^K = \tuple{\mathit{query}^0,
  \ldots, \mathit{query}^K}$, which simulates its $K$-index
derivations (actually, the control words thereof). In this section we
extend this construction to programs with integer variables and data
manipulations (Algorithm \ref{alg:query-nodots}), by defining a set of
procedures $\mathit{query}^k$, for all $0 \leq k \leq K$, such that
each procedure $\mathit{query}^k$ has five sets of local variables,
all of the same cardinality as $\vec{x}$: two sets, named $\vec{x}_I$
and $\vec{x}_O$, are used as input variables, whereas the other three
sets, named $\vec{x}_J,\vec{x}_K$ and $\vec{x}_L$ are used locally by
$\mathit{query}^k$. Besides, each \(\mathit{query}^k\) has local
variables called \textsc{PC}, \(\tau\), \textsc{y}, \textsc{z} and
input variable \(X\). There are no output variables in
$\mathit{query}^k$. Let $\mathcal{V}_{query}^k$ denote the tuple of
local variables of $\mathit{query}^k$, and let
$\mathcal{V}^K_{\mathcal{H}} = \mathcal{V}_{query}^1 \cdot \ldots
\cdot \mathcal{V}_{query}^K$ be the tuple of all variables of
$\mathcal{H}^K$.

%% \noindent %
%% Observe that each procedure $\mathit{query}^k$ only calls procedures
%% $\mathit{query}^{k-1}$, hence that the call stack contains no more than \(K\) frames,
%% namely, there is no recursion. \changed[pg][candidate for removal]{Given that $\mathit{query}^k$ has two copies of $\vec{x}$ as input variables,
%% and no output variables, the input output semantics
%% $\sem{\mathcal{H}}^{i/o}_{query^k_Q} \subseteq
%% \zed^{\vec{x}\times\vec{x}}$ is a set of tuples, rather than a
%% (binary) relation. We denote $\mathit{pre}(\mathit{query}^k_Q) = \{I \cdot O \in
%% \zed^{\vec{x} \times \vec{x}} \mid \mathit{query}^k_Q(I,O) ~\mbox{returns with
%% empty stack}\}$. Clearly $\mathit{pre}(\mathit{query}^k_Q) =
%% \sem{\mathcal{H}}^{i/o}_{query^k_Q}$.}

%
% start input /Users/pierreganty/counter-recursive/data-query.tex
% vim:ts=2:sw=2

\begin{algorithm}
\SetAlgoLined
\SetKw{goto}{goto}
\SetKw{assume}{assume}
\SetKw{havoc}{havoc}
\SetKw{Kor}{or}
\SetKwFunction{swap}{swap}
\Begin{
\textbf{var} \(\vec{x}_J, \vec{x}_K, \vec{x}_L\)\;
\textbf{var} \(\textsc{PC}\), \(\tau, \textsc{y}\), \(\textsc{z}\) \;
\nlset{$\textbf{asgn}^k_0$:}$\textsc{PC} \leftarrow X$ \;
\nlset{\(\textbf{start}^k\):} \goto $\textbf{prod}^k_a$ \Kor $\textbf{prod}^k_b$ \Kor $\textbf{prod}^k_c$ \label{ternary-non-det-choice-data} \;
\nlset{$\textbf{prod}^k_a$:} \havoc(\(\tau\))\;
\assume \( \pi_a(\textsc{PC},\tau) \)\tcc*[r]{ \(Q \rightarrow \tau \)}
\nlset{$\textbf{asgn}^k_a$:} \assume $\rho_{\tau}(\vec{x}_I,\vec{x}_O)$\; 
\Return\;
\nlset{$\textbf{prod}^k_b$:} \havoc(\(\tau, \textsc{y}\))\;
\assume \(\pi_b(\textsc{PC},\tau,\textsc{y})\)\tcc*[r]{\(Q \rightarrow \tau\, Q'\)}
\havoc($\vec{x}_J$)\;
\assume $\rho_\tau(\vec{x}_I,\vec{x}_J)$\;
\(\vec{x}_I \leftarrow \vec{x}_J\)\;
\nlset{$\textbf{asgn}^k_b$:}\(\textsc{PC} \leftarrow \textsc{y}\)\;
\goto \(\textbf{start}^k\) \;
\nlset{$\textbf{prod}^k_c$:} \havoc(\(\tau,\textsc{y},\textsc{z}\))\;
\assume \(\pi_c(\textsc{PC}, \calls\tau, \textsc{y}, \tau\rets, \textsc{z})\)\tcc*[r]{ \(Q {\rightarrow} \calls\tau Q_j^{\mathit{init}} \tau\rets\, Q'\)}
\havoc($\vec{x}_J,\vec{x}_K,\vec{x}_L$)\;
\assume $\rho_{\calls\tau}(\vec{x}_I,\vec{x}_J)$ \tcc*[r]{call relation}
\assume $\rho_{\tau\rets}(\vec{x}_K,\vec{x}_L)$ \tcc*[r]{return relation}
\assume $\phi_{\tau}(\vec{x}_I,\vec{x}_L)$ \tcc*[r]{frame relation}
\nlset{$\textbf{ndet}^k$:} \goto \(\textbf{swap}^k\) \Kor $\textbf{asgn}^k_c$ \label{rod-ord-choice-data} \;
\nlset{$\textbf{swap}^k$:}\swap{\textsc{y}, \textsc{z}}\;
\swap{\(\vec{x}_{J}, \vec{x}_{L}\)}\;
\swap{\(\vec{x}_{K}, \vec{x}_{O}\)}\;
\nlset{$\textbf{asgn}^k_c$:}\(\vec{x}_I \leftarrow \vec{x}_L\)\;
\(\textsc{PC} \leftarrow \textsc{z}\)\;
\(\mathit{query}^{k-1}(\textsc{y}, \vec{x}_J, \vec{x}_K )\)\;
\goto \(\textbf{start}^k\)\;
}
\caption{\(\textbf{proc}\ \mathit{query}^k(\textsc{X},\vec{x}_I,\vec{x}_O)\) for \(1 \leq k \leq K\)\label{alg:query-nodots}}
\end{algorithm}%
 % end input /Users/pierreganty/counter-recursive/data-query.tex
 
%% Inspection of the code of \(\mathcal{H}\) reveals that
%% \(\mathcal{H}\) simulates \(k\)-index depth first derivations of
%% \(G_{\mathcal{P}}\) and interprets the statements of \(\mathcal{P}\)
%% on its local variables while applying derivation steps. The main
%% difference with the normal execution of $\mathcal{P}$ is that
%% $\mathcal{H}$ may interpret a procedure call and its continuation in
%% an order which differs from the expected one.

%
% start input /Users/pierreganty/counter-recursive/query_example.tex
\begin{figure*}
\begin{example}\label{ex:runofquery}
Let us consider an execution of \(\mathit{query}\) for the call
\(\mathit{query}^2(Q_1^{\mathit{init}}, \begin{psmallmatrix} 1 &
  0 \end{psmallmatrix}, \begin{psmallmatrix} 1 &
  2 \end{psmallmatrix})\) following {\footnotesize \(Q_1^{init}
  \stackrel{p_1^b p_2^c}{\Longrightarrow} \tau_1
  \calls\tau_2 Q_1^{init} \tau_2\rets Q_3
  \stackrel{p_3^a}{\Longrightarrow}\tau_1 \calls\tau_2
  Q_1^{init} \tau_2\rets \tau_3\stackrel{p_4^a}{\Longrightarrow}\tau_1 \calls\tau_2
  \tau_4 \tau_2\rets \tau_4 \)}. In the table
  below, the first row (labelled \(\mathrm{PC}\)) gives the value of local variable \(\mathrm{PC}\)
 when control hits the labelled statement given at the
second row (labelled \(ip\)). The third row (labelled
\(\vec{x}_I/\vec{x}_O\)) represents the content of the two
arrays. \(\vec{x}_I/\vec{x}_O=\begin{psmallmatrix} a &
b \end{psmallmatrix}\begin{psmallmatrix} c & d \end{psmallmatrix}\)
says that, in \(\vec{x}_I\), \(x\) has value \(a\) and \(z\) has value
\(b\); in \(\vec{x}_O\), \(x\) has value \(c\) and \(z\) has value
\(d\). 

\noindent
\begin{minipage}[htb]{\linewidth}
\[\begin{array}{l|cccccc}
  \mathrm{PC} & Q_1^{\mathit{init}} & - & Q_2 & - & - \\
  ip & \mathbf{start}^2 & \mathbf{prod}^2_b\,(p_1^b) & \mathbf{start}^2 & \mathbf{prod}^2_c\,(p_2^c) & \mathbf{swap}^2 \\
	\vec{x}_I/\vec{x}_O & \begin{psmallmatrix} 1 & 0 \end{psmallmatrix}\begin{psmallmatrix} 1 & 2 \end{psmallmatrix} & 
		\begin{psmallmatrix} 1 & 0 \end{psmallmatrix}\begin{psmallmatrix} 1 & 2 \end{psmallmatrix}&\begin{psmallmatrix} 1 & 0 \end{psmallmatrix}\begin{psmallmatrix} 1 & 2 \end{psmallmatrix}&\begin{psmallmatrix} 1 & 0 \end{psmallmatrix}\begin{psmallmatrix} 1 & 2  \end{psmallmatrix}&\begin{psmallmatrix} 1 & 0 \end{psmallmatrix}\begin{psmallmatrix} 1 & 2 \end{psmallmatrix}		\\
		\hline
		\mathrm{PC} & \multicolumn{1}{||c}{Q_3} & \multicolumn{1}{c||}{-} & Q_1^{\mathit{init}} & -\\
		ip & \multicolumn{1}{||c}{\mathbf{start}^1} & \multicolumn{1}{c||}{\mathbf{prod}^1_a\,(p_3^a)}  &\mathbf{start}^2&  \mathbf{prod}^2_a\,(p_4^a) \\
\vec{x}_I/\vec{x}_O&
\multicolumn{1}{||c}{\begin{psmallmatrix} 1 & 0 \end{psmallmatrix}\begin{psmallmatrix} 1 & 2 \end{psmallmatrix}}&\multicolumn{1}{c||}{\begin{psmallmatrix} 1 & 0 \end{psmallmatrix}\begin{psmallmatrix} 1 & 2 \end{psmallmatrix}}&\begin{psmallmatrix} 0 & 0 \end{psmallmatrix}\begin{psmallmatrix} 42 & 0 \end{psmallmatrix} &\begin{psmallmatrix} 0 & 0 \end{psmallmatrix}\begin{psmallmatrix} 42 & 0 \end{psmallmatrix}\\
\end{array}\]
\end{minipage}
The execution of 
\(\mathit{query}^2(Q_1^{\mathit{init}}, \begin{psmallmatrix} 1 &
  0 \end{psmallmatrix}, \begin{psmallmatrix} 1 &
  2 \end{psmallmatrix})\) starts on row 1,
column 1 and proceeds until the call to \(\mathit{query}^1(Q_3, \begin{psmallmatrix} 1 &
  0 \end{psmallmatrix}, \begin{psmallmatrix} 1 &
  2 \end{psmallmatrix}) \) at
row 2, column 1 (the out of order case). The latter ends at row 2,
column 2, where the execution of \(\mathit{query}^2(Q_1^{\mathit{init}}, \begin{psmallmatrix} 1 &
  0 \end{psmallmatrix}, \begin{psmallmatrix} 1 &
  2 \end{psmallmatrix})\)
resumes. Since the execution is out of order, and the previous
\(\mathbf{havoc}(\vec{x}_J,\vec{x}_K,\vec{x}_L)\) results into
\(\vec{x}_J=\begin{psmallmatrix} 0 & 0 \end{psmallmatrix}\),
\(\vec{x}_K=\begin{psmallmatrix} 42 & 0 \end{psmallmatrix}\) and
\(\vec{x}_L=\begin{psmallmatrix} 1 & 0 \end{psmallmatrix}\) (this
choice complies with the call relation), the values of
\(\vec{x}_I/\vec{x}_O\) are updated to \(\begin{psmallmatrix} 0 &
0 \end{psmallmatrix}/\begin{psmallmatrix} 42 &
0 \end{psmallmatrix}\). 
$\blacksquare$
\end{example}
\end{figure*}
 % end input /Users/pierreganty/counter-recursive/query_example.tex
 
For two tuples of variables $\vec{x}$ and $\vec{y}$ of equal length,
and a valuation $\nu \in \zed^{\vec{x}}$, we denote by
$\nu[\vec{y}/\vec{x}]$ the valuation that maps $(\vec{y})_i$ into
$(\nu(\vec{x}))_i$, for all $1 \leq i \leq \len{\vec{x}}$. The
following lemma is needed in the proof of Thm. \ref{thm:query}.

\begin{lemma}\label{lemma:derivation-semantics-gsm}
  Let $G_{\mathcal{P}} = \tuple{\Vars,\widehat{\Theta},\prod}$ be a
  visibly pushdown grammar for a program $\mathcal{P} = \tuple{P_1,
    \ldots, P_n}$, let $\vec{x} = \vec{x}_1 \cdot \ldots \cdot
  \vec{x}_n$ be the tuple of variables in $\mathcal{P}$, and let
  $\mathcal{H}^K = \tuple{\mathit{query}^0, \ldots, \mathit{query}^K}$
  be the program defined by Algorithm \ref{alg:query-nodots}. Given a
  nonterminal $Q \in\Vars$, corresponding to a non-final control state
  $q \in \nf{\mathcal{P}}$, $\gamma\in\prod^*$,
  $w\in\widehat{\Theta}^*$, and $1 \leq k \leq K$, such that $Q
  \xRightarrow[\df{k}]{\gamma} w$, we have: \[\sem{\gamma} =
  \set{\left(I \proj_{\vec{x}_I \cdot
      \vec{x}_O}\right)[\vec{x}\cdot\vec{x}/\vec{x}_I \cdot \vec{x}_O]
    \mid I \cdot O \in \sem{\mathit{SC}^k_Q(\gamma)}}\] where
  $\sem{\gamma} \subseteq \zed^{\vec{x} \times \vec{x}}$ and
  $\sem{\mathit{SC}^k_Q(\gamma)} \subseteq
  \zed^{\mathcal{V}^K_{\mathcal{H}}}$.
\end{lemma}
\begin{proof}
By induction on $\len{\gamma} > 0$, applying a case split on the type
of the first production in $\gamma$. \qed
\end{proof}

The following theorem summarizes the first major result in this paper,
namely that any $K$-index underapproximation of the semantics of a
recursive program $\mathcal{P}$ can be computed by looking at the
semantics of a non-recursive program $\mathcal{H}^K$, obtained from
$\mathcal{P}$ by a syntactic source-to-source transformation.  

%% We will sometimes drop the superscript $K$ in the following, when it
%% is clear from the context.

\begin{theorem}\label{thm:query}
Let $\mathcal{P} = \langle P_1, \ldots, P_n \rangle$ be a program,
$\vec{x} = \vec{x}_1 \cdot \ldots \cdot \vec{x}_n$ be the tuple of
variables in $\mathcal{P}$, and let $q \in \nf{P_i}$ be a non-final
control state of $P_i = \langle \vec{x}_i, \vec{x}^{in}_i,
\vec{x}^{out}_i, S_i, q^{\mathit{init}}_i, F_i, \Delta_i
\rangle$. Moreover, let $\mathcal{H}^K = \tuple{\mathit{query}^0,
  \ldots, \mathit{query}^K}$ be the program defined by Algorithm
\ref{alg:query-nodots}. For any $1 \leq k \leq K$, we have:
\begin{multline*}
\sem{\mathcal{P}}_{q}^{(k)} = \{ \tuple{(\widetilde{I} \proj_{\vec{x}_I}[\vec{x}/\vec{x}_I])\proj_{\vec{x}_i}, 
(\widetilde{I} \proj_{\vec{x}_O}[\vec{x}/\vec{x}_O])\proj_{\vec{x}_i} } \mid 	\\ 
\widetilde{I} \cdot \widetilde{O} \in \sem{\mathcal{H}^K}_{\mathit{query}^k}, 
\widetilde{I}(X) = Q\}\enspace .
\end{multline*}
\end{theorem}
\begin{proof}
Let $G_{\mathcal{P}} = \tuple{\Vars,\widehat{\Theta},\prod}$ be the
visibly pushdown grammar corresponding to $\mathcal{P}$. By
definition, we have
\[\sem{\mathcal{P}}^{(k)}_q = \Big\{ \tuple{ I \proj_{\vec{x}_i}, O \proj_{\vec{x}_i} } \mid 
I \cdot O \in \textstyle{\bigcup_{Q \xRightarrow[\df{k}]{\gamma} w}} \sem{\gamma} \Big\} \]

\noindent''$\subseteq$'' Let $Q \xRightarrow[\df{k}]{\gamma} w$ be a
derivation of $G_{\mathcal{P}}$, and $I \cdot O \in \sem{\gamma}$ be a
tuple from $\zed^{\vec{x} \times \vec{x}}$. By Lemma \ref{fsa-dfk},
$\age{Q} \arrow{\gamma}{} \vec{0}$ is a path in
$A^{\df{k}}(G_{\mathcal{P}})$, and by Lemma \ref{query-fsa},
$\mathit{SC}_Q^k(\gamma)$ is a feasible interprocedurally valid path
of $\mathit{query}^k(Q)$. By Lemma
\ref{lemma:derivation-semantics-gsm}, there exists tuples
$\widetilde{I}, \widetilde{O}$ such that $\widetilde{I} \cdot
\widetilde{O} \in \sem{\mathit{SC}_Q^k(\gamma)}$, and $I \cdot O =
\left(\widetilde{I} \proj_{\vec{x}_I \cdot \vec{x}_O}\right)[\vec{x}
  \cdot \vec{x}/\vec{x}_I \cdot \vec{x}_O]$. We obtain thus $I =
\widetilde{I} \proj_{\vec{x}_I}[\vec{x}/\vec{x}_I]$ and $O =
\widetilde{I} \proj_{\vec{x}_O}[\vec{x}/\vec{x}_O]$. 

\noindent''$\supseteq$'' Let $\widetilde{I}, \widetilde{O} \in
\zed^{\mathcal{V}_{\mathit{query}}^k}$, such that $\widetilde{I} \cdot
\widetilde{O} \in \sem{\mathcal{H}^K}_{\mathit{query}^k}$ and
$\widetilde{I}(X) = Q$. Then there exists a feasible interprocedurally
valid path $\pi$ of $query^k(Q)$, such that $\widetilde{I} \cdot
\widetilde{O} \in \sem{\pi}$. By Lemma \ref{query-fsa}, there exists a
control word $\gamma \in \prod^*$, such that $\age{Q} \arrow{\gamma}{}
\vec{0}$ and $\pi = \mathit{SC}_Q^k(\gamma)$. By Lemma
\ref{lemma:derivation-semantics-gsm},
$\left(\widetilde{I}\proj_{\vec{x}_I \cdot \vec{x}_O}\right)[\vec{x}
  \cdot \vec{x} / \vec{x}_I \cdot \vec{x}_O] \in \sem{\gamma}$. By
Lemma \ref{fsa-dfk}, we have that $Q \xRightarrow[\df{k}]{\gamma} w$
is a derivation of $G_{\mathcal{P}}$. We can conclude that
$\tuple{(\widetilde{I}
  \proj_{\vec{x}_I}[\vec{x}/\vec{x}_I])\proj_{\vec{x}_i},
  (\widetilde{I}
  \proj_{\vec{x}_O}[\vec{x}/\vec{x}_O])\proj_{\vec{x}_i} } \in
\sem{\mathcal{P}}_q$. \qed
\end{proof}

As a last point, we observe that the bounded-index sequence
$\{\sem{\mathcal{P}}^{(k)}\}_{k=1}^\infty$ satisfies several
conditions that advocate its use in program analysis, as an
underapproximation sequence. The subset order and set union is
extended to tuples of relations, point-wise.
\[\belowdisplayskip=2pt\abovedisplayskip=2pt\begin{array}{lclllr}
\sem{\mathcal{P}}^{(k)} & \subseteq & \sem{\mathcal{P}}^{(k+1)} & 
\mbox{for all $k \geq 1$} & (A1) 
\\
\sem{\mathcal{P}} & = & \bigcup_{k=1}^\infty \sem{\mathcal{P}}^{(k)} && (A2)
\end{array}
\]
Condition ($A1$) requires that the sequence is monotonically
increasing, the limit of this increasing sequence being the actual
semantics of the program ($A2$). These conditions follow however
immediately from the two first points of Lemma~\ref{lem:derivdecomp}.
To decide whether the limit $\sem{\mathcal{P}}$ has been reached by
some iterate $\sem{\mathcal{P}}^{(k)}$, it is enough to check that the
tuple of relations in $\sem{\mathcal{P}}^{(k)}$ is inductive with
respect to the statements of $\mathcal{P}$. This can be implemented as
an SMT query.

 % end input /Users/pierreganty/counter-recursive/bounded-query.tex
 %
% start input /Users/pierreganty/counter-recursive/completness.tex
% vim:ts=2:sw=2
%
%        File: completness.tex
%     Created: 
% Last Change: $Date: 2015-07-11 00:10:32 +0200 (Sat, 11 Jul 2015) $
%
% Written to be compiled by pdflatex
% usual typos to check:
% twice the same word: \(\<\w*\>\)\_s*\1\>
% too much space in math environment: \\\\\_s*\\end
% useless space at the end of a line: %s/\s*$//
% no concluded proof
% non matching parenthesis

\section{Completeness of Index-Bounded Underapproximations for Bounded Programs}\label{sec:completness}

In this section we define a class of recursive programs for which the
precise summary semantics of each program in that class is effectively
computable.  We show for each program \(\mathcal{P}\) in the class
that \begin{inparaenum}[\upshape(\itshape a\upshape)] \item
  \(\sem{\mathcal{P}}=\sem{\mathcal{P}}^{(k)}\) for some value \(k\geq
  1\), bounded by a linear function in the total number
  $\loc(\mathcal{P})$ of control states in $\mathcal{P}$, and
  moreover \item the semantics of \(\mathcal{H}^{k}\) is effectively
  computable
\end{inparaenum} (and so is that of $\sem{\mathcal{P}}^{(k)}$ 
by Thm.~\ref{thm:query}).

Given an integer relation $R \subseteq \zed^n \times \zed^n$, its {\em
  transitive closure} $R^+ = \bigcup_{i=1}^\infty R^i$, where $R^1 =
R$ and $R^{i+1} = R^i \comp R$, for all $i \geq 1$. In general, the
transitive closure of a relation is not definable within decidable
subsets of integer arithmetic, such as Presburger arithmetic. In this
section we consider two classes of relations, called {\em periodic},
for which this is possible, namely octagonal relations, and finite
monoid affine relations. 
\begin{compactdesc}
	\item[Octagonal relation] An {\em octagonal relation} is defined by a finite
		conjunction of constraints of the form $\pm x \pm y \leq c$, where $x$ and
		$y$ range over the set $\vec{x} \cup \vec{x'}$, and $c$ is an integer
		constant. The transitive closure of any octagonal relation has been shown
		to be Presburger definable and effectively computable~\cite{BIK10}.
	\item[Linear affine relation] A {\em linear affine relation} is defined by a
		formula $\mathcal{R}(\vec{x},\vec{x'}) \equiv C\vec{x} \geq \vec{d}
		~\wedge~ \vec{x'} = A\vec{x} + \vec{b}$, where $A \in \zed^{n \times n}$,
		$C \in \zed^{p \times n}$ are matrices and $\vec{b} \in \zed^n$, $\vec{d}
		\in \zed^p$.  $\mathcal{R}$ is said to have the {\em finite monoid
		property} if and only if the set $\{A^i \mid i \geq 0\}$ is finite. It is
		known that the finite monoid condition is decidable \cite{Boi98}, and
		moreover that the transitive closure of a finite monoid affine relation is
		Presburger definable and effectively computable~\cite{FL02,Boi98}.
\end{compactdesc}

We define a {\em bounded-expression} \(\pat\) to be a regular
expression of the form $\pat=w_1^* \ldots w_d^*$, where \(d\geq 1\)
and each \(w_i\) is a non-empty word.  A language (not necessarily
context-free) $L$ over alphabet \(\Sigma\) is said to be {\em bounded}
if and only if \(L\) is included in (the language of) a bounded
expression \(\pat\).

\begin{theorem}[\cite{Luker78}]\label{bounded-finite-index}
  Let $G = (\Vars,\Sigma,\prod)$ be a grammar, and $X \in
  \Vars$ be a nonterminal, such that $L_X(G)$ is bounded. Then
  there exists a linear function $\mathcal{B} \colon \nats \rightarrow \nats$
  such that $L_X(G) =L^{(k)}_X(G)$ for some \(1 \leq k \leq
	\mathcal{B}(\card{\Vars})\).
\end{theorem}
If the grammar in question is $G_{\mathcal{P}}$, for a program
$\mathcal{P}$, then clearly $\card{\Vars}$ is bounded by the number of
control locations in $\mathcal{P}$, by the definition of
$G_{\mathcal{P}}$. The class of programs for which our method is
complete is defined below:
\begin{definition}\label{bounded-periodic}
  Let $\mathcal{P}$ be a program and $G_{\mathcal{P}} = (\Vars,
  \widehat{\Theta}, \prod)$ be its corresponding visibly pushdown
  grammar. Then $\mathcal{P}$ is said to be {\em bounded periodic} if
  and only if:
  \begin{compactenum}
  \item $L_X(G_{\mathcal{P}})$ is bounded for each $X \in \Vars$;
  \item each relation $\rho_\tau$ occurring in the program, for some
    $\tau \in \widehat{\Theta}$, is periodic.
  \end{compactenum}
\end{definition}
\begin{example}(continued from Ex.~\ref{ex:betterthanboundedstack})
Recall that \(L_{Q_1^{\mathit{init}}}(G_{\mathcal{P}})=L^{(2)}_{Q_1^{\mathit{init}}}(G_{\mathcal{P}})\) which equals to the set 
\(\{ \bigl(\tau_1 \calls \tau_2\bigr)^n \tau_4 \bigl(\tau_2\rets
\tau_3\bigr)^n \mid n \geq 0 \}\subseteq 
\bigl(\tau_1 \tau_2 \calls \bigr)^* \tau_4^* \bigl(\tau_2\rets
\tau_3\bigr)^*\).\hfill \(\blacksquare\)%
\end{example}
Concerning condition \(1\), it is decidable \cite{ginsburg} and
previous work \cite{GT09} defined a class of programs following a
recursion scheme which ensures boundedness of the set of
interprocedurally valid paths.

This section shows that the underapproximation sequence
$\{\sem{\mathcal{P}}^{(k)}\}_{k=1}^\infty$, defined in
Section~\ref{sec:bounded-query}, when applied to any bounded periodic
programs \(\mathcal{P}\), always yields \(\sem{\mathcal{P}}\) in at
most $\mathcal{B}(\loc(\mathcal{P}))$ steps, and moreover each iterate
$\sem{\mathcal{P}}^{(k)}$ is computable and Presburger
definable. Furthermore the method can be applied {\em as it is} to
bounded periodic programs, without prior knowledge of the bounded
expression $\pat \supseteq L_Q(G_{\mathcal{P}})$.

The proof goes as follows. Because \(\mathcal{P}\) is bounded
periodic, Thm.~\ref{bounded-finite-index} shows that the semantics
\(\sem{\mathcal{P}}\) of \(\mathcal{P}\) coincide with its \(k\)-index
semantics \(\sem{\mathcal{P}}^{(k)}\) for some \(1 \leq k \leq
\mathcal{B}(\loc(\mathcal{P})) \). Hence, the result of
Thm.~\ref{thm:query} shows that for each \(q\in\nf{\mathcal{P}}\), the
\(k\)-index semantics \(\sem{\mathcal{P}}^{(k)}_q=\{ \tuple{I \proj_{\vec{x}_I}, I \proj_{\vec{x}_O} } \mid I \cdot O \in \sem{\mathcal{H}^K}_{\mathit{query}^k}, I(X) = Q\}\), that is,  the semantics \(\sem{\mathcal{P}}^{(k)}_q\) is computed from that of procedure \(\mathit{query}^k\) called with \(X=Q\). 
Then, because \(\mathcal{P}\) is
bounded, we show in Thm.~\ref{query-flattable} that every procedure
\(\mathit{query}^k\) of program \(\mathcal{H}\) is {\em flattable}
(Def.~\ref{flattable}). Moreover, since the only transitions of
$\mathcal{H}$ which are not from $\mathcal{P}$ are equalities and {\bf
  havoc}, all transitions of \(\mathcal{H}\) are periodic. Since each
	procedure \(\mathit{query}^k\) is flattable then \(\sem{\mathcal{P}}\) is
computable in finite time by existing tools, such as \textsc{Fast}
\cite{BFLP03} or \textsc{Flata} \cite{BIL09,BIK10}. In fact, these
tools are guaranteed to terminate provided that
\begin{inparaenum}[\upshape(\itshape a\upshape)]
	\item the input program is flattable; and 
	\item loops are labelled with periodic relations.
\end{inparaenum}

\begin{definition}\label{flattable}
 Let $\mathcal{P} = \langle P_1, \ldots, P_n \rangle$ be a
 non-recursive program and $G_{\mathcal{P}} =
 (\Vars,\widehat{\Theta},\prod)$ be its corresponding visibly
 pushdown grammar. Procedure \(P_i\) is said to be {\em flattable} if
 and only if there exists a bounded and regular language $R$ over
 $\widehat{\Theta}$, such that $\sem{\mathcal{P}}_{P_i} =
 \bigcup_{\alpha \in L_{P_i}(G_{\mathcal{P}}) \cap R} \sem{\alpha}$.
\end{definition}
Notice that a flattable program is not necessarily bounded (Def.
\ref{bounded-periodic}), but its semantics can be computed by looking
only at a bounded subset of interprocedurally valid paths.

The proof that the procedures \(\mathit{query}^k\) are flattable
relies on grammar based reasoning, and, in particular, on control-sets
with relative completeness properties. Let us now turn to our main
result, Theorem~\ref{query-flattable} stated next, whose proof is
organized as follows.  First, Proposition~\ref{bounded-control-set}
roughly states that provided \(L(G)\) is bounded, then a bounded
subset of the \(k\)-index depth-first derivations suffices to capture
\(L^{(k)}(G)\) for some \(k\).  The proof of this proposition is split
into Theorem~\ref{bounded-underapprox},
Lemma~\ref{bounded-szilard-underapproximation} and
Lemma~\ref{lem:bowtie}. The rest of the proof uses
Lemma~\ref{query-fsa} which roughly states that there is a
well-behaved mapping from the \(k\)-index depth-first derivations of
\(G_{\mathcal{P}}\) from \(Q\) to the runs of
\(\mathit{query}^{k}(Q)\) for every value of \(k\) and \(Q\).

\begin{theorem}\label{query-flattable}
Let $\mathcal{P} = \langle P_1, \ldots, P_n \rangle$ be a bounded
program, %, and let \(q\in\nf{\mathcal{P}}\). 
then, for any $k \geq 1$,
procedure $\mathit{query}^k$ of program \(\mathcal{H}\) is flattable.
\end{theorem}

\subsection{Bounded languages with bounded control sets}

The following result was proved in \cite{gmm12-fmsd}:
\begin{theorem}[Thm. 1 from \cite{gmm12-fmsd}, also in \cite{Latteux78}]
  \label{bounded-underapprox}
	For every regular language \(L\) over alphabet \(\Sigma\) there exists a bounded expression \(\pat_{\Gamma}\) such that
	\(Pk_{\Sigma}(L\cap \pat_{\Gamma})= Pk_{\Sigma}(L)\).
\end{theorem}

Next we prove a result characterizing a subset of derivations sufficient to
capture a bounded context-free language. But first, given a grammar \(G = (\Vars,\Sigma,\prod)\) and \(X\in\Vars\) define
\[\Gamma_X^{\df{k}}= \{ \gamma\in\prod^* \mid \age{X} \xrightarrow{\gamma}
\vec{0} \text{ in } A^{\df{k}} \}\enspace . \] Observe that
\(\Gamma_X^{\df{k}}\) is a regular language, because \(A^{\df{k}}\) is
a finite state automaton.

\begin{lemma}\label{bounded-szilard-underapproximation}
Let \(G = (\Vars,\Sigma,\prod)\) be a grammar and $X \in \Vars$
be a nonterminal, such that for all \(p\in\prod\), \(X\) does not occur in \(\mathit{tail}(p)\).
Also $L_X(G) \subseteq (a_1 w_1)^* \ldots (a_d w_d)^*$ where \(a_1,\ldots,a_d\) are distinct
symbols of \(\Sigma\) none of which occurs in \(w_1\cdots w_d\). Then, for each \(k\geq 1\)
there exists a bounded expression $\pat_\Gamma$ over \(\prod\)
such that \(L_X^{(k)}(G)=\hat{L}_X(\pat_{\Gamma} \cap \Gamma_X^{\df{k}}, G)\).
%$DF^{(k)}_{X}(\pat_\Gamma,G) = L^{(k)}_X(G)$.
\end{lemma}
\begin{proof}
\noindent 
We first establish the claim that for each \(k\geq 1\), there exists a bounded expression \(\pat_\Gamma\) over \(\prod\)
such that \(Pk_{\prod}(\Gamma^{\df{k}}\cap \pat_\Gamma) = Pk_{\prod}(\Gamma^{\df{k}}).\)
	By Corollary~\ref{szilard-regular},
	$\Gamma^{\df{k}}$ is a regular language, and by Theorem~\ref{bounded-underapprox}, there exists a bounded expression $\pat_\Gamma$ over \(\prod\) such that $Pk_{\prod}(\Gamma^{\df{k}}\cap\pat_\Gamma) =
	Pk_{\prod}(\Gamma^{\df{k}})$ which proves the claim. 
	%Next we prove that
	%\(\hat{L}_X(\Gamma^{(k)},G)=\hat{L}_X(\pat_{\Gamma} \cap \Gamma^{(k)}\cap \Gamma^{\df{k}}, G)\).
	%\(DF_X^{(k)}(\pat_\Gamma,G)=L_X^{(k)}(G)\).

%Let $\prod = \langle X_i \rightarrow v_i \rangle_{i=1}^m$ be
%  the sequence of productions of $G$, taken in some fixed order. 
Define \(\mathcal{A} = \{a_1,\ldots, a_d\}\) and assume \(\prod\) is given as a linearly ordered set of \(m\) productions \(\{p_1,\ldots, p_m\}\).
Then for \(u\) such that \(X \xRightarrow{\gamma} u \), we have $Pk_{\mathcal{A}}(u) = Pk_{\prod}(\gamma) \times \Pi$ where \(\Pi\) is the matrix of \(m\) rows and \(d\) columns where row \(i\) is given by \(Pk_{\mathcal{A}}(\mathit{tail}(p_i))\).
Next, let \(\gamma_1,\gamma_2\) be two control words such that $Pk_{\prod}(\gamma_1) =
Pk_{\prod}(\gamma_2)$ and each \(\gamma_i\) (\(i=1,2\)) generates a word \(u_i\) of \(L_X(G)\), that is $X \stackrel{\gamma_i}{\Longrightarrow} u_i$.
We conclude from the above that  $Pk_{\mathcal{A}}(u_1) = Pk_{\mathcal{A}}(u_2)$. Moreover, the assumption $L_X(G) \subseteq (a_1 w_1)^* \ldots (a_d w_d)^*$ where \(a_1,\ldots,a_d\) are distinct
symbols shows that \(u_1\proj_{\mathcal{A}}=u_2 \proj_{\mathcal{A}}\). Furthermore, because  
no symbol of \(\mathcal{A}\) occurs in \(w_1\cdots w_d\) we find that \(u_1=u_2\).

  \vspace*{\baselineskip}\noindent
	To show \(L_X^{(k)}(G)= \hat{L}_X(\pat_{\Gamma} \cap \Gamma_X^{\df{k}}, G)\) we prove that 
	\(L_X^{(k)}(G)\subseteq \hat{L}_X(\pat_{\Gamma} \cap \Gamma_X^{\df{k}}, G)\)	
	%$L^{(k)}_X(G) \subseteq DF_X^{(k)}(\pat_\Gamma,G)$, 
	the other direction being immediate because of Proposition~\ref{df-are-complete} which says that \(L_X^{(k)}(G)=\hat{L}_X(\Gamma^{\df{k}}, G)\) 
	and because only those control words \(\gamma\) such that \( \mathit{head}((\gamma)_1)=X \) matters.
%	$L_X^{(k)}(G) = DF^{(k)}_X(G)$. 

  So, let $u \in \hat{L}_X(\Gamma_X^{\df{k}}, G)$ be a word, %DF^{(k)}_X(G)$ be a word,
	and $X \xRightarrow[\df{k}]{\gamma} u$ be a depth-first
	derivation of $u$. Since \(Pk_{\prod}( \Gamma_X^{\df{k}}\cap\pat_\Gamma) =
	Pk_{\prod}(\Gamma_X^{\df{k}})\), there
	exists a control word $\beta \in \Gamma^{\df{k}}\cap\pat_\Gamma$ such that $Pk_{\prod}(\beta) = Pk_{\prod}(\gamma)$.
	Also because no production \(p\in\prod\) is such that \(\mathit{tail}(p)\) contains an occurrence of \(X\), we find
	that \( (\beta)_1 = (\gamma)_1 \).
	Finally, Lemma~\ref{lemma:unique-df-control-word} shows that given
	\(\beta\in\Gamma^{\df{k}}\), there exist a (unique) word \(u'\) such that
	\(X \xRightarrow[\df{k}]{\beta} u'\), hence \(u'=u\) as shown above.
	\qed
\end{proof}
For the rest of this section, let $G = (\Vars, \Theta, \prod)$ be a
visibly pushdown grammar (we ignore for the time being the distinction between
tagged and untagged alphabet symbols), and $X_0 \in \Vars$ be an
arbitrarily chosen nonterminal.

Let \(\pat=w_1^*\cdots w_d^*\) be a bounded expression %
\footnote{Recall that each \(w_i\) is a non-empty word.}
over alphabet \(\Theta\) and
define the bounded expression $\patt=(a_1 w_1)^* \dots (a_d w_d)^*$
such that \(\set{a_1,\ldots,a_d}\) and \(\Theta\) are disjoint.
Next, let $\ell_i = \len{a_i\, w_i}$ for every $1 \leq i \leq d$ and let
$G^{\patt}=(\Vars^{\patt},\Theta \cup \set{a_1,\ldots,a_d},\delta^{\patt})$ be the regular grammar where
\[\begin{array}{rcl}
\Vars^{\patt} & = &  \set{\textsc{q}^{(s)}_{r} \mid 1\leq s\leq d \, \land \,  1\leq r\leq \ell_s} \\[0.3cm]
\delta^{\patt} & = & \set{\textsc{q}^{(s)}_{i}\rightarrow (a_s\, w_s)_i\; \textsc{q}^{(s)}_{i+1}\mid 1 \leq  s \leq  d \, \land \, 1 \leq  i < \ell_s} \; \cup\\
& & \set{\textsc{q}^{(s)}_{\ell_s}\rightarrow (a_s\, w_s)_{\ell_s}\; \textsc{q}^{(s')}_{1}\mid 1 \leq  s \leq  s' \leq  d}\enspace .
%& & \set{\textsc{q}_1^{(s)}\rightarrow \varepsilon \mid 1 \leq  s \leq  d}\enspace .
\end{array}\]
Checking \(\{w\mid \textsc{q}_1^{(s)}\Rightarrow^* w\, \textsc{q}_1^{(x)} \text{ for some } 1 {\leq} s {\leq} x {\leq} d\}=L(\patt)\) holds is routine.
Next, given \(G\) and \(G^{\patt}\), define
\(G^{\bowtie}=(\Vars^{\bowtie},\Theta\cup\set{a_1,\ldots,a_d},\prod^{\bowtie})\) such that \(
L_{X^{\bowtie}_0}(G^{\bowtie})=L_{X_0}(G)\parallel L(\patt) \).%
\footnote{%
Given two languages \(L_1 \subseteq \Sigma_1^*\)
and \(L_2 \subseteq \Sigma_2^*\) their asynchronous product, denoted \(L_1 \parallel L_2\), is the language \(L\) over the alphabet
\(\Sigma=\Sigma_1\cup\Sigma_2\) such that \(w\in L\) if{}f
the projections of $w$ to $\Sigma_1$ and $\Sigma_2$ belong to $L_1$ and $L_2$, respectively. %
Observe that the \(L_1\parallel L_2\) depends on \(L_1\), \(L_2\) \textbf{and} also their underlying alphabet \(\Sigma_1\) and \(\Sigma_2\).%
}

\begin{itemize}
	\item $\Vars^{\bowtie}=\set{X^{\bowtie}_0}\cup\set{ [\textsc{q}^{(s)}_r X \textsc{q}^{(x)}_{y}] \mid X\in\Vars,\, \textsc{q}^{(s)}_r,\textsc{q}^{(x)}_{y} \in\Vars^{\patt},\, s\leq x}$
	\item $\prod^{\bowtie}$ is the set containing for every $1\leq s\leq x \leq d$ a production
		$X^{\bowtie}_0\rightarrow [\textsc{q}^{(s)}_1 X_0 \textsc{q}^{(x)}_{1}]$, and:
		\begin{itemize}
			\item for every production $X\rightarrow \gamma\in\prod$, 
				\(\prod^{\bowtie}\) has a production  	
			%	and for every $1\leq s\leq d$, $1\leq r \leq j_s$ a production
				\begin{align}
					[\textsc{q}^{(s)}_r X \textsc{q}^{(x)}_{y}]&\rightarrow\gamma & \text{if } \textsc{q}^{(s)}_r \rightarrow \gamma\; \textsc{q}^{(x)}_{y}\in\prod^{\patt}\enspace ;
				\end{align}
			\item for every production $X\rightarrow \gamma\; Y\in\prod$, \(\prod^{\bowtie}\) has a production
				%for every $1\leq s\leq u\leq d$, $1\leq r \leq j_s$, $1\leq v\leq j_u$ productions
				\begin{multline}
					[\textsc{q}^{(s)}_{r} X \textsc{q}^{(x)}_y]\rightarrow \gamma \; [\textsc{q}^{(z)}_{t} Y \textsc{q}^{(x)}_y]\\
					\text{if } \textsc{q}^{(s)}_r \rightarrow \gamma\; \textsc{q}^{(z)}_{t}\in \prod^{\patt}; \label{eq:symbolepsilon}
				\end{multline}
				%\begin{align}
				%	[\textsc{q}^{(s)}_{r} X \textsc{q}^{(u)}_v]&\rightarrow \gamma \, [\textsc{q}^{(s)}_{r+1} Y \textsc{q}^{(u)}_v] &\text{if $r< j_s$ and $\gamma=b_{r}^{(s)}$; and}\label{eq:symbolepsilon}\\
				%	[\textsc{q}^{(s)}_{j_s} X \textsc{q}^{(u)}_v]&\rightarrow \gamma \, [\textsc{q}^{(s')}_{1} Y \textsc{q}^{(u)}_v] &\text{if $\gamma=b_{j_s}^{(s)}$ and $s\leq s'\leq u$}\enspace ;\label{eq:symbolsigmat}
				%\end{align}
				\item for every production $X\rightarrow \tau\; Z\; \sigma\; Y\in\prod$, \(\prod^{\bowtie}\) has a production
				\begin{multline}
					[\textsc{q}^{(s)}_{r} X \textsc{q}^{(x)}_{y}] \rightarrow \tau\; [\textsc{q}^{(z)}_t Z \textsc{q}^{(u)}_v]\;\sigma\;[\textsc{q}^{(\ell)}_k Y \textsc{q}^{(x)}_y]\\
					\text{if } \textsc{q}^{(s)}_{r}\rightarrow \tau\;\textsc{q}^{(z)}_t\in\prod^{\patt} \text{ and } \textsc{q}^{(u)}_v\rightarrow \sigma\;\textsc{q}^{(\ell)}_k\in\prod^{\patt};\label{eq:procvar}
				\end{multline}
			\item for every production \(\textsc{q}^{(s)}_1 \rightarrow a_s \; \textsc{q}^{(u)}_v\in \delta^{\patt}\), \(\prod^{\bowtie}\) has a production
				\begin{align}
					[\textsc{q}^{(s)}_{1} X \textsc{q}^{(x)}_y]&\rightarrow a_s \; [\textsc{q}^{(u)}_{v} X \textsc{q}^{(x)}_{y}] \enspace .
				\end{align}
		\end{itemize}
		\(\prod^{\bowtie}\) has no other production.
\end{itemize}

Next we define the mapping \(\xi\) which maps each nonterminal
$[\textsc{q}^{(s)}_r X \textsc{q}^{(x)}_y] \in \Vars^\bowtie$ onto \(X\),
\(X_0^{\bowtie}\) onto \(X_0\), every \(a_i\), \(1\leq i\leq d\), onto
\(\varepsilon\) and maps any other terminal (\(\Theta\)) onto itself. Then
\(\xi\) is naturally extended to words over \(\Theta\cup\set{a_1,\ldots,a_d}\cup
\Vars^{\bowtie}\).  Next we lift \(\xi\) to productions of
\(\prod^\bowtie\) such that the mapping of a production is defined by the
mapping of its head and tail.  The lifting of \(\xi\) to sequences of
productions and sets of sequences of productions is defined in the obvious way.

From the above definition we observe that given a derivation
$D^\bowtie \equiv X^\bowtie_0 \Rightarrow [\textsc{q}^{(s)}_1 X_0 \textsc{q}^{(x)}_{1}]
\Longrightarrow^* w$ in $G^\bowtie$, \(\xi\) maps \(D^{\bowtie}\) onto
a derivation of \(G\) of the form \(X_0 \Rightarrow X_0 \Longrightarrow^* w\proj_{\Theta}\). 

\noindent
\begin{lemma}\label{lem:bowtie}
	Let $G = (\Vars, \Theta, \prod)$ be a visibly pushdown grammar, $X_0
	\in \Vars$ be a nonterminal such that \(L_{X_0}(G)\subseteq \pat\) for
	a bounded expression $\pat = w_1^* \ldots w_d^*$. Let
	\(\set{a_1,\ldots,a_d}\) be a set of \(d\) symbols disjoint from \(\Theta\).
	Then for every \(k\geq 1\), the following hold:
\begin{compactenum}
	%\item	\(L^{(k)}_{X_0}(G)=L^{(k)}_{X_0^{\bowtie}}(G^{\bowtie})\) 
	\item Let \(i_1,\ldots,i_d \in\nats\) we have 
		\[w_1^{i_1} \ldots w_d^{i_d} \in L^{(k)}_{X_0}(G) \text{ if{}f } (a_1 w_1)^{i_1} \ldots(a_d w_d)^{i_d} \in L^{(k)}_{X_0^{\bowtie}}(G^{\bowtie})\enspace ;\]
	\item Given a control set \(\Gamma\) over \(\prod^{\bowtie}\) such that
		\[\hat{L}_{X_0^{\bowtie}}(\Gamma\cap \Gamma^{\df{k}}(G^{\bowtie}),G^{\bowtie})= L^{(k)}_{X_0^{\bowtie}}(G^{\bowtie})\]
		then the control set \(\Gamma'=\xi(\Gamma)\) over \(\prod\) satisfies
		\[\hat{L}_{X_0}(\Gamma'\cap \Gamma^{\df{k}}(G),G)= L^{(k)}_{X_0}(G)\enspace .\]
\end{compactenum}
\end{lemma}
\begin{proof}
	The proof of point 1 is by induction. As customary, we show the following stronger statement: let \(k\geq 1\) and \(w\in (\Theta\cup\set{a_1,\ldots,a_d})^*\cdot \Theta\), 
	we have $[\textsc{q}_r^{(s)} X
	\textsc{q}_{v}^{(u)}] \xRightarrow[(k)]{}^* w$ if{}f
	$\textsc{q}_r^{(s)}\Rightarrow^* w \; \textsc{q}_{v}^{(u)}$ and $X\xRightarrow[(k)]{}^* w\proj_{\Theta}$.  The proof of the
	if direction is by induction on the length of
	\(\textsc{q}_r^{(s)}\Rightarrow^* w \; \textsc{q}_{v}^{(u)}\).

	\noindent \(\mathbf{i=1}\). 
	Then \(\textsc{q}_r^{(s)}\rightarrow \tau \; \textsc{q}_{v}^{(u)}\in\prod^{\patt}\).	Two cases can occur:
  \begin{inparaenum}[\upshape(\itshape i\upshape)]
	\item \(\tau\in\Theta\); or
	\item \(\tau \in \set{a_1,\ldots,a_d}\).
	\end{inparaenum}

	In case \upshape(\itshape i\upshape), we conclude from
	\(X \xRightarrow[(k)]{}^* w\proj_{\Theta}\)  that \(w{=}w\proj_{\Theta}{=}\tau\) and \(X\rightarrow
	\tau\in \prod\), hence that \([\textsc{q}_r^{(s)} X
	\textsc{q}_{v}^{(u)}]\rightarrow \tau\in\prod^{\bowtie}\), and finally that
	\([\textsc{q}_r^{(s)} X
	\textsc{q}_{v}^{(u)}] \xRightarrow[(k)]{}^*  w\).  Case
	\upshape(\itshape ii\upshape) is not allowed since
	\(w\) must end with a symbol in \(\Theta\). 

	\noindent \(\mathbf{i>1}\). Then \(\textsc{q}_r^{(s)}\Rightarrow \tau\;
	\textsc{q}_{r'}^{(s')} \Rightarrow\circ\Rightarrow^* \overbrace{\tau \; y}^{w} \textsc{q}_{v}^{(u)}\). As seen previously, two cases can occur:
  \begin{inparaenum}[\upshape(\itshape i\upshape)]
	\item \(\tau \in \set{a_1,\ldots,a_d}\); or
	\item \(\tau\in\Theta\).
	\end{inparaenum}
	In case \upshape(\itshape i\upshape), because \(w=\tau\, y\) and \(\tau\notin \Theta\) we find that \(X \xRightarrow[(k)]{}^* w\proj_{\Theta}=y\proj_{\Theta}\). Hence the induction hypothesis shows that 
	\([\textsc{q}_{r'}^{(s')} X \textsc{q}_{v}^{(u)}] \xRightarrow[(k)]{}^* y \). Finally the definition of \(G^{\bowtie}\) shows
	that \([\textsc{q}_r^{(s)}X\textsc{q}_v^{(u)} ]\rightarrow\tau\; [\textsc{q}_{r'}^{(s')} X \textsc{q}_v^{(u)}]\in \prod^{\bowtie}\), hence that
	\([\textsc{q}_r^{(s)}X\textsc{q}_v^{(u)} ] \\ \xRightarrow[(k)]{}^* \tau\; [\textsc{q}_{r'}^{(s')} X \textsc{q}_v^{(u)}] \xRightarrow[(k)]{}^* \tau\; y = w\) and we are done.

	For case \upshape(\itshape ii\upshape) (\(\tau\in\Theta\)), we do a (sub)case analysis according
	to the first production rule used in the derivation
	\(X \xRightarrow[(k)]{}^* w\proj_{\Theta}\).

	\begin{itemize}
		\item \(X\rightarrow\tau\). Then \(X \xRightarrow[(k)]{}^* w\proj_{\Theta}=\tau\).  On
			the other hand  \(\textsc{q}_r^{(s)}\Rightarrow \tau\,
			\textsc{q}_{r'}^{(s')} \Rightarrow\circ\Rightarrow^* \tau \, y\,
			\textsc{q}_{v}^{(u)}\) and our assumption on \(w=\tau\, y\) shows that
			\(y\) ends with a symbol in \(\Theta\). 	Hence a contradiction since
			\(w\proj_{\Theta}=\tau\) does not coincide with the projection of \(w=\tau\, y\).
		\item \(X\rightarrow \tau\, Y\). Then \(X \xRightarrow[(k)]{} \tau\,
			Y \xRightarrow[(k)]{}^* \tau\, y\proj_{\Theta}=w\proj_{\Theta}\). Also
			\(\textsc{q}_r^{(s)}\Rightarrow \tau\, \textsc{q}_{r'}^{(s')}
			\Rightarrow\circ\Rightarrow^* \tau \, y\, \textsc{q}_{v}^{(u)}\).  The
			induction hypothesis applied on \(Y \xRightarrow[(k)]{}^* y\proj_{\Theta}\) and
			\(\textsc{q}_{r'}^{(s')} \Rightarrow^* y\, \textsc{q}_{v}^{(u)}\) shows
			that \([\textsc{q}_{r'}^{(s')} Y \textsc{q}_v^{(u)}]
			\xRightarrow[(k)]{}^* y\). Finally, \(X\rightarrow \tau\, Y\in\prod\) and \(\textsc{q}_r^{(s)}\rightarrow \tau\, \textsc{q}_{r'}^{(s')}\in\prod^{\patt}\) show that \([\textsc{q}_r^{(s)} X \textsc{q}_{v}^{(u)}] \rightarrow \tau\, [\textsc{q}_{r'}^{(s')} Y \textsc{q}_{v}^{(u)}]\in\prod^{\bowtie}\), hence
			that
			\([\textsc{q}_r^{(s)} X \textsc{q}_{v}^{(u)}] \xRightarrow[(k)]{}^* \tau\, [\textsc{q}_{r'}^{(s')} Y \textsc{q}_{v}^{(u)}] \xRightarrow[(k)]{}^* \tau\, y=w\) and we are done.
		\item \(X\rightarrow \tau\, X_1\, \sigma\, X_2\).
			Then \(X \xRightarrow[(k)]{} \tau\, X_1 \sigma\, X_2 \xRightarrow[(k)]{}^* \tau\, w_1{\proj_{\Theta}}\, \sigma\, w_2\proj_{\Theta}=w\proj_{\Theta}\).
			Moreover, since \(\textsc{q}_s^{(r)} \Rightarrow^* w\, \textsc{q}_v^{(u)}\) and \(\tau,\sigma\in\Theta\) we find that
			there exist \(\textsc{q}_s^{(r)}\Rightarrow \tau\, \textsc{q}_{a}^{(b)} \Rightarrow^* \tau\, w_1\, \textsc{q}_{a'}^{(b')} \Rightarrow \tau\, w_1\, \sigma\,
			\textsc{q}_{c}^{(d)} \Rightarrow^* \tau\, w_1\, \sigma\, w_2 \textsc{q}_{v}^{(u)}\).
			Hence, the definition of \(G^{\bowtie}\) shows that 

			\[[\textsc{q}_s^{(r)} X \textsc{q}_{v}^{(u)}] \rightarrow \tau \; [\textsc{q}_{a}^{(b)} X_1 \textsc{q}_{a'}^{(b')}] \; \sigma\; [\textsc{q}_{c}^{(d)} X_2 \textsc{q}_{v}^{(u)}] \enspace.\]
			On the other hand, since \(X_1 X_2  \xRightarrow[(k)]{}^* w_1\proj_{\Theta}\; w_2\proj_{\Theta}\) (simply delete \(\tau\) and \(\sigma\)),
			Lemma~\ref{lem:derivdecomp} shows
			that either  \(X_1 \xRightarrow[(k-1)]{}^* w_1\proj_{\Theta}\) and \(X_2  \xRightarrow[(k)]{}^* w_2\proj_{\Theta}\); or
			\(X_1 \xRightarrow[(k)]{}^* w_1\proj_{\Theta}\) and \(X_2  \xRightarrow[(k-1)]{}^* w_2\proj_{\Theta}\).
			Let us assume the latter holds (the other being treated similarly).
			Applying the induction hypothesis, we find that 
			\([\textsc{q}_{a}^{(b)} X_1 \textsc{q}_{a'}^{(b')}] \xRightarrow[(k)]{}^* w_1 \) and
			\([\textsc{q}_{c}^{(d)} X_2 \textsc{q}_{v}^{(u)}] \xRightarrow[(k-1)]{}^* w_2 \), hence
			we conclude the case with the \(k\)-index derivation
			\([\textsc{q}_s^{(r)} X \textsc{q}_{v}^{(u)}] \xRightarrow[(k)]{}^*\, \tau\, [\textsc{q}_{a}^{(b)} X_1 \textsc{q}_{a'}^{(b')}] \, \sigma\, [\textsc{q}_{c}^{(d)} X_2 \textsc{q}_{v}^{(u)}] \xRightarrow[(k)]{}^*\) \(\tau\, [\textsc{q}_{a}^{(b)} X_1 \textsc{q}_{a'}^{(b')}] \, \sigma \, w_2 \xRightarrow[(k)]{}^* \tau\, w_1\, \sigma \, w_2\).
	\end{itemize}
	The ``only if'' direction is proved similarly, this time by induction
	on the length of the derivation 
	\([\textsc{q}_r^{(s)} X
	\textsc{q}_{v}^{(u)}]\xRightarrow[(k)]{}^* w\).
	
	\medskip

	For the proof of point 2 the ``\(\subseteq\)'' direction is obvious by definition
	of depth-first derivations. For the reverse direction ``\(\supseteq\)'' point 1
	combined with the assumption shows that 
	for every \(i_1,\ldots,i_d \in\nats\) the following equivalence holds:
	\[\begin{array}{c}
		w_1^{i_1} \ldots w_d^{i_d} \in L^{(k)}_{X_0}(G)\\ 
		\text{if{}f}\\ 
		(a_1 w_1)^{i_1} \ldots(a_d w_d)^{i_d} \in \hat{L}_{X_0^{\bowtie}}(\Gamma\cap\Gamma^{\df{k}},G^{\bowtie})\enspace .
	\end{array}\]
	So let \(D\equiv X_0^{\bowtie} \xRightarrow[(k)]{}^* w\) be a depth-first \(k\)-index derivation of \(G^{\bowtie}\) with
	control word conforming to \(\Gamma\). Now consider \(\xi(D)\), it defines
	again a depth-first \(k\)-index derivation except that this time the control
	word conforms to \(\xi(\Gamma)\). Further, the definition of \(\xi\) shows
	that the word generated by \(\xi(D)\) results
	from deleting the symbols \(\set{a_1,\ldots,a_d}\) from \(w=(a_1 w_1)^{i_1}\cdots (a_d w_d)^{i_d}\). To conclude, observe that \(w_1^{i_1}\cdots w_d^{i_d}\in L^{(k)}_{X_0}(G)\) and we are done. \qed
\end{proof}

The following proposition shows that $L_Q^{(k)}(G_{\mathcal{P}})$ is
captured by a subset of depth-first derivations whose control words
belong to some bounded expression.

\begin{proposition}\label{bounded-control-set}
	Let $G = (\Vars, \widehat{\Theta}, \prod)$ be a visibly pushdown grammar,
$X_0 \in \Vars$ be a nonterminal such that \(L_{X_0}(G)\) is bounded.
Then for each \(k\geq 1\) there exists a bounded expression \(\pat_{\Gamma}\)
over \(\prod\) such that \(\hat{L}_{X_0}(\pat_{\Gamma}\cap \Gamma^{\df{k}},G)=L_{X_0}^{(k)}(G)\).
\end{proposition}
%\subsection{Proof of Theorem~\ref{bounded-control-set}}
\begin{proof}
	Since \(L_{X_0}(G)\) is bounded there exists a bounded expression \(\pat= w_1^*
	\ldots w_d^*\) such that \(L_{X_0}(G) \subseteq \pat\).

  Next, define \(\set{a_1,\ldots,a_d}\) be an alphabet disjoint from $\Theta$.
	Lemma~\ref{lem:bowtie} shows that 
	for every \(i_1,\ldots,i_d \in\nats\) the equivalence \(w_1^{i_1} \ldots w_d^{i_d} \in L^{(k)}_{X_0}(G)\) if{}f \( (a_1 w_1)^{i_1} \ldots(a_d w_d)^{i_d} \in L^{(k)}_{X_0^{\bowtie}}(G^{\bowtie})\) holds.
	Next, applying Lemma~\ref{bounded-szilard-underapproximation} on
	\(L^{(k)}_{X_0^{\bowtie}}(G^{\bowtie})\) (whose assumptions holds by definition of \(G^{\bowtie}\)) we obtain a bounded expression
	\(\pat_{\Gamma^{\bowtie}}\) over \(\prod^{\bowtie}\) such that
	\(\hat{L}_{X_0^{\bowtie}}(\pat_{\Gamma^{\bowtie}}\cap \Gamma^{\df{k}},G^{\bowtie})=L^{(k)}_{X_0^{\bowtie}}(G^{\bowtie})\).
	Our next step is to apply the results of Lemma~\ref{lem:bowtie} (second point)
	to obtain that \(L_{X_0}^{(k)}(G)= \hat{L}_{X_0}(\xi(\pat_{\Gamma^{\bowtie}})\cap \Gamma^{\df{k}}), G)\).
	Finally, since $\pat_{\Gamma^{\bowtie}}$ is a bounded expression, and \(\xi\) is
	an homomorphism we have that
	\(\xi(\pat_{\Gamma^\bowtie})\) is bounded (see Lem.~\ref{lem:gsmbounded}), hence included in a bounded expression and we are done by setting \(\pat_{\Gamma}\) to
	\(\xi(\pat_{\Gamma^\bowtie})\).\qed
\end{proof}

% the \(k\)-index depth-first derivations of \(G\) into the interprocedurally valid
% paths of \(query^k_Q\). Then, applying Thm.~\ref{thm:ginsburg} on that mapping, we conclude the
% existence of a bounded and regular set of interprocedurally valid paths of
% \(query^k_Q\) sufficient to capture its semantics.
%
% \begin{theorem}\label{thm:ginsburg}
% Given two alphabets \(\Sigma\) and \(\Delta\), let \(f\) be a function
% from \(\Sigma^*\) into \(\Delta^*\) such that
% \begin{inparaenum}[\upshape(\itshape i\upshape)]
% \item if \(u\) is a prefix of \(v\) then \(f(u)\) is a prefix of
%   \(f(v)\);
% \item there exists an integer \(M\) such that \(|f(wa)|-|f(w)|\leq M\)
%   for all \(w\in\Sigma^*\) and \(a\in\Sigma\);
% \item \(f(\varepsilon)=\varepsilon\);
% \item \(f^{-1}(R)\) is regular for all regular languages \(R\).
% \end{inparaenum}
% Then \(f\) preserves regular sets. Furthermore, for each bounded
% expression \(\pat\) we have that \(f(\pat)\) is bounded.
% \end{theorem}

\subsection{Proof of Theorem~\ref{query-flattable}}

We recall two results from Ginsburg \cite{ginsburg}. 
\begin{theorem}[Theorem~3.3.2, \cite{ginsburg}]\label{thm:gsmregular}
	Each gsm mapping preserves regular sets.
\end{theorem}

\begin{lemma}[Lemma~5.5.3, \cite{ginsburg}]\label{lem:gsmbounded}
	\(S(w_1^* \ldots w_n^*)\) is bounded for each gsm \(S\) and all words \(w_1,\ldots,w_n\).
\end{lemma}

And finally, the proof that \(\mathit{query}^k\) is flattable.
\begin{proof}[of Theorem.~\ref{query-flattable}]
	Since \(\mathcal{P}\) is bounded periodic we can apply
	Proposition~\ref{bounded-control-set} showing the existence of a bounded expression
	\(\pat_\Gamma\) over \(\prod\) such that
	\(\hat{L}_Q(\pat_{\Gamma}\cap \Gamma^{\df{k}},G_{\mathcal{P}})=L_Q^{(k)}(G_{\mathcal{P}})\).
	Hence we find that \(\sem{\mathcal{P}}^{(k)}_q\) coincides with \(\bigcup_{\alpha \in
	L^{(k)}_Q(G_{\mathcal{P}})} \sem{\alpha}\) which in turn is equal
	to \(\bigcup_{\alpha\in
	\hat{L}_Q(\pat_{\Gamma}\cap \Gamma^{\df{k}},G_{\mathcal{P}})} \sem{\alpha}\). 

	Lemma~\ref{lemma:sem-equiv} shows that for all control word \(\gamma\in\prod^*\) such that $Q \xRightarrow[\textbf{df}]{\gamma} \alpha$
	we have that $\sem{\gamma} = \set{I \cdot O \mid \tuple{I\proj_{\vec{x}_i}, O\proj_{\vec{x}_i}} \in \sem{\alpha}}$.
  This enables the use of Lemma~\ref{query-fsa} showing that such control word \(\gamma\) is such that $\sem{\gamma} = \sem{\mathit{SC}^k_Q(\gamma)}$.
	This is saying the semantics of \(\gamma\) in \(\mathcal{P}\) can be obtained by computing that of \(\mathit{SC}^k_Q(\gamma)\) in \(query^k\).

%	Let \(\alpha\in \hat{L}_Q(\pat_{\Gamma}\cap \Gamma^{\df{k}},G_{\mathcal{P}})\) and let \(\gamma\) be
%	the control word of the derivation \(D_{\gamma}\) thereof which is unique
%	by Lemma~\ref{lemma:unique-df-control-word}. 
%	Proposition~\ref{query-fsa}	shows that \(\mathrm{SC}_Q^{k}\) maps \(\gamma\) onto a feasible interprocedurally valid path \(\beta\) of \(\mathit{query}^k(Q)\).
	%\todo{PG: previously it stated ``, that is \(\beta\in L_{\mathit{query}^k}(G_{\mathcal{H}})\)''} 

	We then conclude from Lem.~\ref{lem:gsmbounded} and Thm.~\ref{thm:gsmregular}, that \(\mathrm{SC}_Q^{k}(\pat_\Gamma)\) is a bounded and regular language. 
	Back to \(\sem{\mathcal{H}}_{\mathit{query}^k}\), we find that
	\[
	\sem{\mathcal{H}}_{\mathit{query}^k}=\textstyle{\bigcup_{\alpha\in L_{\mathit{query}^k}(G_{\mathcal{H}})}\sem{\alpha}}= \textstyle{\bigcup_{\alpha\in L_{\mathit{query}^k}(G_{\mathcal{H}})\cap \mathrm{SC}_Q^{k}(\pat_\Gamma)} \sem{\alpha}}
	\]
	and that \(\sem{\mathcal{H}}_{\mathit{query}^k}\) is flattable since \(\mathrm{SC}_Q^{k}(\pat_\Gamma)\) is a bounded regular set.
\qed
\end{proof}
 % end input /Users/pierreganty/counter-recursive/completness.tex
 %
% start input /Users/pierreganty/counter-recursive/experiments.tex
\section{Experiments}\label{sec:experiments}

%% %\vspace*{-11mm}
%% \begin{table}
%% \centering
%% \begin{tabular}{|l|c|c|}
%% \hline
%% Program & Time [s] & $k$ \\
%% \hline
%% %\texttt{identity} & 0.5 & 2 \\
%% \texttt{timesTwo} & 0.7 & 2 \\
%% \texttt{leq} & 0.8 & 2 \\
%% \texttt{parity} & 0.8 & 2 \\
%% \texttt{plus} & 1.5 & 2 \\
%% $F_{a=2}$ & 1.5 & 3 \\
%% %$F_{a=5}$ & 12.9 & 4 \\
%% $F_{a=8}$ & 36.9 & 4 \\
%% $G_{b=12}$ & 3.5 & 3 \\
%% $G_{b=13}$ & 23.2 & 3 \\
%% $G_{b=14}$ & 23.4 & 3 \\
%% \hline
%% \end{tabular}
%% \caption{Experiments}
%% \label{table:experiments}
%% \end{table}

\begin{table}
\centering
{\footnotesize\begin{tabular}{|@{\,}l@{\,}|ccc|ccc|ccc@{\,}|}
\hline
 & \multicolumn{3}{c|}{$k=2$} & \multicolumn{3}{c|}{$k=3$} & \multicolumn{3}{c|}{$k=4$} \\
 \cline{2-10}
& \# & t & fp & \# & t & fp & \# & t & fp \\
\hline
\texttt{identity} & 210 & 0.10 & no & 330 & 0.22 & yes & \multicolumn{3}{c|}{-} \\
\texttt{leq} & 152 & 0.12 & no & 240 & 0.27 & no & 328 & 0.41 & yes \\
\texttt{parity} & 384 & 0.14 & no & 606 & 0.54 & no & 828 & 1.31 & yes \\
\texttt{plus} & 462 & 0.53 & no & 728 & 2.54 & no & 994 & 9.20 & yes \\
\texttt{times2} & 210 & 0.14 & no & 330 & 0.35 & yes & \multicolumn{3}{c|}{-} \\
%\texttt{mccarthy} & 72 & 0.06 & no & 114 & 0.74 & yes & \multicolumn{3}{c|}{-}  \\
\hline
\end{tabular}}
\caption{Experiments with recursive implementations of basic arithmetic functions and predicates \cite{termination-competition}.}
\label{tab:exp1}
\end{table} 

We have implemented the proposed method in the \textsc{Flata} verifier
\cite{HKGIKR12} and experimented with several benchmarks. The
\textsc{Flata} tool is publicly
available\footnote{\url{https://github.com/filipkonecny/flata}} and
the benchmarks used in this section are given in the
repository. First, we have considered several programs from external
sources \cite{termination-competition}, that compute arithmetic
functions or predicates in a recursive way such as \texttt{identity}
(identity), \texttt{plus} (addition), \texttt{times2} (multiplication
by two), \texttt{leq} (comparison), and \texttt{parity} (parity
checking). It is worth noting that all of these programs have bounded
index visibly pushdown grammars, i.e.\ \(L(G^P)\) is of bounded index,
for each program \(P \in \set{\texttt{identity}, \texttt{plus},
  \texttt{times2}, \texttt{leq}, \texttt{parity}}\), the 
stabilization of the under-approximation sequence is thus
guaranteed. For all our benchmarks, the condition that the tuple of relation $\sem{\mathcal{P}}^{(k)}$ 
is inductive with respect to the statements of \(\mathcal{P}\) is met for
$k\leq3$. Table \ref{tab:exp1} shows the results, giving the size (\#) of each
under-approximation $query^k$ (the number of transitions) and the time (t)
needed to compute its summary (in seconds). The
column fp indicates whether the fixpoint check was successful. The
platform used for all experiments is MacBookPro with Intel Core i7
\(2,3\,\text{GHz}\) with \(16\,\text{GB}\) of RAM.

\begin{table}
\centering
{\footnotesize\begin{tabular}{|l|ccc|ccc|ccc|}
\hline
 & \multicolumn{3}{c|}{$k=2$} & \multicolumn{3}{c|}{$k=3$} & \multicolumn{3}{c|}{$k=4$} \\
\hline
& \# & t & fp & \# & t & fp & \# & t & fp \\
\hline
$F_1$ & 32 & 0.05 & no & 50 & 0.07 & no & 68 & 0.09 & yes \\
$F_2$ & 72 & 0.06 & no & 114 & 0.74 & no & 156 & 1.55 & yes \\
$F_3$ & 128 & 0.06 & no & 204 & 0.30 & no & 280 & 1.59 & yes \\
$F_4$ & 200 & 0.06 & no & 320 & 0.44 & no & 440 & 4.02 & yes \\
$F_5$ & 288 & 0.07 & no & 462 & 0.63 & no & 636 & 5.97 & yes \\
$F_6$ & 392 & 0.07 & no & 630 & 0.82 & no & 868 & 7.54 & yes \\
$F_7$ & 512 & 0.08 & no & 824 & 0.86 & no & 1136 & 14.23 & yes \\
$F_8$ & 648 & 0.08 & no & 1044 & 1.09 & no & 1440 & 12.87 & yes \\
\hline
\end{tabular}}
\[
  F_a(x)=\left\{\begin{array}{ll}
    x-10 & \textrm{ if } x\geq 101 \\
    (F_a)^a(x+10\cdot a-9) & \textrm{ if } x\leq 100
  \end{array}\right.
\]
\caption{\label{tab:exp-fa}Generalized McCarthy $F_a$ Experiments. The function \(F_2\) is the original McCarthy91 function.}
\end{table} 

\begin{table}
\centering
{\footnotesize\begin{tabular}{|l|ccc|ccc|ccc|}
\hline
 & \multicolumn{3}{c|}{$k=2$} & \multicolumn{3}{c|}{$k=3$} & \multicolumn{3}{c|}{$k=4$} \\
\hline
& \# & t & fp & \# & t & fp & \# & t & fp \\
\hline
$G_{11}$ & 72 & 0.06 & no & 114 & 0.74 & no & 156 & 1.55 & yes \\
$G_{12}$ & 72 & 0.08 & no & 114 & 1.53 & no & 156 & n/a & ? \\
$G_{13}$ & 72 & 0.08 & no & 114 & 5.07 & no & 156 & n/a & ? \\
$G_{14}$ & 72 & 0.08 & no & 114 & 7.07 & no & 156 & n/a & ? \\
\hline
\end{tabular}}
\[
  G_b(x)=\left\{\begin{array}{ll}
    x-10 & \textrm{ if } x\geq 101 \\
    G(G(x+b)) & \textrm{ if } x\leq 100
  \end{array}\right.
\]    %\todo{why not $G_b(G_b(x+10+b))$ for arbitrary $b$ ?}
\caption{\label{tab:exp-gb}Generalized McCarthy $G_b$ Experiments. 
The function \(G_{11}\) is the original McCarthy91 function.}
\end{table} 

Next, we have considered two generalizations of the McCarthy 91
function \cite{cowles}, a~well-known verification benchmark that has
long been a challenge. We have automatically computed precise
summaries of its generalizations $F_a$ (Table \ref{tab:exp-fa}) and
$G_b$ (Table \ref{tab:exp-gb}) above for $a=2,\ldots,8$ and
$b=12,13,14$. For the $F_a$ functions, the computed summaries are
given by: \[F_a(x) = \left\{\begin{array}{ll}
91 & \mbox{if $x \leq 100$} \\
x-10 & \mbox{if $x \geq 101$}
\end{array}\right. \text{ for all $a=1,\ldots,8$ .} \]
The computed summaries for the $G_b$ functions are given in Table
\ref{tab:summaries}.

\begin{table}
\centering
\begin{tabular}{|l|c|}
\hline
\(G_{11}(x)\) & \(\begin{array}{ll} 91 & \mbox{if $x \leq 100$} \\ x-10 & \mbox{if $x \geq 101$} \end{array}\) \\
\hline
\(G_{12}(x)\) & \(\begin{array}{ll} 91 & \mbox{if $x \leq 100$ and $2 | x+1$} \\ 92 & \mbox{if $x \leq 100$ and $2 | x$} \\ x-10 & \mbox{if $x\geq 101$} \end{array}\) \\
\hline
\(G_{13}(x)\) & \(\begin{array}{ll} 91 & \mbox{if $x \leq 100$ and $3 | x+1$} \\ 92 & \mbox{if $x \leq 100$ and $3 | x$} \\ 93 & \mbox{if $x \leq 100$ and $3|x+2$} \\ x-10 & \mbox{if $x\geq 101$} \end{array}\) \\
\hline
\(G_{14}(x)\) & \(\begin{array}{ll} 91 & \mbox{if $x \leq 100$ and $4 | x+3$} \\ 92 & \mbox{if $x \leq 100$ and $4 | x+2$} \\ 93 & \mbox{if $x \leq 100$ and $4|x+1$} \\ 94 & \mbox{$x \leq 100$ and $4|x$} \\ x-10 & \mbox{if $x\geq 101$} \end{array}\) \\
\hline
\end{tabular}
\caption{\label{tab:summaries}Automatically computed summaries for the
  generalized McCarthy $G_b$ functions (for index $k=3$).}
\end{table}

The visibly pushdown grammars corresponding to the recursive programs
implementing the $F_a,G_b$ functions are not bounded. In the case of
the $F_a$ function, the under-approximation sequence reaches a
fixpoint after $4$ iterations. In the case of $G_b$, for $b=12,13,14$,
the summary of $query^{3}$ is the expected result. However, due to the
limitations of the \textsc{Flata} tool, which is based on an
acceleration procedure without abstraction, we could not compute the
summary of $query^{4}$, and we could not verify automatically that the
fixpoint has been reached.
 % end input /Users/pierreganty/counter-recursive/experiments.tex
 %
% start input /Users/pierreganty/counter-recursive/conclusion.tex
\section{Conclusions}\label{sec:conclusions}

We have presented an underapproximation method for computing summaries
of recursive programs operating on integers. The underapproximation is
driven by bounding the index of derivations that produce the execution
traces of the program, and computing the summary, for each index, by
analyzing a non-recursive program. We also present a class of programs
on which our method is complete. Finally, we report on an
implementation and experimental evaluation of our technique.
 % end input /Users/pierreganty/counter-recursive/conclusion.tex
 
\paragraph{Acknowledgements.}
Pierre Ganty is supported by the EU FP7 2007--2013 program under agreement
610686 POLCA, by the Madrid Regional Government under CM project S2013/ICE-2731
(N-Greens) and RISCO: RIgorous analysis of Sophisticated COncurrent and
distributed systems, funded by the Spanish Ministry of Economy and
Competitiveness No.\@ TIN2015-71819-P (2016--2018).
Pierre thanks Thomas Reps for pointing out inconsistencies in the examples.

% This is the lncs style, it gives a compact bib.
%\begin{spacing}{0.9}
\bibliographystyle{abbrv}
\bibliography{../ref}
%\end{spacing}

%\newpage

%%%%%%%%%%%%%%%%%%%%%%%%%%%%%%%%%%%%%%%%%%%%%%%%%%%%%%%%%%%%%%%%%%%%%%%%%%%%%%%
\end{document}